\documentclass[a4paper,UKenglish,cleveref, autoref, thm-restate, pdfa]{lipics-v2021}

\usepackage{enumitem}
\usepackage[numbers,sort]{natbib}
\usepackage{apptools}
\usepackage{longtable}
\AtAppendix{\counterwithin{lemma}{section}}
\AtAppendix{\counterwithin{definition}{section}}

\title{Tuple Interpretations for Higher-Order Complexity}

\author{Cynthia Kop}
{Department of Software Science, Radboud University Nijmegen, The Netherlands \and \url{https://www.cs.ru.nl/~cynthiakop}}
{c.kop@cs.ru.nl}
{https://orcid.org/0000-0002-6337-2544}{}

\author{Deivid Vale}
{Department of Software Science, Radboud University Nijmegen, The Netherlands \and \url{https://www.cs.ru.nl/~deividvale}}
{deividvale@cs.ru.nl}
{https://orcid.org/0000-0003-1350-3478}{}

\authorrunning{C. Kop and D. Vale} 

\Copyright{Cynthia Kop and Deivid Vale} 

\ccsdesc{Theory of computation~Equational logic and rewriting}

\keywords{Complexity, higher-order term rewriting, many-sorted term rewriting, polynomial interpretations, weakly monotonic algebras}

\funding{The authors are supported by the NWO TOP project ``ICHOR'', NWO 612.001.803/7571 and the NWO VIDI project ``CHORPE'', NWO VI.Vidi.193.075.}

\acknowledgements{}

\nolinenumbers 

\hideLIPIcs

\usepackage{amsmath,amssymb,amsfonts,amscd,amsbsy,mathtools}
\usepackage{microtype}

\usepackage{quiver}

\usepackage{xcolor-material}
\usepackage{stmaryrd}
\SetSymbolFont{stmry}{bold}{U}{stmry}{m}{n}

\newcommand{\signature}{\mathcal{F}}             
\newcommand{\F}{\mathcal{F}}                     
\newcommand{\var}{\mathcal{X}}
\newcommand{\terms}{T(\signature,\var)}
\newcommand{\termsfo}{T_{fo}(\signature,\var)}

\newcommand{\fvars}[1]{\mathtt{fv}(#1)}    

\newcommand{\dom}[1]{\mathtt{dom}(#1)}      
\newcommand{\sortset}{\mathcal{S}}
\newcommand{\simpletypeset}{\mathcal{S}\!\mathcal{T}}
\newcommand{\rules}{\mathcal{R}}

\newcommand{\Nat}{\mathbb{N}}

\newcommand{\trs}{{(\signature,\rules)}}

\newcommand{\sAr}[1]{\mathsf{#1}} 
\newcommand{\sTp}[1]{\mathsf{#1}} 

\newcommand{\pair}[1]{\langle #1 \rangle}

\newcommand{\interpret}[1]{\llbracket #1\rrbracket}

\newcommand{\arrz}{\to_\rules}

\newcommand{\dht}[1]{\mathtt{dh}_\rules(#1)}

\newcommand{\arrtype}{\Rightarrow}
\newcommand{\arrfunc}{\Longrightarrow}

\newcommand{\nat}{\sTp{nat}}
\newcommand{\lst}{\sTp{list}}
\newcommand{\real}{\sTp{real}}

\newcommand{\zero}{\sAr{0}}
\newcommand{\nil}{\sAr{nil}}
\newcommand{\suc}{\sAr{s}}
\newcommand{\cons}{\sAr{cons}}
\newcommand{\minus}{\sAr{minus}}
\newcommand{\quot}{\sAr{quot}}

\newcommand{\add}{\mathrel{\sAr{\oplus}}}

\newcommand{\append}{\sAr{append}}
\newcommand{\rev}{\sAr{rev}}
\newcommand{\sumList}{\sAr{sum}}

\newcommand{\map}{\sAr{map}}
\newcommand{\foldl}{\sAr{foldl}}

\newcommand{\listvar}{q}

\newcommand{\sortinterpret}[1]{A_{#1}}
\newcommand{\typeinterpret}[1]{\mathcal{M}_{#1}}
\newcommand{\wmatypes}{\overrightarrow{\typeinterpret{\atype}}}
\newcommand{\weakset}{\text{wm}}
\newcommand{\SM}{\mathcal{M}}
\newcommand{\SMA}{\mathcal{A}_{\SM}}
\newcommand{\ainterpret}[1]{\llbracket#1\rrbracket_\alpha}
\newcommand{\funcinterpret}[1]{\mathcal{J}_{#1}}
\newcommand{\varinterpret}{\alpha}

\newcommand{\sortgr}[1]{>_{#1}}
\newcommand{\sortgeq}[1]{\geq_{#1}}
\newcommand{\typegr}[1]{\sqsupset_{#1}}
\newcommand{\typegeq}[1]{\sqsupseteq_{#1}}
\newcommand{\makesm}[1]{\mathit{MakeSM}_{#1}}

\newcommand{\nul}[1]{\mathtt{0}_{#1}}
\newcommand{\costof}[1]{\mathtt{costof}_{#1}}
\newcommand{\addcost}[1]{\mathtt{addc}_{#1}}

\newcommand{\abs}[2]{\lambda #1.#2}
\newcommand{\app}[2]{#1 \cdot #2}

\newcommand{\asort}{\iota}
\newcommand{\bsort}{\kappa}
\newcommand{\atype}{\sigma}
\newcommand{\btype}{\tau}
\newcommand{\ctype}{\rho}
\newcommand{\afun}{\mathsf{f}}

\newcommand{\avar}{x}
\newcommand{\bvar}{y}
\newcommand{\cvar}{z}
\newcommand{\aListVar}{xs}
\newcommand{\bListVar}{ys}

\newcommand{\aFuncVar}{F}
\newcommand{\bFuncVar}{G}
\newcommand{\cFuncVar}{H}
\newcommand{\aterm}{s}
\newcommand{\bterm}{t}
\newcommand{\cterm}{u}

\newcommand{\typecount}[1]{K[#1]}

\newcommand{\cost}{{\sAr{c}}}
\newcommand{\size}{{\sAr{s}}}
\newcommand{\leng}{{\sAr{l}}}
\newcommand{\mmax}{{\sAr{m}}}

\newcommand{\N}{\mathbb{N}}

\newcommand{\secshort}{\S}
\newcommand{\seclong}{Section~}

\begin{document}

\maketitle

\begin{abstract}
    We develop a class of algebraic interpretations for many-sorted and higher-order term rewriting systems
    that takes type information into account.
    Specifically, base-type terms are mapped to \emph{tuples} of natural numbers
    and higher-order terms to functions between those tuples.
    Tuples may carry information relevant to the type;
    for instance, a term of type $\nat$ may be associated to a pair
    $\pair{\mathsf{cost}, \mathsf{size}}$ representing its evaluation cost and size.
    This class of interpretations results in a more fine-grained notion of complexity than runtime or derivational complexity,
    which makes it particularly useful to obtain complexity bounds for higher-order rewriting systems.

    We show that rewriting systems compatible with tuple interpretations admit finite bounds on derivation height.
    Furthermore, we demonstrate how to mechanically construct tuple interpretations and how to orient $\beta$ and $\eta$ reductions within our technique.
    Finally, we relate our method to runtime complexity and prove that specific interpretation shapes imply certain runtime complexity bounds.
\end{abstract}

\section{Introduction}

Term rewriting systems (TRSs) are a conceptually simple but powerful computational model.
It is simple because computation is modelled straightforwardly by step-by-step applications of transformation rules.
It is powerful in the sense that any algorithm can be expressed in it (Turing Completeness).
These characteristics make TRSs a formalism well-suited as an abstract analysis language,
for instance to study properties of functional programs.
We can then define specific analysis techniques for each property of interest.

One such property is \emph{complexity}.  The study of complexity has long been a topic of interest
in term rewriting
\cite{bon:cic:mar:tou:98,hof:lau:89,hof:92,ava:mos:08,nao:moser:08,mos:sch:wal:08},
as it both holds relations to computational complexity \cite{arai:moser:05,bon:cic:mar:tou:98,bonfante:et-al:01}
and resource analysis \cite{ava:lag:mos:15,bro:emm:fal:fuh:gie:14}
and is highly challenging.
Most commonly studied are the notions of runtime and derivational complexity,
which capture the number of steps that may be taken when starting with terms of a given size and shape.
In essence, this is a form of resource analysis which abstracts away from the true machine cost of reduction in a rewriting engine but still has a close relation to it
\cite{avanzini:moser:10,dal-lago:martini:10,accattoli:dal-lago:16,bonfante:et-al:01}.

These notions do not obviously extend to the \emph{higher-order} setting, however.
In higher-order term rewriting, 
a term may represent a function;
yet, the size of a function does not 
tell us much
about its behaviour.
Rather, properties such as ``the function is size-increasing'' may be more relevant.
Clearly a more sophisticated complexity notion is needed.

In this paper we will propose a new method to analyse many-sorted and higher-order term rewriting
systems, which can be used as a foundation to obtain a variety of complexity results.
This method is based on \emph{interpretations} in a monotonic algebra as also used for
termination analysis \cite{pol:96,fuh:kop:12}, where a term of function type is mapped to a
monotonic function.
Unlike \cite{pol:96,fuh:kop:12}, we map a term of base type not to an integer,
but rather to a vector of integers describing different values of interest in the term.
This will allow us to reason separately about---for instance---the length of a list and the size
of its greatest element, and to describe the behaviour of a term of function type in a fine-grained
way.

This method is also relevant for termination analysis, since we essentially generalise and extend
\emph{matrix interpretations} \cite{mos:sch:wal:08} to higher-order rewriting.
In addition, the technique may add some power to the arsenal
of a complexity or termination analysis tool for first-order term rewriting; in particular
\emph{many-sorted} term rewriting due to the way we use type information.

\subparagraph*{A note on terminology.} We use the word ``complexity'' as it is commonly used in
term rewriting: a worst-case measure of the number of steps in a reduction.  In this paper we do
not address the question of true resource use or connections to computational complexity.  In
particular, we do not address the true cost of beta-reduction.  This is left to future work.

\subparagraph*{Outline of the paper.}
We will start by recalling the definition of and fixing notation for many-sorted and higher-order
term rewriting (\secshort\ref{sec:preliminaries}).
Then, we will define tuple interpretations for many-sorted first-order rewriting to explore the
idea (\secshort\ref{sec:fo-tp-int}), discuss our primary objective of \emph{higher-order} tuple
interpretations (\secshort\ref{sec:ho-tp-int}), and relate our method to runtime complexity
(\secshort\ref{sec:bounds}).  Finally, we will discuss related work (\secshort\ref{sec:related})
and end with conclusions and future work (\secshort\ref{sec:disc}).

\section{Preliminaries}\label{sec:preliminaries}

We assume the reader is familiar with first-order term rewriting and $\lambda$-calculus.
In this section, we fix notation and discuss the higher-order rewriting format used in the paper.

\subsection{First-Order Many-Sorted Rewriting}\label{sec:prelim:fo}
Many-sorted term rewriting \cite{ohl:02} is in principle the same as first-order term rewriting.
The only difference is that we impose a sort system and limit interest to well-sorted terms.

Formally,
we assume given a non-empty set of \emph{sorts} $\sortset$.
A \emph{many-sorted signature} consists of a set $\signature$ of function symbols together with
two functions that map each symbol to a finite sequence of \emph{input sorts} and an \emph{output sort}.
Fixing a many-sorted signature, we will denote $\afun :: [\asort_1 \times \dots \times \asort_k] \arrtype \bsort$
if $\afun \in \signature$ and $\afun$ has input sorts $\asort_1, \dots, \asort_k$ and output sort $\kappa$.
We also assume given a set $\var = \bigcup_{\asort \in \sortset} \var_\asort$ of variables disjoint from $\signature$,
such that all $\var_\asort$ are pairwise disjoint.
The set $\termsfo$ of \emph{many-sorted terms} is inductively defined as the
set of expressions $\aterm$
such that $\aterm :: \bsort$ can be derived for some sort $\bsort$ using the clauses:
\[
    \begin{array}{lcl}
        \avar :: \kappa\ \text{if}\ \avar \in \var_\kappa & \quad\quad &
        \afun(\aterm_1,\dots,\aterm_k) :: \bsort\ \text{if}\ \afun :: [\asort_1 \times \dots \times \asort_k] \arrtype \bsort\ \text{and each}\ \aterm_i :: \asort_i \\
    \end{array}
\]
If $\aterm :: \bsort$, we call $\bsort$ the sort of $\aterm$.
Substitutions, rewrite rules and 
reduction
are defined as usual in first-order term rewriting,
except that substitutions are sort-preserving (each variable is mapped to a term of the
same sort) and both sides of a rule have the same sort.
We omit these definitions, since they are
a special case of the 
higher-order definitions
in \seclong\ref{sec:prelim:ho}.

\begin{example}\label{ex:append-reverse}
    We fix $\nat$ and $\lst$ for the sorts of natural numbers and lists of natural numbers,
    respectively; and a signature with the symbols:
    $\zero :: \nat$ (this is shorthand notation for $[] \arrtype \nat$),
    $\suc :: [\nat] \arrtype \nat$,
    $\nil :: \lst$,
    $\cons :: [\nat \times \lst] \arrtype \lst$,
    $\rev :: [\lst] \arrtype \lst$,
    $\sumList :: [\lst] \arrtype \nat$,
    $\append :: [\lst \times \lst] \arrtype \lst$,
    and $\add :: [\nat \times \nat] \arrtype \nat$.
    The rules below compute well-known functions over lists and numbers.
    We follow the convention of using infix notation for $\cons$ and $\add$, i.e., $\cons(\avar,
        \aListVar)$ is written $\avar : \aListVar$ and $\sAr{\oplus}(\avar, \bvar)$ is written
    $\avar \add \bvar$.
    \begin{align*}
        \avar \add \zero                      & \to \avar                                 & \sumList(\nil)              & \to \zero                                  \\
        \avar \add \suc(\bvar)                & \to \suc(\avar \add \bvar)                & \sumList(\avar : \aListVar) & \to \sumList(\aListVar) \add \avar         \\
        \append(\nil, \aListVar)              & \to \aListVar                             & \rev(\nil)                  & \to \nil                                   \\
        \append(\avar : \aListVar, \bListVar) & \to \avar : \append(\aListVar, \bListVar) & \rev( x : \aListVar)        & \to \append(\rev(\aListVar), \avar : \nil)
    \end{align*}
\end{example}

\subsection{Higher-Order Rewriting}\label{sec:prelim:ho}

For higher-order rewriting, we will use \textit{algebraic functional systems} (AFS), a slightly simplified form of a
higher-order language introduced by Jouannaud and Okada~\cite{jou:oka:91}.
This choice gives an easy presentation, as it combines algebraic definitions in a first-order style with a function mechanism using $\lambda$-abstractions and term applications.

Given a non-empty set of \emph{sorts} $\sortset$, the set $\simpletypeset$ of \emph{simple types} (or just \emph{types}) is given by: (a) $\sortset \subseteq \simpletypeset$; (b) if $\sigma,\tau \in \simpletypeset$ then $\sigma \arrtype \tau \in \simpletypeset$.
Types are denoted by $\sigma, \tau$ and sorts by $\iota, \kappa$.
A \emph{higher-order signature} consists of a set $\F$ of function symbols together with two functions
that map each symbol to a finite sequence of \emph{input types} and an \emph{output type};
fixing a signature,
we denote this type information $\sAr{f} :: [\sigma_1 \times \dots \times \sigma_k] \arrtype \tau$.
A function symbol is said to be higher-order if at least one of its input types or its output type is an
arrow type.

We also assume given a set $\var = \bigcup_{\atype \in \simpletypeset} \var_\atype$ of variables disjoint from $\F$ (and pairwise disjoint) so that each $\var_\atype$ is countably infinite.
The set $\terms$ of terms is inductively defined as the set of expressions whose type can be derived using the following clauses:
\[
    \begin{array}{llcll}
        \avar :: \sigma                                & \text{if}\ x \in \var_\sigma                                                  & \quad\quad &
        (\abs{\avar}{\aterm}):: \atype \arrtype \btype & \text{if}\ x \in \var_\atype\ \text{and}\ s :: \btype                                                                             \\
        (\aterm \, \bterm) :: \tau                     & \text{if}\ s :: \sigma \arrtype \tau\ \text{and}\  :: \sigma                  & \quad\quad &
        \sAr{f}(s_1, \dots, s_k) :: \btype             & \text{if}\ \sAr{f} :: [\atype_1 \times \dots \times \atype_k] \arrtype \btype                                                     \\
                                                       &                                                                               &            &   & \text{and each}\ s_i :: \atype_i \\
    \end{array}
\]
If $s :: \sigma$, we say that $\sigma$ is the type of $s$.
It is easy to see that each term has a unique type.

As in the $\lambda$-calculus, a variable $\avar$ is \emph{bound} in a term if it occurs in the scope of an abstractor $\abs{x}{}$;
it is \emph{free} otherwise.
A term is called \emph{closed} if it has no free variables and \emph{ground} if it also has no bound variables.
Term equality is modulo $\alpha$-conversion and bound variables are renamed if necessary.
Application is left-associative and has precedence over abstractions;
for example, $\lambda x.s \, t \, u$ reads $\lambda x.((s \, t) \, u)$.
A substitution is a finite, type-preserving mapping $\gamma : \var \to \terms$, typically denoted
$[x_1:=s_1,\dots,x_n:=t_n]$.  Its \emph{domain} $\{x_1,\dots,x_n\}$ is denoted $\dom{\gamma}$.
A substitution $\gamma$ is applied to a term $s$, notation $s\gamma$, by renaming all bound
variables in $s$ to fresh variables and then replacing each $x \in \dom{\gamma}$ by $\gamma(x)$.
Formally:
\[
    \begin{array}{rcllcrcll}
        \avar\gamma                & = & \gamma(\avar)                    & \text{if}\ \avar \in \dom{\gamma}    & \quad &
        (s\, t)\gamma              & = & (s\gamma)\, (t\gamma)                                                             \\
        \avar\gamma                & = & \avar                            & \text{if}\ \avar \notin \dom{\gamma} &       &
        \afun(s_1,\dots,s_k)\gamma & = & \afun(s_1\gamma,\dots,s_k\gamma)                                                  \\
                                   &   &                                  &                                      &       &
        (\abs{x}{s})\gamma         & = & \abs{y}{(s([x:=y]\gamma))}       & \text{for}\ y\ \text{fresh}                    \\
    \end{array}
\]
Here, $[x:=y]\gamma$ is the substitution that maps $x$ to $y$ and all variables in $\dom{\gamma}$
other than $x$ to $\gamma(x)$.  The result of $s\gamma$ is unique modulo $\alpha$-renaming.

A \textit{rewriting rule} is a pair of terms $\ell \to r$ of the same type such that all free variables of $r$ also occur in $\ell$.
Given a set of rewriting rules $\rules$, the rewrite relation induced by $\rules$ on the set $\terms$ is the smallest monotonic relation
that is stable under substitution and contains both all elements of $\rules$ and $\beta$-reduction.
That is, it is inductively generated by:
\[
    \begin{array}{rcllcrcll}
        (\abs{x}{s})\ t      & \arrz & s[x:=t]              &                                  & \quad &
        \abs{x}{s}           & \arrz & \abs{x}{t}           & \text{if}\ s \arrz t                       \\
        \ell\gamma           & \arrz & r\gamma              & \text{if}\ \ell \to r \in \rules &       &
        s\ u                 & \arrz & t\ u                 & \text{if}\ s \arrz t                       \\
        \afun(\dots,s,\dots) & \arrz & \afun(\dots,t,\dots) & \text{if}\ s \arrz t             &       &
        u\ s                 & \arrz & u\ t                 & \text{if}\ s \arrz t                       \\
    \end{array}
\]
Note that we do not, by default, include the common $\eta$-reduction rule scheme (``$\abs{x}{s\ x}
\arrz s$ if $x$ is not a free variable in $s$'').  We avoid this because not all sources consider it,
and it is easy to add by including, for all types $\sigma,\tau$, a rule $\abs{x}{F\ x} \to F$ with
$F \in \var_{\sigma \arrtype \tau}$ in $\rules$.

An \emph{algebraic functional system} (AFS) is the combination of a set of terms $\terms$ and a rewrite relation
$\arrz$ over $\terms$.  An AFS is typically given by supplying
$\signature$ and $\rules$.

A \emph{many-sorted term rewriting system} (TRS), as discussed in \seclong\ref{sec:prelim:fo}, is a pair
$(\termsfo,\arrz)$ where $\F$ is a many-sorted signature and $\arrz$ a rewrite relation over $\termsfo$.
That is, it is essentially an AFS where we only consider first-order terms.

\begin{example}\label{ex:ho_trs}
    Following common examples in higher-order rewriting, we will use (as a running example) the AFS
    $\trs_{\mathtt{fold}}$, with symbols
    $\nil :: \lst$,
    $\cons :: [\nat \times \lst] \arrtype \lst$,
    $\map :: [ (\nat \arrtype \nat) \times \lst] \arrtype \lst$,
    $\foldl :: [(\nat \arrtype \nat \arrtype \nat) \times \nat \times \lst] \arrtype \nat$,
    and rules:
    \begin{align*}
        \foldl(\aFuncVar, \cvar, \nil)              & \to \cvar                                                     & \map(\aFuncVar,\nil)               & \to \nil                                              \\
        \foldl(\aFuncVar, \cvar, \avar : \aListVar) & \to \foldl(\aFuncVar, (\aFuncVar \ \cvar \ \avar), \aListVar) & \map(\aFuncVar, \avar : \aListVar) & \to (\aFuncVar \, \avar) : \map(\aFuncVar, \aListVar)
    \end{align*}
\end{example}

\subsection{Functions and orderings}

An extended well-founded set is a tuple $(A,>,\geq)$ such that
$>$ is a well-founded ordering on $A$;
$\geq$ is a quasi-ordering on $A$;
$x > y$ implies $x \geq y$; and
$x > y \geq z$ implies $x > z$.
Hence, it is permitted, but not required, that $\geq$ is the reflexive closure of $>$.

For sets $A,B$, the notation $A \arrfunc B$ denotes the set of functions from $A$ to $B$.
Function equality is extensional: for $f,g \in A \arrfunc B$ we say $f = g$ iff $f(x) =
    g(x)$ for all $x \in A$.

If $(A,>,\geq)$ and $(B,\succ,\succeq)$ are extended well-founded sets,
we say that $f \in A \arrfunc B$ is \emph{weakly monotonic} if $x \geq y$ implies $f(x) \succeq f(y)$.
In addition, if $(A_1,>_1,\geq_1),\dots,(A_n,>_n,\geq_n)$ are all well-founded sets, we say that
$f \in A_1 \times \dots \times A_n \arrfunc B$ is weakly monotonic if we have
$f(x_1,\dots,x_n) \succeq f(y_1,\dots,y_n)$ whenever $x_i \geq_i y_i$ for all $1 \leq i \leq n$.
We say that $f$ is \emph{strict} in argument $j$ if $x_j >_j y_j$ (and also $x_i \geq_i y_i$ for
all $i$) implies $f(x_1,\dots,x_n) \succ f(y_1,\dots,y_n)$.

We say that $f \in A_1 \times \dots \times A_n \arrfunc B$ is \emph{strongly monotonic} if $f$ is
weakly monotonic and strict in all its arguments (and similar for $f \in A \arrfunc B$).

\section{First-Order tuple interpretation}\label{sec:fo-tp-int}

In this section, we will introduce the concept of tuple interpretations for many-sorted term rewriting.
This is the core methodology which the higher-order theory is built on top of.
This theory also has value by itself as a first-order termination and complexity technique.

\smallskip
It is common in the rewriting literature to use termination proofs to assess the \textit{difficulty} of rewriting a term to normal form~\cite{ava:mos:08, hof:lau:89}.
The intuition comes from the idea that by ordering rewriting rules in descending order we gauge the order of magnitude of reduction.
The same principle applies for syntactic~\cite{nao:moser:08,hof:92,moser:06} and semantic~\cite{hof:lau:89,hof:01,mos:sch:wal:08} termination proofs.

On the semantic side there is a natural strategy:
given an extended well-founded set $\mathcal{A} = (A, >, \geq)$ find an interpretation from terms
to elements of $A$ so that $\interpret{s} > \interpret{t}$ whenever $s \arrz t$.  (This can
typically be done by showing that $\interpret{\ell} > \interpret{r}$ for all rules $\ell \to r$).
This interpretation holds information about the complexity of $\trs$ since the maximum length of a
reduction starting in a term $s$ is bounded by number of $>$ steps that may be done starting in
$\interpret{s}$.  If $\interpret{s}$ is a natural number, this gives a bound immediately.

In the setting of many-sorted term rewriting, we may formally define this as follows.

\begin{definition}\label{def:mult-sorted-alg}
    Let $\sortset$ be a set of sorts and $\signature$ an $\sortset$-signature.
    A many-sorted monotonic algebra $\mathcal{A}$ consists of a family of extended well-founded sets
    $(A_\asort, \sortgr{\asort},\sortgeq{\asort})_{\asort \in \sortset}$ together with an interpretation $\funcinterpret{}$
    which associates to each $\sAr{\afun :: [\asort_1 \times \dots \times \asort_k] \arrtype \bsort}$ in
    $\signature$ a strongly monotonic function $\funcinterpret{\afun} \in A_{\asort_1} \times \dots
        \times A_{\asort_k} \arrfunc A_{\bsort}$.
    Let $\alpha$ be a function that maps variables of sort $\asort$ to elements of $A_\asort$.
    We extend $\funcinterpret{}$ to a function $\ainterpret{\cdot}$ that maps terms of sort $\asort$ to elements of $A_\asort$,
    by letting
    $\ainterpret{x} = \alpha(x)$ if $x$ is a variable of sort $\asort$, and
    $\ainterpret{\afun(\aterm_1, \dots, \aterm_k)} = \funcinterpret{\afun}(\ainterpret{\aterm_1}, \dots, \ainterpret{\aterm_k})$.
    We say that a TRS $\trs$ is \emph{compatible} with $\mathcal{A}$ if $\ainterpret{\ell} > \ainterpret{r}$
    for all $\varinterpret$ and all $\ell \to r \in \rules$.
\end{definition}

We will generally omit the subscript $\alpha$ when it is clear from
context, 
writing
$\interpret{s}$ instead of $\ainterpret{s}$.
In examples, we may write something like $\interpret{s} = x + y$ to mean $\ainterpret{s} =
    \varinterpret(x) + \varinterpret(y)$.

\begin{theorem}\label{thm:sortinterpretworks}
    If $\trs$ is compatible with $\mathcal{A}$ then for all $\varinterpret$:
    $\ainterpret{\aterm} > \ainterpret{\bterm}$ whenever $\aterm \arrz \bterm$.
\end{theorem}
\begin{proof}[Proof Sketch]
    By induction on the size of $s$ using strong monotonicity of each $\funcinterpret{\afun}$.
\end{proof}

A common notion in the literature on complexity of term rewriting is \emph{derivation height}:
\[
    \dht{t} := \max \{ n \in \Nat \mid \exists s. \, t \to^n s \}.
\]
Intuitively, $\dht{t}$ describes the worst-case number of steps for all possible reductions starting in $t$.
If $\trs$ is terminating, then $\dht{\cdot}$ is a total function.
If $(A_\asort,>_\asort) = (\N,>)$ then we easily see that $\dht{t} \leq \interpret{t}$ for any term $t : \asort$.
Hence, $\interpret{\cdot}{}$ can be used to bound the derivation height function.
However, this 
may give a severe overestimation,
as demonstrated below.

\begin{example}\label{ex:poly_overstimation}
    Let $\trs_{\sAr{ab}}$ be the TRS with only a rule $\sAr{a}(\sAr{b}(x)) \to \sAr{b}(\sAr{a}(x))$
    and signature $\sAr{a},\sAr{b} : [\sAr{string}] \arrtype \sAr{string}$ and $\sAr{\epsilon}: \sAr{string}$.
    We can prove termination by the following interpretation:
    \begin{align*}
        \interpret{\sAr{a}(x)} & = 2*x & \interpret{\sAr{b}(x)} & = x + 1 & \interpret{\sAr{\epsilon}} & = 0
    \end{align*}
    Indeed, we have $\interpret{\ell} > \interpret{r}$ for the only rule as
    $\interpret{\sAr{a}(\sAr{b}(x))} = 2*x+2 > 2*x+1 = \interpret{\sAr{b}(\sAr{a}(x))}$.
    Now consider a term $t = \sAr{a}^n(\sAr{b}^m(\sAr{\epsilon}))$.
    Then $\dht{t} = n*m$ whereas $\interpret{t}{} = 2^nm$; an exponential difference!
    Such an overestimation is problematic if we want to use $\interpret{\cdot}{}$ to bound $\dht{\cdot}$.
\end{example}

We could find a tight bound for the system of Example \ref{ex:poly_overstimation} by a reasoning like the
following: for every term $s$, let $\mathit{\#bs}(s)$ be the number of $\sAr{b}$ occurrences in $s$.  For a
term $t$, let $\mathit{cost}(t)$ denote $\sum \{\!\{ \mathit{\#bs}(s) \mid \sAr{a}(s)$ is a subterm of $t\}\!\}$.
Then, the cost of a term decreases exactly by $1$ in each step.
As the normal form has cost $0$, we find the tight bound $\mathit{cost}(\sAr{a}^n(\sAr{b}^m(\epsilon))) = n*m$.

This reasoning relies on tracking more than one value.
We can formalise this reasoning using an algebra
interpretation (and will do so in Example \ref{ex:ab-int}), by choosing the right $\mathcal{A}$:

\begin{definition}\label{def:tuple_algebra}
    A \emph{tuple algebra} is an algebra $\mathcal{A} = (A,\funcinterpret{})$ with $A = (A_\asort,
        \sortgr{\asort},\sortgeq{\asort})_{\asort \in \sortset}$ such that each $A_\asort$ has the form
    $\N^{\typecount{\asort}}$ (for an integer $\typecount{\asort} \geq 1$) and
    we let $\pair{n_1,\dots,n_{\typecount{\asort}}} \sortgeq{\asort} \pair{n_1',\dots,
            n_{\typecount{\asort}}'}$ if each $n_i \geq n_i'$, and $\pair{n_1,\dots,n_{\typecount{\asort}}}
        \sortgr{\asort} \pair{n_1',\dots,n_{\typecount{\asort}}'}$ if additionally $n_1 > n_1'$.
\end{definition}

Intuitively, the first component always indicates ``cost'': the number of steps needed to reduce a term
to normal form.  This is the component that needs to decrease in each rewrite step to have $\interpret{s} >
    \interpret{t}$ whenever $s \arrz t$.  The remaining components represent some value of interest for the
sort.  This could for example be the size of the term (or its normal form), the length of a list,
or following Example \ref{ex:poly_overstimation}, the number of occurrences of a specific symbol.
For these components, we only require that they do not increase in a reduction step.

By the definition of $\sortgr{\asort}$, and using Theorem \ref{thm:sortinterpretworks}, we can
conclude:

\begin{corollary}
    If a TRS $\trs$ is compatible with a tuple algebra
    then it is terminating and $\dht{t} \leq \interpret{t}_1$, for all terms $t$.
    (Here, $\interpret{t}_1$ indicates the first component of the tuple $\interpret{t}$.)
\end{corollary}

Using this, we obtain a tight bound on the derivation height of $\sAr{a}^n(\sAr{b}^m(\sAr{\epsilon}))$
in Example \ref{ex:poly_overstimation}:

\begin{example}\label{ex:ab-int}
    The TRS $\trs_{\sAr{ab}}$ is compatible with the tuple algebra with $A_{\sAr{string}} = \N^2$ and
    \begin{align*}
        \interpret{\sAr{a}(x)}     & = \pair{x_1 + x_2, x_2} &
        \interpret{\sAr{b}(x)}     & = \pair{x_1,x_2 + 1}    &
        \interpret{\sAr{\epsilon}} & = \pair{0, 0}
    \end{align*}
    Here, again, subscripts indicate tuple indexing; i.e., $\pair{n,m}_1 = n$ and $\pair{n,m}_2 = m$.
    Note that for every ground term $s$ we have $\interpret{s}_2 = \mathit{\#bs}(s)$.  The first component
    exactly sums $\mathit{\#bs}(t)$ for every subterm $t$ of $s$ which has the form $\sAr{a}(t')$.  We have:
    $
        \interpret{\sAr{a}(\sAr{b}(x))} = \pair{ x_1 + x_2 + 1, x_2 + 1}
        \sortgr{\nat} \pair{ x_1 + x_2, x_2 + 1 }
        = \interpret{\sAr{b}(\sAr{a}(x))}
    $.
    The interpretation functions $\funcinterpret{\sAr{a}}$ and $\funcinterpret{\sAr{b}}$ are indeed monotonic.
    For example, for $\funcinterpret{\sAr{a}}$: if $x \sortgr{\nat} y$ then $x_1 + x_2 > y_1 + y_2$ (since $x_1
    > y_1$ and $x_2 \geq y_2$) and $x_2 \geq y_2$; and if $x \sortgeq{\nat} y$ then $x_1 + x_2 \geq y_1 + y_2$
    and $x_2 \geq y_2$.
    We have $\interpret{\sAr{a}^n(\sAr{b}^m(\sAr{\epsilon}))} = (n*m,m)$.
\end{example}

To build strongly monotonic functions we can for instance use the following observation:
\begin{lemma}\label{lem:maxpol}
    A function
    $F : \N^{\typecount{\iota_1}} \times \dots \times \N^{\typecount{\iota_k}} \arrfunc \N^{\typecount{\kappa}}$
    is strongly monotonic if we can write
    $F(x^1,\dots,x^k) = \pair{x^1_1 + \dots + x^k_1 + S_1(x^1,\dots,x^k),\ S_2(x^1,\dots,x^k), \dots,\ S_{\typecount{\kappa}}(x^1, \dots,x^k)}$,
    where each $S_i$ is a weakly monotonic function in $\N^{\typecount{\iota_1}} \times
        \dots \times \N^{\typecount{\iota_k}} \arrfunc \N$.

    Moreover, a function
    $S : \N^{\typecount{\iota_1}} \times \dots \times \N^{\typecount{\iota_k}} \arrfunc \N$
    is weakly monotonic if it is built from constants in $\N$,
    variable components $x_j^n$ and weakly monotonic functions in $\N^n \arrfunc \N$.
\end{lemma}
For the ``weakly monotonic functions in $\N^n \arrfunc \N$'' we could for instance use $+$, $*$ or $\max$.

To determine the length $\typecount{\iota}$ of the tuple for a sort $\iota$,
we use a semantic approach,
similar to one 
used in \cite{danner:et-al:15} in the context of functional languages:
the elements of the tuple are values of interest for the sort.
The two prominent examples 
in this paper are the sort $\sAr{nat}$ of natural numbers---which is constructed from the symbols
$\sAr{0 :: nat}$ and $\sAr{s :: [nat] \arrtype nat}$---and the sort $\lst$ of lists of natural numbers---which is constructed using
$\sAr{nil :: list}$ and $\sAr{cons :: [nat \times list] \arrtype list}$.
For natural numbers, 
we
consider their size, so the number of $\sAr{s}$s.
For lists, we consider both their length and an upper bound on the size of their elements.
This gives $\typecount{\sAr{nat}} = 2$ (cost of reducing the term, size of its normal form) and
$\typecount{\sAr{list}} = 3$ (cost of reducing, length of normal form, maximum element size).
In the remainder of this paper, we will use $x_\cost$ as syntactic sugar for $x_1$
(the cost component of $x$), $x_\size$ and $x_\leng$ as $x_2$ and $x_\mmax$ as $x_3$.

\begin{example}\label{ex:revinterpret}
    Consider the TRS defined in Example \ref{ex:append-reverse}.
    We start by giving an interpretation for the type constructors: the symbols $\zero,\nil,\suc$
    and $\cons$ which are used to construct natural numbers and lists.  To be in line with the
    semantics for the type interpretation, we let:
    \begin{align*}
        \interpret{\zero} & = \pair{0, 0}  & \interpret{\suc(\avar)}       & = \pair{\avar_\cost, \avar_\size + 1}                                                           \\
        \interpret{\nil}  & = \pair{0,0,0} & \interpret{\avar : \aListVar} & = \pair{\avar_\cost + \aListVar_\cost, \aListVar_\leng + 1, \max(\avar_\size, \aListVar_\mmax)}
    \end{align*}
    This expresses that $\zero$ has no evaluation cost and size $0$;
    analogously, $\nil$ has no evaluation cost and $0$ as length and maximum element.
    The 
    cost of evaluating a term $\suc(t)$ depends entirely on the cost of the term's argument $t$;
    the size component counts the number of $\suc$s.
    The cost component for $\cons$ similarly sums the costs of its arguments,
    while the length is increased by 1, and the maximum element 
    is the maximum between its head and tail.

    For the remaining symbols we choose the following interpretations:
    \[
        \begin{array}{rcllrcl}
            \interpret{\avar \add \bvar}              & = & \pair{\avar_\cost + \bvar_\cost + \bvar_\size + 1, \avar_\size + \bvar_\size}                                                      \\
            \interpret{\sumList(\aListVar)}           & = & \pair{\aListVar_\cost + 2*\aListVar_\leng + \aListVar_\leng * \aListVar_\mmax + 1, \aListVar_\leng * \aListVar_\mmax}              \\
            \interpret{\rev(\aListVar)}               & = & \pair{\aListVar_\cost + \aListVar_\leng + \frac{\aListVar_\leng * (\aListVar_\leng + 1)}{2} + 1, \aListVar_\leng, \aListVar_\mmax} \\
            \interpret{\append(\aListVar, \bListVar)} & = &
            \pair{\aListVar_\cost + \bListVar_\cost + \aListVar_\leng + 1, \aListVar_\leng + \bListVar_\leng, \max(\aListVar_\mmax, \bListVar_\mmax)}
        \end{array}
    \]
    Checking compatibility is easily done for the interpretation above, and strong monotonicity follows by
    Lemma \ref{lem:maxpol} (as $n \mapsto \frac{n * (n+1)}{2} \in \N \arrfunc \N$ is weakly monotonic).
    We see that the cost of evaluating $\append$ is linear in the first list length and independent of the size of the list elements,
    while evaluating $\sumList$ gives a quadratic dependency on length and size combined.
\end{example}

Our tuple interpretations have some similarities with matrix interpretations
\cite{waldmann:zantema:06}, where also each term is associated to an $n$-tuple.  In essence, matrix
interpretations \emph{are} tuple interpretations, for systems with only one sort.  However, the
shape of the interpretation functions $\funcinterpret{\afun}$ in matrix interpretations is limited
to functions following Lemma \ref{lem:maxpol} where each $S$ is a linear multivariate polynomial.
Hence, our interpretations are a strict generalisation, which also admits interpretations such as
those used for $\sumList$, $\rev$ and $\append$ in Example \ref{ex:revinterpret}.

For the purpose of termination, tuple interpretations strictly extend the power of both
polynomial interpretations and matrix interpretations already in the first-order case.

\begin{example}\label{ex:quot}
    A TRS that implements division in \cite{art:gie:00:1} shows a limitation of polynomial interpretations:
    it contains a rule $\quot(\suc(x), \suc(y)) \to \suc(\quot(\minus(x,y), \suc(y)))$ which cannot
    be oriented by any polynomial interpretation, because $\interpret{\minus(x,\suc(x))} >
    \interpret{\suc(x)}$ for any strongly monotonic polynomial $\funcinterpret{\minus}$. 
    Due to the duplication of $y$, this rule also cannot be handled by a matrix interpretation.
    However, we do have a compatible tuple interpretation:
    \[
        \begin{array}{rclcrcl}
            \interpret{\zero}        & = & \pair{0,0}                                                                            & \quad &
            \interpret{\minus(x, y)} & = & \pair{x_\cost + y_\cost + y_\size + 1, x_\size}                                                 \\
            \interpret{\suc(x)}      & = & \pair{x_\cost,x_\size+1}                                                              & \quad &
            \interpret{\quot(x,y)}   & = & \pair{x_\cost + x_\size + y_\cost + x_\size * y_\cost + x_\size * y_\size + 1, x_\size}           \\
        \end{array}
    \]
\end{example}

In practice, in first-order termination or complexity analysis one would not exclusively use
interpretations, but rather a combination of different techniques.
In that context, tuple interpretations may be used as one part of a large toolbox.
They are likely to offer a simple complexity proof in many cases,
but they are unlikely to be an essential technique since so many other methods have already been developed.
Indeed, all examples in this section can be handled with previously established theory.
For instance, Example \ref{ex:poly_overstimation} can be handled with matrix interpretations,
while $\sumList$ and $\rev$ may be analysed using ideas from \cite{nao:moser:08} and \cite{mos:sch:wal:08}.

However, developing a new technique for first-order termination and traditional complexity analysis
is not our goal.
Our method \emph{does} provide a more fine-grained notion of complexity, which may
consider information such as the length of a list.
Moreover, the first-order case is an important
stepping stone towards higher-order analysis, where far fewer methods exist.

\section{Higher-order tuple interpretations}\label{sec:ho-tp-int}

In this section, we will extend the ideas from Section \ref{sec:fo-tp-int} to the higher-order setting,
and hence define the core notion of this paper: higher-order tuple interpretations.  To do this, we will
build on the notion of \emph{strongly monotonic algebras} originating in \cite{pol:96}.

\subsection{Strongly monotonic algebras}\label{subsec:SM}

In first-order term rewriting, the complexity of a TRS is often measured as \emph{runtime} or
\emph{derivational} complexity.
Both measures consider initial terms $s$ of a certain shape,
and supply a bound on $\dht{s}$ given the size of $s$.
However, this is not a good approach for higher-order terms: the behaviour of a term of higher
type generally cannot be captured in an integer.

\begin{example}\label{ex:hor_f}
    Consider the AFS obtained by combining Examples \ref{ex:append-reverse} and \ref{ex:ho_trs}.
    The evaluation cost of a term $\foldl(F, n, q)$ depends almost completely on the behaviour of the
    functional subterm $F$, and not only on its evaluation cost.
    To see this, consider two cases:
    $F_1 := \abs{x}{\abs{y}{y \add x}}$, and
    $F_2 := \abs{x}{\abs{y}{x \add x}}$.
    For natural numbers $n,m$, the evaluation cost of both $F_1(n,m)$ and $F_2(n,m)$ is the same:
    $n + 1$.  However, the \emph{size} of the result is different.  Hence, the number of
    steps needed to compute $\foldl(F_1,n,q)$ for a number $n$ and list $q$ is quadratic in
    the size of $n$ and $q$, while the number of steps needed for $\foldl(F_2,n,q)$ is exponential.
\end{example}

As Example \ref{ex:hor_f} shows, higher-order rewriting is a natural place to separate cost and size.
But more than that, we need to know what a function does with its arguments: whether it is
size-increasing, how long it takes to evaluate them, and more.

This is naturally captured by the notion of (weakly or strongly) monotonic algebras for higher-order
rewriting introduced by v.d.~Pol \cite{pol:96}: here, a term of arrow type is interpreted as a
function, which allows the interpretation to retain all relevant information.

Monotonic interpretations were originally defined for a different higher-order rewriting
formalism, which does make some difference in the way abstraction and application is handled.
\emph{Weakly} monotonic algebras were transposed to AFSs in \cite{fuh:kop:12};
however, here we extend the more natural notion of \emph{hereditarily} monotonic algebras
which v.d.~Pol only briefly considered.%
\footnote{In \cite{pol:96}, v.d.~Pol rejects hereditarily (or: strongly) monotonic algebras because
    they are not so well-suited for analysing the HRS format \cite{nip:91} where reasoning is modulo
    $\to_{\beta}$: it is impossible to both interpret all terms of function type to strongly monotonic
    functions and have $\interpret{(\abs{x}{s})\ t} = \interpret{s[x:=t]}$.
    In the AFS format, we do not have the latter requirement.
    In \cite{fuh:kop:12}, where the authors considered the AFS format like we do here (but for interpretations to $\N$ rather than to tuples), weakly monotonic
    algebras were used because they are a more natural choice
    in the context of dependency pairs.
}

\begin{definition}\label{def:SM}
    Let $\sortset$ be a set of sorts and $\signature$ a higher-order signature.
    We assume given for every sort $\asort$ an extended well-founded set
    $(\sortinterpret{\asort},\sortgr{\asort},\sortgeq{\asort})$.
    From this, we define the set of \emph{strongly monotonic functionals}, as follows:
    \begin{itemize}
        \item For all sorts $\asort$: $\typeinterpret{\asort} := \sortinterpret{\asort}$ and $\typegr{\asort} \, := \, \sortgr{\asort}$
              and $\typegeq{\asort} \, := \, \sortgeq{\asort}$.
        \item For an arrow type $\atype \arrtype \btype$:
              \begin{itemize}
                  \item $\typeinterpret{\atype \arrtype \btype} := \{ F \in \typeinterpret{\atype} \arrfunc \typeinterpret{\btype} \mid
                            F$ is strongly monotoic$\}$
                  \item $F \typegr{\atype \arrtype \btype} G$ iff $\typeinterpret{\atype}$ is non-empty and
                    $\forall x \in \typeinterpret{\atype}.F(x) \typegr{\btype} G(x)$, and
                    \\
                  $F \typegeq{\atype \arrtype \btype} G$ iff $\forall x \in \typeinterpret{\atype}.F(x) \typegeq{\btype} G(x)$.
              \end{itemize}
    \end{itemize}
\end{definition}
That is, $\typeinterpret{\atype \arrtype \btype}$ contains strongly monotonic functions from $\typeinterpret{\atype}$ to
$\typeinterpret{\btype}$ and both $\typegr{\atype \arrtype \btype}$ and $\typegeq{\atype \arrtype\btype}$ do a point-wise comparison.
By a straightforward induction on types we have:
\begin{restatable}{lemma}{typeOrd}\label{lemma:typeOrd}
    For all types $\atype$, $(\typeinterpret{\atype},\typegr{\atype},\typegeq{\atype})$ is an extended well-founded set;
    that is:
    \begin{itemize}
        \item $\typegr{\atype}$ is well-founded and $\typegeq{\atype}$ is reflexive;
        \item both $\typegr{\atype}$ and $\typegeq{\atype}$ are transitive;
        \item for all $x, y, z \in \typeinterpret{\atype}$, $x \typegr{\atype} y$ implies $x \typegeq{\atype} y$ and
              $x \typegr{\atype} y \typegeq{\atype} z$ implies $x \typegr{\atype} z$.
    \end{itemize}
\end{restatable}

We will define \emph{higher-order strongly monotonic algebras} as an extension of Definition \ref{def:mult-sorted-alg},
mapping a term of type $\atype$ to an element of $\typeinterpret{\atype}$.
Functional terms $\afun(s_1,\dots,s_k)$ and variables can be handled as before, but we now also have to deal with application and abstraction.
Application is straightforward: since terms of higher type are mapped to functions, we can interpret application as function application, i.e.,
$\ainterpret{\aterm \cdot \bterm} := \ainterpret{\aterm}(\ainterpret{\bterm})$.
However, abstraction is more difficult.
The natural choice would be to view abstraction as defining a function; i.e.,
let $\ainterpret{\abs{\avar}{\aterm}}$ be the function $d \mapsto \interpret{\aterm}_{\alpha[x := d]}$.
Unfortunately, this is not necessarily monotonic: $d \mapsto \interpret{\aterm}_{\alpha[x := d]}$ is strongly monotonic only if $\avar$ occurs freely in $\aterm$.
For example $\abs{\avar}{\zero}$ would be mapped to a constant function, which is not in $\typeinterpret{\nat \arrtype \nat}$.
Moreover, this definition would give $\ainterpret{(\abs{x}{s}) \cdot t} = \ainterpret{s[x:=t]}$, so $\beta$-steps would not be counted toward the evaluation cost.

We handle both problems by using a choosable function $\makesm{\atype,\btype}$, which takes a function that may be strongly monotonic or constant, and turns it strongly monotonic.

\begin{definition}\label{def:monomaker}
    A $(\atype,\btype)$-monotonicity function $\makesm{\atype,\btype}$ is a
    strongly monotonic function in $C_{\atype,\btype} \arrfunc \typeinterpret{\atype \arrfunc \btype}$,
    where the set $C_{\atype,\btype}$ is defined as $\typeinterpret{\atype \arrtype \btype} \cup
        \{ F \in \typeinterpret{\atype} \arrfunc \typeinterpret{\btype} \mid F(x) = F(y)$ for all $x,y \in
        \typeinterpret{\atype} \}$.  (Here, the set $C_{\atype,\btype}$ is ordered by point-wise comparison.)
\end{definition}

With this definition, we are ready to define strongly monotonic algebras.

\begin{definition}\label{def:sminterpret}
    A strongly monotonic algebra $\SMA$ consists of a family
    $(\typeinterpret{\atype},\typegr{\atype},\typegeq{\atype})_{\atype \in \simpletypeset}$,
    an interpretation function $\funcinterpret{}$ which associates to each
    $\sAr{\afun :: [\atype_1 \times \dots \times \atype_k] \arrtype \btype}$
    in $\signature$ an element of
    $\typeinterpret{\atype_1 \arrtype \dots \arrtype\atype_k \arrtype \btype}$,
    and a $(\atype,\btype)$-monotonicity function $\makesm{\atype,\btype}$,
    for each $\atype,\btype \in \simpletypeset$.

    Let $\alpha$ be a function that maps variables of type $\atype$
    to elements of $\typeinterpret{\atype}$.
    We extend $\funcinterpret{}$ to a function $\ainterpret{\cdot}$ that
    maps terms of type $\atype$ to elements of $\typeinterpret{\atype}$, as follows:
    \[
        \begin{array}{lcl}
            \ainterpret{\avar} = \alpha(\avar)\ \text{for variables}\ \avar & \quad &
            \ainterpret{\afun(\aterm_1,\dots,\aterm_k)} = \funcinterpret{\afun}(\ainterpret{\aterm_1},\dots,\ainterpret{\aterm_k})\\
            \ainterpret{\app{\aterm}{\bterm}} = \ainterpret{\aterm}(\ainterpret{\bterm}) & &
            \ainterpret{\abs{x}{\aterm}} = \makesm{\atype,\btype}(d \mapsto \interpret{\aterm}_{\alpha[x:=d]})\ \text{if}\ \avar :: \atype\ \text{and}\ \aterm :: \btype
        \end{array}
    \]
\end{definition}
We can see by induction on $\aterm$ that for $\aterm :: \atype$ indeed $\ainterpret{\aterm} \in \typeinterpret{\atype}$.
We say that an AFS $\trs$ is \emph{compatible} with $\SMA$ if for all valuations $\alpha$ both
(1) $\ainterpret{\ell} \typegr{} \ainterpret{r}$,
for all $\ell \to r \in \rules$; and
(2) $\ainterpret{(\abs{\avar}{\aterm}) \, \bterm} \typegr{} \ainterpret{\aterm[\avar := \bterm]}$,
for any $s :: \atype$, $t :: \btype$ and $x \in \var_\btype$.

As before, we will typically omit the $\alpha$ subscript and use notation like $\interpret{s} = F(x+3)$
to denote $\ainterpret{s} = \varinterpret(F)(\varinterpret(x) + 3)$.
When types are not relevant, we will denote $\typegr{}$ instead of specifying $\typegr{\atype}$, and
we may write $f \in \typeinterpret{}$ to mean $f \in \typeinterpret{\atype}$ for some $\atype \in \simpletypeset$.

We extend Theorem \ref{thm:sortinterpretworks} into the following compatibility result.

\begin{restatable}{theorem}{typeIntWorks}\label{thm:typeinterpretworks}
    If $\trs$ is compatible with $\SMA$,
    then for all $\alpha$, $\ainterpret{s} \typegr{} \ainterpret{t}$ when $s \arrz t$.
\end{restatable}

For Definition \ref{def:SM} and Theorem \ref{thm:typeinterpretworks}, we can choose the well-founded sets
$(A_\asort,\sortgr{\asort},\sortgeq{\asort})$ for each sort, and the functions $\makesm{\atype, \btype}$ for each pair of types, 
as we desire.
A \emph{higher-order tuple algebra} is a strongly monotonic algebra where each
$(A_\asort,\sortgr{\asort}, \sortgeq{\asort})$ follows Definition \ref{def:tuple_algebra}.

\begin{example}\label{ex:comp_of_map}
    Let $A_{\sAr{nat}} = \mathbb{N}^2$ and $A_{\sAr{list}} = \mathbb{N}^3$ as before, and assume $\cons$ and $\nil$ are interpreted as in Example~\ref{ex:revinterpret}.
    Consider the rules for $\map$ in Example \ref{ex:ho_trs}.
    We let:
    \[
        \interpret{\map(F,\aListVar)} = \pair{
            (\aListVar_\leng + 1) * (F(\pair{\aListVar_\cost,\aListVar_\mmax})_\cost + 1), \,
            \aListVar_\leng, \,
            F(\aListVar_\cost,\aListVar_\mmax)_\size}
    \]
    This expresses that $\map$ does not increase the list length (as the length component is just $\aListVar_\leng$),
    the greatest element of the result is bounded by the value of $F$ on the greatest element of $\aListVar$,
    and the evaluation cost is mostly expressed by a number of $F$ steps that is linear in the length of $\aListVar$.
    We will see in Lemma \ref{lem:wm} that $\funcinterpret{\map}$ is indeed strongly monotonic.

    To prove compatibility of the AFS with $\SMA$, we must first see that $\interpret{\ell} \typegr{}
        \interpret{r}$ for all rules $\ell \arrz r$.  For the first $\map$ rule this is easy:
    $\interpret{\map(F,\nil)} = \pair{F(\pair{0,0})_\cost + 1,0,F(\pair{0,0})_\size} \typegr{\lst} \pair{0,0,0} =
        \interpret{\nil}$.
    For the second $\map$ rule, we must check that $\pair{\text{cost-}\ell,\text{len-}\ell,\text{max-}\ell}
        \typegr{\lst} \pair{\text{cost-}r,\text{len-}r,\text{max-}r}$; that is,
    $\text{cost-}\ell > \text{cost-}r$ and $\text{len-}\ell \geq \text{len-}r$ and $\text{max-}\ell \geq \text{max-}r$, where:
    \[
        \begin{array}{lclclcccl}
            \text{cost-}\ell & = & \interpret{\map(F, x : \aListVar)}_\cost & = &
              \multicolumn{5}{l}{
                  (\aListVar_\leng + 2) * (F(\pair{x_\cost + \aListVar_\cost, \max(x_\size,\aListVar_\mmax)})_\cost + 1)
                  } \\
            \text{cost-}r & = & \interpret{F(x) : \map(F, \aListVar)}_\cost & = &
              \multicolumn{5}{l}{F(\pair{x_\cost,x_\size})_\cost + (\aListVar_\leng + 1) * (F(\pair{\aListVar_\cost,\aListVar_\mmax})_\cost + 1)} \\
            \text{len-}\ell & = & \interpret{\map(F, x : \aListVar)}_\leng & = & \aListVar_\leng + 1 & = & \interpret{F(x) : \map(F, \aListVar)}_\leng & = & \text{len-}r \\
            \text{max-}\ell & = & \interpret{\map(F, x : \aListVar)}_\mmax & = &
              \multicolumn{5}{l}{F(\pair{x_\cost + \aListVar_\cost,\max(x_\size,\aListVar_\mmax)})_\size} \\
            \text{max-}r & = & \interpret{F(x) : \map(F, \aListVar)}_\mmax & = &
              \multicolumn{5}{l}{\max(F(\pair{x_\cost,x_\size})_\size, F(\pair{\aListVar_\cost,\aListVar_\mmax})_\size)} \\
        \end{array}
    \]
    To see why $\text{cost-}\ell > \text{cost-}r$, we
    observe that for all $x,\aListVar$: $\pair{x_\cost + \aListVar_\cost,\max(x_\size + \aListVar_\mmax)}
        \typegeq{\nat}$ both $\pair{x_\cost,x_\size}$ and $\pair{\aListVar_\cost,\aListVar_\mmax}$.
    Since $F \in \typeinterpret{\nat \arrtype \nat}$ therefore
    $F(\pair{x_\cost + \aListVar_\cost,\max(x_\size + \aListVar_\mmax)}) \typegeq{\nat}$ both
    $F(\pair{x_\cost,x_\size})$ and $F(\pair{\aListVar_\cost,\aListVar_\mmax})$.
    We find $\text{max-}\ell \geq \text{max-}r$ by a similar reasoning.
\end{example}

\subsection{Interpreting abstractions}\label{subsec:makesm}

Example \ref{ex:comp_of_map} is not complete: we have not yet defined the functions
$\makesm{\atype,\btype}$
, and we have not shown that
$\interpret{(\abs{x}{s}) \ t} \typegr{} \interpret{s[x:=t]}$ always holds.
To achieve this, we will define some standard functions to build elements of $\SM$.
This allows us to easily construct strongly monotonic functionals, both to build $\makesm{\atype,\btype}$
and to create interpretation functions $\funcinterpret{\afun}$.

\begin{definition}
    For every type $\atype$, we define:
    $\nul{\atype} \in \typeinterpret{\atype}$;
    $\costof{\atype} \in \typeinterpret{\atype} \arrfunc \N$; and
    $\addcost{\atype} \in \N \times \typeinterpret{\atype} \arrfunc \typeinterpret{\atype}$
    by mutual recursion on $\atype$ as follows.
    \begin{align*}
         & \nul{\asort} = \pair{0,\dots,0}                                                                                  &  & \nul{\atype \arrtype \btype} = d \mapsto \addcost{\btype}(\costof{\atype}(d),\nul{\btype}) \\
         & \costof{\asort}( \pair{n_1,\dots,n_{\typecount{\asort}}} ) = n_1                                                 &  & \costof{\atype \arrtype \btype}(F) = \costof{\btype}(F(\nul{\atype}))                      \\
         & \addcost{\asort}( c, \pair{n_1,\dots,n_{\typecount{\asort}}} ) = \pair{c + n_1,n_2,\dots,n_{\typecount{\asort}}} &  & \addcost{\atype \arrtype \btype}(c, F) = d \mapsto \addcost{\btype}(c, F(d))
    \end{align*}
\end{definition}

Here, $\nul{\atype}$ defines the \emph{minimal element} of $\typeinterpret{\atype}$.
The function $\costof{\atype}$ maps every $F$ to the cost component of
$F(\nul{\atype_1},\dots,\nul{\atype_m})$; hence, if $F \typegr{\atype} G$ we have
$\costof{\atype}(F) > \costof{\atype}(G)$.
The function $\addcost{\atype}$ pointwise increases an element of $\typeinterpret{\atype}$ by adding to the cost component:
if $F(\avar_1,\dots,\avar_m) = \pair{n_1,\dots,n_k}$, then
$\addcost{}(c,F)(\avar_1,\dots,\avar_m) = \pair{c + n_1,n_2,\dots,n_k}$.

It is easy to see that $\nul{\atype}$ and $\addcost{\atype}(n,X)$ are in $\SM$ for all $\atype$ (by
induction on $\atype$), and that $\costof{\atype}$ and $\addcost{\atype}$ are
strict in all their arguments.
Various properties of these functions are detailed in the appendix
(Lemmas \ref{lem:builders}--\ref{lem:argumentincrease}).
We will particularly use that always
$\aFuncVar(\addcost{}(n,\avar)) \typegeq{} \addcost{}(n,\aFuncVar(\avar))$ (Lemma \ref{lem:estimatebycost}) and
$\costof{}(\aFuncVar(\avar)) \geq \costof{}(\avar)$ (Lemma \ref{lem:argumentincrease}).

We can use these functions to for instance create candidates for $\makesm{\atype,\btype}$.
While many suitable definitions are possible, we will particularly consider the following:

\newcommand{\addarg}{\Phi}

\begin{definition}\label{def:makesmchoices}
For types $\atype,\btype$, and $F$ a weakly monotonic function in $\typeinterpret{\atype}
\arrfunc \typeinterpret{\btype}$, let:
\[
  \addarg_{\atype,\btype}(F) = \left\{
    \begin{array}{ll}
      d \mapsto \addcost{\atype \arrtype \btype}(1,F(d))      & \text{if}\ F\ \text{is in}\ \typeinterpret{\atype \arrtype \btype} \\
      d \mapsto \addcost{\atype \arrtype \btype}(\costof{\atype}(d) + 1, F(d)) & \text{otherwise}  \\
    \end{array}
    \right.
\]
\end{definition}

Then $\addarg_{\atype,\btype}$ is a $(\atype,\btype)$-monotonicity function.
To see this, the most challenging part is proving that $\addarg_{\atype,\btype}(F) \typegr{}
\addarg_{\atype,\btype}(G)$ if $F \typegr{} G$ and $F \in \typeinterpret{\atype \arrtype \btype}$
while $G$ is a constant function.  We can prove this using the result that $x \typegr{} y$ implies
$\addcost{}(1,x) \typegeq{} y$ for all $x,y$.
We have:

\begin{restatable}{lemma}{makeSMbeta}
If $\makesm{\atype,\btype} = \addarg_{\atype,\btype}$ then
$\interpret{(\abs{x}{\aterm}) \, \bterm} \typegr{\btype} \interpret{\aterm[x := \bterm]}$, for
$s :: \btype$, $t :: \atype$, $x \in \var_\atype$.
\end{restatable}

\begin{proof}[Proof Sketch]
    We expand $\makesm{\atype, \btype}$ to achieve
    $
        \interpret{(\abs{\avar}{\aterm}) \, \bterm}_{\alpha} =
        \addcost{\btype}(\costof{\atype}(\interpret{\bterm}_\alpha) + 1,
        \interpret{\aterm}_{\alpha[x := \interpret{\bterm}]})
    $
    or
    $        \interpret{(\abs{\avar}{\aterm}) \, \bterm}_{\alpha} = \addcost{\btype}(1, \interpret{\aterm}_{\alpha[x := \interpret{\bterm}]})
    $.
    By induction on $\tau$ we prove that $\addcost{\tau}(n, F) \typegr{\tau} F$ for all $n \geq 1$.
    So either way, $\interpret{(\abs{\avar}{\aterm}) \, \bterm}_{\alpha} \typegr{\tau} \interpret{\aterm}_{\alpha[x := \interpret{\bterm}]}$.
    Finally, we prove a substitution lemma,
    $\interpret{\aterm}_{\alpha[x :=\interpret{\bterm}_\alpha]} = \interpret{\aterm[\avar := \bterm]}_\alpha$,
    by induction on $s$.
\end{proof}

In examples in the remainder of this paper, we will assume that $\makesm{\atype,\btype} =
\addarg_{\atype,\btype}$.
With these choices we do not only orient the $\beta$-rule
(and thus satisfy item (2) of the compatibility conditions),
but also the $\eta$-reduction rules mentioned in \seclong\ref{sec:prelim:ho}.

\begin{restatable}{lemma}{makeSMeta}
    If $\makesm{\atype,\btype} = \addarg_{\atype,\btype}$
    then for any $F \in \var_{\atype \arrtype \btype}$ we have: $\interpret{\abs{x}{F \, x}} \typegr{
            \atype \arrtype \btype} \interpret{F}$.
\end{restatable}

\begin{proof}[Proof Sketch]
    Since $F \neq \avar$, we have $\interpret{F}_{\varinterpret[\avar:=d]} = \varinterpret(F)$
    for all $\varinterpret$ and $d$.
    Consequently,
    $\interpret{\abs{\avar}{F \, \avar}} \typegeq{\atype \arrtype \btype} d \mapsto \addcost{\btype}(1,F(d))$
    either way.
    We are done as: $\addcost{\btype}(1,F(d)) \typegr{\btype} F(d)$.
\end{proof}

\subsection{Creating strongly monotonic interpretation functions}\label{sec:createmono}

We can use Theorem \ref{thm:typeinterpretworks} to obtain bounds on the derivation heights of given terms.
However, to achieve this, we must find an interpretation function $\funcinterpret{}$,
and prove that each $\funcinterpret{\afun}$ is in $\typeinterpret{}$.
We will now explore ways to construct such strongly monotonic functions.
It turns out to be useful to also consider \emph{weakly} monotonic functions.
In the following, we will write ``$f$ is $\weakset(A_1,\dots,A_k;B)$''
to mean that $f$ is a weakly monotonic function in
$A_1 \times \dots \times A_k \arrfunc B$.

\begin{lemma}\label{lem:wm}
    Let $x^1,\dots,x^k$ be variables ranging over $\typeinterpret{\atype_1},\dots,
        \typeinterpret{\atype_k}$ respectively; we shortly denote this sequence $\vec{x}$.
    We let $\wmatypes$ denote the sequence $\typeinterpret{\atype_1},\dots,\typeinterpret{\atype_k}$.
    Then:
    \begin{enumerate}
        \item\label{lem:wm:project}
              if $F(\vec{x}) = x^i$ then $F$ is $\weakset(\wmatypes;\typeinterpret{\atype_i})$,
              and $F$ is strict in argument $i$;
        \item\label{lem:wm:projectapply}
              if $F(\vec{x}) = x^i(F_1(\vec{x}),\dots,F_n(\vec{x}))$,
              $\atype_i = \btype_1 \arrtype \dots \arrtype \btype_n \arrtype \ctype$, and each
              $F_j$ is $\weakset(\wmatypes;\typeinterpret{\btype_j})$ then
              $F$ is $\weakset(\wmatypes;\typeinterpret{\ctype})$ and for all $p \in \{1,\dots,k\}$:
              $F$ is strict in argument $p$ if $p = i$ or some $F_j$ is strict in argument $p$;
        \item\label{lem:wm:smaller}
              if $F(\vec{x}) = \pair{ G_1(\vec{x}),\dots, G_{\typecount{\asort}}(\vec{x}) }$ and
              each $G_j$ is $\weakset(\wmatypes;\N)$ then $F$ is $\weakset(\wmatypes;\typeinterpret{\asort})$,
              and for all $p \in \{1,\dots,k\}$: $F$ is strict in argument $p$ if $G_1$ is.
    \end{enumerate}
    The last result uses functions mapping to $\N$; these can be constructed using the 
    observations:
    \begin{enumerate}[resume]
        \item\label{lem:wm:constant}
              if $G(\vec{x}) = n$ for some $n \in \N$ then $G$ is $\weakset(\wmatypes;\N)$;
        \item\label{lem:wm:projectindex}
              if $G(\vec{x}) = x^i_j$ and $\atype_i = \asort \in \sortset$ and $1 \leq j \leq \typecount{\asort}$,
              then $G$ is $\weakset(\wmatypes;\N)$, and $G$ is strict in argument $i$ if $j = 1$;
        \item\label{lem:wm:operator}
              if $G(\vec{x}) = f(G_1(\vec{x}),\dots,G_n(\vec{x}))$ and all $G_j$ are
              $\weakset(\wmatypes;\N)$ and $f$ is $\weakset(\N,\dots,\N;\N)$, then $G$ is
              $\weakset(\wmatypes;\N)$, and for all $p \in \{1,\dots,k\}$: $G$ is strict in argument
              $p$ if, for some $j \in \{1,\dots,n\}$: $G_j$ is
              strict in argument $p$ and $f$ is strict in argument $j$;
        \item\label{lem:wm:bigger}
              if $G(\vec{x}) = F(\vec{x})_j$ and
              $F$ is $\weakset(\wmatypes;\typeinterpret{\asort})$ and $1 \leq j \leq \typecount{\asort}$ then $G$
              is $\weakset(\wmatypes;\N)$ and if $j = 1$ then for all $p \in \{1,\dots,k\}$:
              $G$ is strict in argument $p$ if $F$ is.
    \end{enumerate}
\end{lemma}

\begin{proof}[Proof Sketch]
    We easily see that in each case, $F$ or $G$ is in the given function space.  To show weak monotonicity, assume
    given both $\vec{x}$ and $\vec{y}$ such that each $x^i \typegeq{} y^i$; we then check for all cases
    that $F(\vec{x}) \typegeq{} F(\vec{y})$, or $G(\vec{x}) \geq G(\vec{y})$.
    For the strictness conditions, we assume that $x^p \typegr{} y^p$ and similarly check all cases.
\end{proof}

The reader may recognise items (\ref{lem:wm:constant}--\ref{lem:wm:operator}): these largely
correspond to the sufficient conditions for a weakly monotonic function $S$ in Lemma \ref{lem:maxpol}.
For the function $f$ in item (\ref{lem:wm:operator}), we could for instance choose
$+$, $*$ or $\max$, where $+$ is strict in all arguments.
However, we can get beyond Lemma \ref{lem:maxpol} by using the other items; for example, applying
variables to each other.

Now, if a function $f$ is $\weakset(\wmatypes;\typeinterpret{\btype})$ and $f$ is strict in all its
arguments, then we easily see that the function $d_1 \mapsto \dots \mapsto d_k \mapsto f(d_1,\dots,
    d_k)$ is an element of $\typeinterpret{\atype_1 \arrtype \dots \arrtype \atype_k \arrtype \btype}$.
To illustrate how this can be used in practice, we show monotonicity of $\funcinterpret{\map}$ of
Example \ref{ex:comp_of_map}:

\begin{example}
    Suppose
    $\funcinterpret{\map}(F,\listvar) = (\ F(\pair{\listvar_\cost,\listvar_\mmax})_\cost + \listvar_\leng
        * F(\pair{\listvar_\cost,\listvar_\mmax})_\cost + \listvar_\leng + 1\ ,\ \listvar_\leng\ ,\
        F(\pair{\listvar_\cost,\listvar_\mmax})_\leng\ )$.
    By (\ref{lem:wm:projectindex}), the functions
    $(F,\listvar) \mapsto \listvar_i$ are $\weakset(\typeinterpret{\nat \arrtype\nat},
        \typeinterpret{\lst};\N)$ for $i \in \{\cost,\leng,\mmax\}$ and moreover,
    $(F,\listvar) \mapsto \listvar_\cost$ is strict in argument 2.
    Hence, by (\ref{lem:wm:smaller}),
    $(F,\listvar) \mapsto \pair{\listvar_\cost,\listvar_\mmax}$ is $\weakset(\typeinterpret{\nat
            \arrtype \nat},\typeinterpret{\lst};\linebreak
        \typeinterpret{\nat})$ and strict in argument 2.
    Therefore, by (\ref{lem:wm:projectapply}),
    $(F,\listvar) \mapsto F(\pair{\listvar_\cost,\listvar_\mmax})$ is
    $\weakset(\typeinterpret{\nat
            \arrtype \nat},\linebreak
        \typeinterpret{\lst};\typeinterpret{\nat})$ and strict in both arguments.
    Hence, by (\ref{lem:wm:bigger}), $(F,\listvar) \mapsto F(\pair{\listvar_\cost,\listvar_\mmax})_\cost$ and
    $(F,\listvar) \mapsto F(\pair{\listvar_\cost,\listvar_\mmax})_\leng$ are
    $\weakset(\typeinterpret{\nat \arrtype\nat},\typeinterpret{\lst};\N)$ and the former is
    strict in both arguments.

    Continuing like this, it is not hard to see how we can iteratively prove that $(F,\listvar) \mapsto
    (\ F(\pair{\listvar_\cost,\listvar_\mmax})_\cost + \listvar_\leng * F(\pair{\listvar_\cost,
    \listvar_\mmax})_\cost + \listvar_\leng + 1\ ,\ \listvar_\leng\ ,\ F(\pair{\listvar_\cost,
    \listvar_\mmax})_\leng\ )$ is $\weakset(\typeinterpret{\nat \arrtype \nat},
    \typeinterpret{\lst};\typeinterpret{\lst})$ and strict in both arguments,
    which immediately gives $\funcinterpret{\map} \in \typeinterpret{(\nat \arrtype \nat) \arrtype \lst \arrtype \lst}$.
\end{example}

In practice, it is usually not needed to write such an elaborate proof: Lemma \ref{lem:wm}
essentially tells us that if a function is built exclusively using variables and variable
applications, projections $F(\vec{x})_j$, constants, and weakly monotonic operators over the natural
numbers, then that function is weakly monotonic; we only need to check that the cost component
indeed increases if one of the variables $x^i$ is increased.

Unfortunately, while Lemma \ref{lem:wm} is useful for rules like the ones for $\map$, it is not
enough to handle functions like $\foldl$, where the same function is repeatedly applied on a term.
As $\foldl$-like functions occur more often in higher-order rewriting, we should also address this.

To handle iteration, we define: for a function $Q \in A \arrfunc A$ and natural number $n$, let
$Q^n(a)$ indicate repeated function application; that is, $Q^0(a) = a$ and $Q^{n+1}(a) = Q^n(Q(a))$.

\begin{restatable}{lemma}{iteration}\label{lem:iterate}
    Suppose $F$ is $\weakset(\wmatypes,\typeinterpret{\btype \arrtype \btype})$ and
    $G$ is $\weakset(\wmatypes;\N)$.
    Suppose that for all $u^1 \in \typeinterpret{\atype_1},\dots,u^k \in \typeinterpret{\atype_k}$
    and $v \in \typeinterpret{\btype}$ we have: $F(u^1,\dots,u^k,v) \typegeq{\btype} v$.
    Then the function $(x^1,\dots,x^k) \mapsto F(x^1,\dots,x^k)^{G(x^1,\dots,x^k)}$ is
    $\weakset(\wmatypes,\typeinterpret{\btype \arrtype \btype})$.
\end{restatable}

With this in hand, we can orient the $\foldl$ rules of Example \ref{ex:ho_trs}.

\begin{example}\label{ex:foldl}
    For $F \in \typeinterpret{\nat \arrtype \nat \arrtype \nat}$ and $x,y \in \typeinterpret{\nat}$,
    let $\mathit{Helper}$ be defined by:
    \[\mathit{Helper}(F,y,x) = \pair{F(x,y)_\cost, \ \max(x_\size,F(x,y)_\size)}.\]
    Then $\mathit{Helper}$ is $\weakset(\typeinterpret{
            \nat \arrtype \nat \arrtype \nat},\typeinterpret{\nat},\typeinterpret{\nat};\typeinterpret{\nat})$
    and strict in its third argument by Lemma
    \ref{lem:wm}(\ref{lem:wm:project},\ref{lem:wm:projectapply},\ref{lem:wm:smaller},\ref{lem:wm:operator},\ref{lem:wm:bigger}),
    Hence, $\mathit{Helper}$ is $\weakset(\typeinterpret{
            \nat \arrtype \nat \arrtype \nat},\typeinterpret{\nat};\typeinterpret{\nat \arrtype \nat})$.
    Since, in general, $\costof{\nat}(F(x,y)) \geq \costof{\nat}(x)$, we have $\mathit{Helper}(F,y,x) \typegeq{\nat} x$.
    Using Lemma \ref{lem:iterate}, we therefore see that the function
    $(\aFuncVar,\cvar,\aListVar) \mapsto \mathit{Helper}(\aFuncVar,\pair{\aListVar_\cost,
            \aListVar_\mmax})^{\aListVar_\leng}(\cvar)$ is weakly monotonic, and strict in its second argument.
    This ensures that the following function is in $\typeinterpret{}$.
    \[
        \interpret{\foldl(\aFuncVar,\cvar,\aListVar)} =
        \mathit{Helper}(\aFuncVar,\pair{\aListVar_\cost,\aListVar_\mmax})^{\aListVar_\leng}(\pair{1 + \aListVar_\cost +
          \aListVar_\leng + \aFuncVar(\nul{\nat},\nul{\nat})_\cost + \cvar_\cost,\cvar_\size})
    \]
    This interpretation function is compatible with the rules for $\foldl$ in Example \ref{ex:ho_trs}.
    First, we have
    $\interpret{\foldl(\aFuncVar, \cvar, \nil)} = \pair{\ 1 + F(\nul{\nat},\nul{\nat})_\cost + \cvar_\cost,\ \cvar_\size \ }
        \typegr{\nat} \pair{\cvar_\cost,\cvar_\size} = \cvar$, which orients the first rule.  For the second,
    we will use the general property that (**) $\aFuncVar(\addcost{}(n,\avar),\bvar) \typegeq{} \addcost{}(n,\aFuncVar(\avar,\bvar))$
    (Lemma \ref{lem:addcostincrease}).
    We denote $A := \pair{\avar_\cost + \aListVar_\cost,\max(\avar_\size,\aListVar_\mmax)}$ and
    $B := 1 + \aListVar_\cost + \aListVar_\leng + \aFuncVar(\nul{\nat},\nul{\nat})_\cost + \cvar_\cost$.
    Then we have $\interpret{\foldl(\aFuncVar, \cvar, \avar : \aListVar)} =
    \mathit{Helper}(\aFuncVar,A)^{\aListVar_\leng+1}(\pair{B + \avar_\cost + 1,\cvar_\size})$, which:
    \[
        \begin{array}{rll}
             & \typegr{\nat} & \mathit{Helper}(\aFuncVar,A)^{\aListVar_\leng}(\mathit{Helper}(\aFuncVar,A,\pair{B,\cvar_\size}))
                \text{\ because}\ \pair{B + \avar_\cost + 1,\cvar_\size} \typegr{\nat} \pair{B,z_\size} \\
             & \typegeq{\nat} & \mathit{Helper}(\aFuncVar,A)^{\aListVar_\leng}(\aFuncVar(\pair{B,\cvar_\size},A))
                \text{\ because}\ \mathit{Helper}(\aFuncVar,n,m) \typegeq{\nat} \aFuncVar(m,n) \\
             & \typegr{\nat} &
            \mathit{Helper}(\aFuncVar,\pair{\aListVar_\cost,\aListVar_\mmax})^{\aListVar_\leng}(\aFuncVar(\pair{B,z_\size},\avar)) \text{\ because}\
               A \typegeq{\nat} \pair{\aListVar_\cost,\aListVar_\mmax}\ \text{and}\ A \typegeq{\nat} \avar \\
             & \typegeq{\nat} & \mathit{Helper}(F,\pair{\aListVar_\cost,\aListVar_\mmax})^{\aListVar_\leng}(\addcost{\nat}(
                1 + \aListVar_\cost + \aListVar_\leng + \aFuncVar(\nul{\nat},\nul{\nat})_\cost,\aFuncVar(\cvar,\avar)))\ \text{by (**)} \\
             & = & \interpret{\foldl(\aFuncVar, (\aFuncVar\ \cvar\ \avar), \aListVar)}.\\
        \end{array}
    \]
\end{example}

The interpretation in Example \ref{ex:foldl} may \emph{seem} too convoluted for practical use:
it does not obviously tell us something like ``$F$ is applied a linear number of times on terms whose
size is bounded by $n$''.
However, its value becomes clear when we plug in specific bounds for $F$.

\begin{example}\label{ex:foldl-sum}
    The function $\sumList$, defined in Example \ref{ex:append-reverse}, could alternatively be defined in terms of $\foldl$:
    let \(\sumList(\aListVar) \to \foldl(\abs{\avar\bvar}{(\avar \add \bvar)}, \zero, \aListVar)\).
    To find an interpretation for this function, we use the interpretation functions for $\zero$, $\suc$, $\nil$, $\cons$
    and $\add$ from Example \ref{ex:revinterpret}.
    Then $\interpret{\abs{\avar\bvar}{(\avar \add \bvar)}} = d,e \mapsto (d_\cost + e_\cost + e_\size + 3, d_\size + e_\size)$.
    We easily see that
    $\mathit{Helper}(\interpret{\abs{\avar\bvar}{(\avar \add \bvar)}},\pair{\aListVar_\cost,\aListVar_\mmax},z) =
        \pair{ z_\cost + \aListVar_\cost + \aListVar_\mmax + 3, z_\size + \aListVar_\mmax }$.  Importantly, the iteration variable
    $z$ is used in a very innocent way: although its size is increased, this increase is by the same number ($\aListVar_\mmax$)
    in every iteration step.  Moreover, the length of $z$ does not affect the evaluation cost.  Hence, we can choose
    $\interpret{\sumList(\aListVar)} = \pair{5 + \aListVar_\cost + \aListVar_\leng + \aListVar_\leng * (\aListVar_\cost + \aListVar_\mmax + 3), \aListVar_\leng * \aListVar_\mmax}$.
    This is close to the interpretation from Example \ref{ex:revinterpret}
    but differs both in a small overhead for the $\beta$-reductions, and
    because our interpretation of $\foldl$ slightly overestimates the true cost.

    This approach can be used to obtain bounds for any function that may be defined in terms of $\foldl$, which
    includes many first-order functions.  For example, with a small change to the signature of $\foldl$, we could let
    $\rev(\aListVar) = \foldl(\abs{\avar \bvar}{(\bvar : \avar)}, \nil, \aListVar)$;
    however, this would necessitate corresponding changes in the interpretation of $\foldl$.
\end{example}

\section{Finding complexity bounds}\label{sec:bounds}

A key notion in complexity analysis of first-order rewriting is \emph{runtime complexity}.  In
this section, we will define a conservative notion of runtime complexity for higher-order term
rewriting, and show how our interpretations can be used to find runtime complexity bounds.

In first-order (and many-sorted) term rewriting, a \emph{defined symbol} is any function symbol
$\afun$ such that there is a rule $\afun(\ell_1,\dots,\ell_k) \to r$ in the system;
all other symbols are called \emph{constructors}.
A \emph{ground constructor term} is a ground term without defined symbols.
A \emph{basic term} has the form $\afun(\aterm_1,\dots,\aterm_k)$ with $\afun$ a defined
symbol and $\aterm_1,\dots,\aterm_k$ all ground constructor terms.
The \emph{runtime complexity} of a TRS is then a function $\varphi$ in $(\N \setminus \{0\}) \arrfunc \N$ that maps each $n$ to
a number $\varphi(n)$ so that for every basic term $s$ of size at most $n$: $\dht{s} \leq \varphi(n)$.

The comparable notion of \emph{derivational complexity} considers the derivation height for
arbitrary ground terms of size $n$, but we will not use that here, since it can often give very
high bounds that are not necessarily representative for realistic use of the system.
In practice, a computation with a TRS would typically start with a main function,
which takes \emph{data} (e.g., natural numbers, lists) as input.
This is exactly a basic term.
Hence, the notion of runtime complexity roughly captures the worst-case number of steps for a realistic computation.

It is not obvious how this notion translates to the higher-order setting.
It may be tempting to literally apply the definition to an AFS,
but a ``ground constructor term'' (or perhaps ``closed constructor term'')
is not a natural concept in higher-order rewriting; it does not intuitively capture data.
Moreover, we would like to create a \emph{robust} notion which can be extended to simple functional programming languages,
so which is not subject to minor language difference like whether partial application of function symbols is allowed.

Instead, there are two obvious ways to capture the idea of input in higher-order rewriting:
\begin{itemize}
    \item \emph{closed irreducible terms}; this includes all ground constructor terms, but also
          for instance $\abs{x}{\zero \add x}$ (but not $\abs{x}{x \add \zero}$, since this can be
          rewritten following the rules in Example \ref{ex:append-reverse});
    \item \emph{data}: this includes only ground constructor terms with no higher-order subterms.
\end{itemize}

As we observed in Example \ref{ex:hor_f}, the size of a higher-order term does not capture its
behaviour.
Hence, a notion of runtime complexity using closed irreducible terms is not obviously meaningful---%
and might be closer to \emph{derivational} complexity due to defined symbols inside abstractions.
Therefore, we here take the conservative choice and consider \emph{data}.
\begin{definition}
    In an AFS $\trs$,
    a \emph{data constructor} is a function symbol $\sAr{c} :: [\asort_1 \times \dots \times \asort_k]
        \arrtype \asort_0$ with each $\asort_i \in \sortset$, such that there is no rule of the form
    $\sAr{c}(\ell_1,\dots,\ell_k) \to r$.
    A \emph{data term} is a term $\sAr{c}(d_1,\dots,d_k)$ such that $\sAr{c}$ is a constructor and all
    $d_i$ are also data terms.
\end{definition}

In practice, a sort is defined by its data constructors.
For example, $\nat$ is defined by $\zero$ and $\suc$, and $\lst$ by $\nil$ and $\cons$.
In typical examples of first- and higher-order term rewriting systems,
rules are defined to exhaustively pattern match on all constructors for a sort.

With this definition, we can conservatively extend the original notion of runtime complexity to be
applicable to both many-sorted and higher-order term rewriting.

\begin{definition}\label{def:runtime}
    A \emph{basic term} is a term of the form $\afun(d_1, \dots, d_k)$ with
    all $d_i$ data terms and $\afun$ not a data constructor.
    We let $|d|$ denote the total number of symbols in a basic term $d$.

    The \emph{runtime complexity} of an AFS is a function
    $\varphi \in (\N \setminus \{0\}) \arrfunc \N$
    so that for all $n$
    and basic terms $d$, with $|d| \leq n$: $\dht{d} \leq \varphi(n)$.
\end{definition}

Note that if $\afun(d_1,\dots,d_k)$ is a basic term, then $\afun :: [\asort_1 \times \dots \times \asort_k]
\arrtype \btype$ with all $\asort_i$ sorts.  Hence,
higher-order runtime complexity considers the same (first-order) notion of basic terms
as the first-order case;
terms such as $\map(\aFuncVar,s)$ or even $\map(\abs{\avar}{\suc(\avar)},\nil)$ are not basic.
One might reasonably question whether such a first-order notion is useful when studying the
complexity of \emph{higher-order} term rewriting.
However, we argue that it is:
runtime complexity aims to address the length of computations that begin at a typical starting point.
When performing a \emph{full program} analysis of 
an AFS,
the computation will still typically start in a basic term, for instance;
the entry-point symbol $\mathsf{main}$ applied to some user input $d_1, \dots, d_k$.

\begin{example}\label{ex:extrec}
We consider an AFS from the Termination Problem Database, v11.0 \cite{tpdb}.
\[
\begin{array}{rclcrcl}
\avar \add \zero & \arrz & \avar & \quad &
\sAr{rec}(\zero,\bvar,\aFuncVar) & \arrz & \bvar \\
\avar \add \suc(\bvar) & \arrz & \suc(\avar \add \bvar) & \quad &
\sAr{rec}(\suc(\avar),\bvar,\aFuncVar) & \arrz & \aFuncVar \cdot \avar \cdot \sAr{rec}(\avar,\bvar,\aFuncVar) \\
& & & &
\avar \otimes \bvar & \arrz & \sAr{rec}(\bvar,\zero,\abs{n}{\abs{m}{\avar \oplus m}}) \\
\end{array}
\]
Here, $\sAr{rec} :: [\nat \times \nat \times (\nat \arrtype \nat \arrtype \nat)] \arrtype \nat$.
The only basic terms have the form $\suc^n(\zero) \add \suc^m(\zero)$ or $\suc^n(\zero) \otimes
\suc^m(\zero)$. Using our method, we obtain cubic runtime complexity; to be precise:
$\mathcal{O}(m^2 * n)$.  The interpretation functions are found in Appendix \ref{app:examples}.
\end{example}

To derive runtime complexity for both first- and higher-order rewriting,
our approach is to consider bounds for the functions $\funcinterpret{\afun}$;
we only need to consider the first-order symbols $\afun$.

\begin{definition}\label{def:linear-bol-int}
    Let $P \in \typeinterpret{\asort_1 \arrtype \dots \arrtype \asort_m \arrtype \bsort}$ be of the form
    $P(x^1,\dots,x^m) = \pair{ P_1(x^1,\dots,x^m), \linebreak \dots, P_{\typecount{\bsort}}(x^1, \dots, x^m) }$.
    Then $P$ is \emph{linearly bounded} if each component function $P_l$ of $P$ is upper-bounded by a positive
    linear polynomial, i.e., there is a constant $a \in \N$ such that
    $P_l(x^1,\dots,x^m) \leq a * (1 + \sum_{i=1}^m\sum_{j=1}^{\typecount{\asort_i}} x^i_j)$.
    We say that $P$ is \emph{additive} if
    there exists a constant $a \in \N$ such that $\sum_{l = 1}^{\typecount{\bsort}} P_l(x^1,\dots,x^m) \leq
    a + \sum_{i=1}^m\sum_{j=1}^{\typecount{\asort_i}} x^i_j$.
\end{definition}

By this definition, $P_l$ is not required to be a linear function, only to be bounded by one.
This means that for instance $\min(x^i_j,2 * x^a_b)$ can be used, but $x^i_j * x^a_b$ cannot.
It is easily checked that all the data constructors in this paper have an additive interpretation.
For example, for $\funcinterpret{\cons}$: $(\avar_\cost + \aListVar_\cost) + (\avar_\leng + 1) + \max(\avar_\size,
\aListVar_\mmax) \leq 1 + \avar_\cost + \avar_\size + \aListVar_\cost + \aListVar_\leng + \aListVar_\size$.

\begin{restatable}{lemma}{boundingIntData}\label{lem:boun-const-int-data}
    Let $\trs$ be an AFS or TRS that is compatible with a strongly monotonic algebra with interpretation
    function $\funcinterpret{}$.
    Then:
    \begin{enumerate}
        \item\label{lem:boun-const-int-data:additive}
              if $\funcinterpret{\sAr{c}}$ is additive for all data constructors $\sAr{c}$, then there exists
              a constant $b > 0$ in $\N$ so that for all data terms $s$: if $|s| \leq n$ then
              $\interpret{s}_l \leq b * n$, for each component $\interpret{s}_l$ of $\interpret{s}$;
        \item\label{lem:boun-const-int-data:linear}
              if $\funcinterpret{\sAr{c}}$ is linearly bounded for all data constructors $\sAr{c}$, then there exists
              a constant $b > 0$ in $\N$ so that for all data terms $s$: if $|s| \leq n$ then $\interpret{s}_l
              \leq 2^{b * n}$, for each component $\interpret{s}_l$ of $\interpret{s}$.
    \end{enumerate}
\end{restatable}

By using Lemma \ref{lem:boun-const-int-data}, we quickly obtain some ways to bound runtime complexity:

\begin{corollary}\label{col:bound-basic-func-pol}
    Let $\trs$ be an AFS or TRS that is compatible with a strongly monotonic algebra with interpretation
    function $\funcinterpret{}$, and let $\signature_C$ denote its set of data constructors, and
    $\signature_B$ the set of all other symbols $\afun$ with a signature $\afun :: [\asort_1\times
    \dots \times \asort_m] \arrtype \btype$.
    Then:
    \begin{itemize}
        \item if $\funcinterpret{\afun}$ is additive for all $\afun \in \signature_C \cup
                  \signature_B$, then $\trs$ has linear runtime complexity;
        \item if $\funcinterpret{\sAr{c}}$ is additive for all $\sAr{c} \in \signature_C$ and
              for all $\afun \in \signature_B$, $\funcinterpret{\afun}(\vec{x}) = (P_1(\vec{x}),
              \dots,P_k(\vec{x}))$ where $P_1$ is bounded by a polynomial, then $\trs$ has
              polynomial runtime complexity;
        \item if $\funcinterpret{\afun}$ is linearly bounded for all $\afun \in \signature_C \cup
                  \signature_B$, then $\trs$ has exponential runtime complexity.
    \end{itemize}
\end{corollary}

We could easily use these results as part of an automatic complexity tool---and indeed, combine
them with other methods for complexity analysis.  However, this is not truly our goal: runtime
complexity is only a part of the picture, especially in higher-order term rewriting where we may
want to analyse modules that get much more hairy input.  Our technique aims to give more fine-grained
information, where we consider the impact of input with certain properties---like the length of a
list or the depth of a tree.  For this, the person interested in the analysis should be the one to
decide on the interpretations of the constructors.

With this information given, though, it should be possible to automatically find interpretations
for the other functions.
The search for the best strategy requires dedicated research,
which we leave to future work; however, we expect Lemmas \ref{lem:wm} and \ref{lem:iterate} to play a
large role.  We also note that while the cost component may depend on the other components, the
other components (which represent a kind of size property) typically do not depend on the cost.

\section{On Related Work}\label{sec:related}

\subparagraph*{Rewriting.}
There are several first-order complexity techniques based on interpretations.
For example, in \cite{bon:cic:mar:tou:98}, the consequences of using additive, linear,
and polynomial interpretations to the natural numbers are investigated;
and in \cite{hof:01}, context-dependent interpretations are introduced, which map terms to real numbers to obtain tighter bounds.
Most closely related to our approach are \emph{matrix interpretations}
\cite{waldmann:zantema:06,mos:sch:wal:08}, and a technique by the first author
for complexity analysis of conditional term rewriting \cite{kop:mid:ste:17}.
In both cases, terms are mapped to tuples as they are in our approach, although neither considers sort information,
and matrix interpretations use linear interpretation functions.
Our technique is a generalisation of both.

\subparagraph*{Higher-order Rewriting.}
In \emph{higher-order} term rewriting (but a formalism without $\lambda$-abstraction),
{Baillot and Dal Lago}~\cite{baillot:lago:16} develop a version of higher-order polynomial interpretations which,
like the present work, is based on v.d.~Pol's higher-order interpretations~\cite{pol:96}.
In similar ways to our Section \ref{sec:bounds}, the authors enforce polynomial bounds on derivational complexity
by imposing restrictions on the shape of interpretations.
Their method differs from ours in various ways, most importantly by mapping terms to $\N$ rather than tuples.
In addition, the interpretations are limited to higher-order polynomials.
This yields an ordering with the subterm property (i.e., $\afun(\dots,s,\dots) \typegr{} s$), which means that
TRSs like Example \ref{ex:quot} cannot be handled.  Moreover, it is not possible to find a general
interpretation for functions like $\foldl$ or $\sAr{rec}$; the method can only handle instances of
$\foldl$ with a linear function.

Beyond this, it unfortunately seems that relatively little work has thus far been done on complexity analysis of higher-order term rewriting.
However, complexity of \emph{functional programs} is an active field of research with a close relation to higher-order term rewriting.

\subparagraph*{Functional Programming.}
There are various techniques to statically analyse resource 
use of functional programs.
The
se
may be fully automated~\cite{ava:lag:17,benzinger:04,sinn:zuleger:veith:14},
semi-automated designed to reason about programmer specified-bounds~\cite{wang:wang:chlipala:17,cicek:garg:deepak:acar:15,handley:vazou:hutton:19},
or even manual techniques, integrated with type system or program logic semantics~\cite{carbonneaux:hoffman:ramananandro:shao:14,lago:gaboardi:11}.
We discuss the most pertinent ones.

An approach using rewriting for full-program analysis is to translate functional programs to TRSs \cite{ava:lag:mos:15},
which can be analysed using first-order complexity techniques.
This takes advantage of the large body of work on first-order complexity, but loses
information; the transformation often yields a system that is harder to analyse than the original.

The research methodology in most studies in functional programming differs significantly from 
rewriting
techniques.
Nevertheless,
there are some studies with clear connections to our approach;
in particular our separation of cost and size (and other structural properties).
Most relevant, in \cite{danner:et-al:15} the authors use a similar approach by giving semantics to a complexity-aware
intermediate language allowing arbitrary user-defined notions for size---such as list length or maximum
element size; recurrence relations are then extracted to represent the 
complexity.

Additionally, most modern complexity analysis is done via enhancements at the type system level
\cite{alves:kesner:ventura:20,ava:lag:17,hoffman:aehlig:hofman:12,rajani:gaboardi:garg:hoffman:21,handley:vazou:hutton:19,das:balzer:hoffman:19}.
For example, types may be annotated with a counter, the heap size or a data type's size measure.
Notably, a line of work on Resource-Aware ML \cite{hoffman:aehlig:hofman:12,niu:hoffman,kahn:hoffman:20}
studies resource use of OCaml programs with methods based on Tarjan's amortized analysis~\cite{tarjan:85}.
Types are annotated with \emph{potentials} (a cost measure), and type inference generates a set of
linear constraints which is sent over to an external solver.
For Haskell, Liquid Haskell~\cite{rondon:kawaguci:jhala:08,vazou:rondon:jhala:13} provides a language to annotate types, which can be
used to prove properties of the program; this was recently extended to include complexity~\cite{handley:vazou:hutton:19}.
Unlike RAML, this approach is not fully automatic: type annotations are checked, not derived.

These works in functional programming have a different purpose from ours: they study the resource
use in a specific language, typically with a fixed evaluation strategy.
Our method, in contrast, allows for arbitrary evaluation, which could be specified to various
strategies in future work.
Moreover, most of these works limit interest to full-program analysis.  We do this for runtime
complexity, but our method offers more, by providing general interpretations for individual functions
like $\map$ or $\foldl$.  Similarly, most of these works impose additive type annotations for the
constructors; we do not restrict the constructor interpretations outside Lemma \ref{lem:boun-const-int-data}.
On the other hand, many do consider (shallow) polymorphism, which we do not.

While in functional programming one considers resource usage~\cite{rajani:gaboardi:garg:hoffman:21,hoffman:aehlig:hofman:12},
rewriting is concerned with the number of steps,
which can be translated to a form of resource measure if the true cost of each step is kept low.
This is achieved by imposing restrictions on reduction strategy and term representation~\cite{accattoli:dal-lago:16,dal-lago:martini:10}.
Our approach carries the blessing of being general and machine independent
and the curse of not necessarily being a reasonable cost model.

\section{Conclusion and Future Work}\label{sec:disc}

In this paper, we have introduced tuple interpretations for many-sorted and higher-order term rewriting.
This includes providing a new definition of strongly monotonic algebras,
a compatibility theorem, a function $\makesm{}$
that orients $\beta$- and $\eta$-reductions,
and several lemmas to prove monotonicity of interpretation functions.
We also show that for certain restrictions on interpretation functions, we find linear,
polynomial or exponential bounds on runtime complexity
(for a simple but natural definition of higher-order runtime complexity).

Our type-based, semantical approach allows us to relate various ``size'' notions (e.g., list length,
tree depth, term size
. etc.) to reduction cost,
and thus offers a more fine-grained analysis 
than traditional notions like runtime complexity.
Most importantly, we can express the complexity of a higher-order function in terms of the behaviour of its
(function) arguments.  In the future, we hope that this 
could be used towards a truly higher-order complexity notion.

\subparagraph{Some further examples and weaknesses.}
Aside from the three higher-order examples in this paper, we have successfully applied our method
to a variety of higher-order benchmarks in the Termination Problem Database \cite{tpdb}, all with
additive interpretations for the constructors.  Two additional examples
(\textsf{filter} and \textsf{deriv}) are included in Appendix \ref{app:examples}.

A clear weakness we discovered was that our method can only handle ``plain function-passing''
systems \cite{kus:sak:07:1}.  That is, we typically do not succeed on systems where a
variable of function type occurs inside a subterm of base type, and occurs outside this subterm
in the right-hand side.
Examples of such systems are \textsf{ordrec}, which has a rule $\sAr{ordrec}(\sAr{lim}(\aFuncVar),
\avar,\bFuncVar,\cFuncVar) \arrz \cFuncVar \cdot \aFuncVar \cdot (\abs{n}{\sAr{ordrec}(\aFuncVar
\cdot n,\avar,\bFuncVar,\cFuncVar)})$ with $\sAr{lim} :: [\nat \arrtype \sTp{ord}] \arrtype
\sTp{ord}$
, and \textsf{apply}, which has a rule $\sAr{lapply}(\avar,\sAr{fcons}(\aFuncVar,
\aListVar)) \arrz \aFuncVar \cdot \sAr{lapply}(\avar,\aListVar)$ with $\sAr{fcons} ::
[(\sTp{a} \arrtype \sTp{a}) \times \sTp{listf}] \arrtype \sTp{listf}$.

\subparagraph{Future work.}
We intend to consider the effect of different evaluation strategies, such
as innermost evaluation, weak-innermost evaluation (where rewriting below an abstraction is not allowed, as
is commonly the case in functional programming) or outermost evaluation.
This extension is likely to be an important step towards another goal: to more closely relate our complexity
notion to a reasonable measure of resource consumption in a rewriting engine.

In addition, we plan to extend first-order complexity techniques like dependency tuples \cite{nao:moser:08},
which may allow us to overcome the weakness described above.
Another goal is to enrich our type system 
to support a notion of polymorphism and
add polymorphic interpretations into the play.
We also aim to develop a tool to automatically find suitable tuple interpretations.

\bibliography{main}

\newpage
\appendix
\markboth{Tuple Interpretations for Higher-Order Rewriting}{C. Kop and D. Vale}

\section{Extended examples}\label{app:examples}

In this appendix, we provide the proof details for the examples that were discussed in the text.
We also include some higher-order examples that were only briefly mentioned in the paper.

In all higher-order examples, we use $\addarg_{\atype,\btype}$ for $\makesm{\atype,\btype}$.

\subsection{Rev/Append}\label{app:rev}
Recall the interpretations given in Example \ref{ex:revinterpret}.  We use $\listvar$ instead of
$\aListVar$ to make the proofs easier to read.
\[
    \begin{array}{rclcrcl}
        \interpret{\zero}                        & = & \pair{0, 0}                                                                                                              &  &
        \interpret{\avar \add \bvar}             & = & \pair{\avar_\cost + \bvar_\cost + \bvar_\size + 1, \avar_\size + \bvar_\size}                                                                                                                       \\
        \interpret{\suc(\avar)}                  & = & \pair{\avar_\cost, \avar_\size + 1}                                                                                      &  &
        \interpret{\sumList(\listvar)}           & = & \pair{\listvar_\cost + 2\listvar_\leng + \listvar_\leng \listvar_\mmax + 1, \listvar_\leng\listvar_\mmax}                                                                                           \\
        \interpret{\nil}                         & = & \pair{0,0,0}                                                                                                             &  &
        \interpret{\rev(\listvar)}               & = & \pair{\listvar_\cost + \listvar_\leng + \listvar_\leng \frac{\listvar_\leng + 1}{2} + 1, \listvar_\leng, \listvar_\mmax}                                                                            \\
        \interpret{\avar : \listvar}             & = & \pair{\avar_\cost + \listvar_\cost, 1 + \listvar_\leng,                                                                  &  &
        \interpret{\append(\listvar, \listvar')} & = & \pair{\listvar_\cost + \listvar_\cost' + \listvar_\leng + 1, \listvar_\leng + \listvar_\leng',                                                                                                      \\
                                                 &   & \phantom{AB}\max(\avar_\size, \listvar_\mmax)}                                                                           &       &   &  & \phantom{ABCDEFGi} \max(\listvar_\mmax, \listvar_\mmax')} \\
    \end{array}
\]
With this interpretation, all rules in Example \ref{ex:append-reverse} are oriented.  First, we show the simple cases:
\[
    \begin{array}{rcccccl}
        \interpret{\avar \add \zero}                        & = & \pair{\avar_\cost + 0 + 0 + 1, \avar_\size}                                 & \sortgr{\nat} & \pair{\avar_\cost,\avar_\size} & = & \interpret{\avar} \\
        \interpret{\sumList(\nil)}                          & = & \pair{0+0+0+1,0*0}                                                          & \sortgr{\nat} & \pair{0,0}                     & = & \interpret{\zero} \\
        \interpret{\rev(\nil)}                              & = & \pair{0 + 0 + 1,0,0}                                                        & \sortgr{\lst} & \pair{0,0,0}                   & = & \interpret{\nil}  \\
        \interpret{\append(\nil,\listvar)}                  & = & \pair{0 + \listvar_\cost + 0 + 1,0 + \listvar_\leng,\max(0,\listvar_\mmax)} & \sortgr{\lst} &
        \pair{\listvar_\cost,\listvar_\leng,\listvar_\mmax} & = & \interpret{\listvar}                                                                                                                                 \\
    \end{array}
\]
As for the cases that require a bit more explanation:
\[
    \begin{array}{rll}
        \interpret{\avar \add \suc(\bvar)}              & =             & \pair{\avar_\cost + \bvar_\cost + (\bvar_\size + 1) + 1,(\avar_\size + \bvar_\size) + 1}                                                                                          \\
                                                        & =             & \pair{\avar_\cost + \bvar_\cost + \bvar_\size + 2,\avar_\size + \bvar_\size + 1}                                                                                                  \\
                                                        & \sortgr{\nat} & \pair{\avar_\cost + \bvar_\cost + 1,\avar_\size + \bvar_\size + 1}                                                                                                                \\
                                                        & =             & \interpret{\suc(\avar \add \bvar)}                                                                                                                                                \\
        \interpret{\sumList(\avar : \listvar)}          & =             & \pair{(\avar_\cost + \listvar_\cost) + 2(\listvar_\leng+1) + (\listvar_\leng +1) * \max(\avar_\size,\listvar_\mmax) + 1,
        (\listvar_\leng + 1) * \max(\avar_\size,\listvar_\mmax)}                                                                                                                                                                                            \\
                                                        & =             & \pair{\avar_\cost + \listvar_\cost + 2\listvar_\leng + \listvar_\leng * \max(\avar_\size,\listvar_\mmax) + \max(\avar_\size,\listvar_\mmax) + 3,                                  \\
                                                        &               & \phantom{AB} \listvar_\leng * \max(\avar_\size,\listvar_\mmax) + \max(\avar_\size,\listvar_\mmax)}                                                                                \\
                                                        & \sortgr{\nat} & \pair{ \avar_\cost + \listvar_\cost + 2\listvar_\leng + \listvar_\leng \listvar_\mmax + \avar_\size + 2,\listvar_\leng \listvar_\mmax + \avar_\size}                              \\
                                                        &               & \text{because}\ \max(\avar_\size,\listvar_\mmax) \geq \listvar_\mmax\ \text{and}\ \max(\avar_\size,\listvar_\mmax) \geq \avar_\size\ \text{and}\ 3 > 2                            \\
                                                        & =             & \pair{ (\listvar_\cost + 2\listvar_\leng + \listvar_\leng \listvar_\mmax + 1) + \avar_\cost + \avar_\size + 1,\listvar_\leng \listvar_\mmax + \avar_\size}                        \\
                                                        & =             & \interpret{\sumList(\listvar) \add \avar}                                                                                                                                         \\
        \interpret{\append(\avar : \listvar,\listvar')} & =             & \pair{(\avar_\cost + \listvar_\cost) + \listvar'_\cost + (1 + \listvar_\leng) + 1,(1 + \listvar_\leng) + \listvar'_\leng,
        \max(\max(\avar_\size,\listvar_\mmax),\listvar'_\mmax)}                                                                                                                                                                                             \\
                                                        & =             & \pair{\avar_\cost + \listvar_\cost + \listvar'_\cost + \listvar_\leng + 2,\listvar_\leng + \listvar'_\leng + 1,
        \max(\avar_\size,\listvar_\mmax,\listvar'_\mmax)}                                                                                                                                                                                                   \\
                                                        & \sortgr{\lst} & \pair{\avar_\cost + (\listvar_\cost + \listvar'_\cost + \listvar_\leng + 1),1 + (\listvar_\leng + \listvar'_\leng),\max(x_\size,\max(\listvar_\mmax,\listvar'_\mmax))}            \\
                                                        & =             & \interpret{\avar : \append(\listvar, \listvar')}                                                                                                                                  \\
        \interpret{\rev(\avar : \listvar)}              & =             & \pair{(\avar_\cost + \listvar_\cost) + (1 + \listvar_\leng) + (1 + \listvar_\leng) * (2 + \listvar_\leng) / 2 + 1,
        1 + \listvar_\leng,\max(\avar_\size,\listvar_\mmax)}                                                                                                                                                                                                \\
                                                        & =             & \pair{\avar_\cost + \listvar_\cost + \listvar_\leng + (1 + \listvar_\leng) * (2 + \listvar_\leng) / 2 + 2,
        1 + \listvar_\leng,\max(\avar_\size,\listvar_\mmax)}                                                                                                                                                                                                \\
                                                        & =             & \pair{\avar_\cost + \listvar_\cost + \listvar_\leng + (1 + \listvar_\leng) + \listvar_\leng * (1 + \listvar_\leng) / 2 + 2,
        1 + \listvar_\leng,\max(\avar_\size,\listvar_\mmax)}                                                                                                                                                                                                \\
                                                        &               & \text{because for all}\ n\ \text{we have:}                                                                                                                                        \\
                                                        &               & \frac{(n + 1)(n+2)}{2} = \frac{2 + 3n + n^2}{2} = 1 + n + \frac{n + n^2}{2} = 1 + n + \frac{n * (1+n)}{2}                                                                         \\
                                                        & =             & \pair{\avar_\cost + \listvar_\cost + 2\listvar_\leng + \listvar_\leng * (1 + \listvar_\leng) / 2 + 3,
        1 + \listvar_\leng,\max(\avar_\size,\listvar_\mmax)}                                                                                                                                                                                                \\
                                                        & \sortgr{\lst} & \pair{\avar_\cost + \listvar_\cost + 2\listvar_\leng + \listvar_\leng (\listvar_\leng + 1)/ 2 + 2,\listvar_\leng + 1,\max(\listvar_\mmax,\avar_\size)}                            \\
                                                        & =             & \pair{(\listvar_\cost + \listvar_\leng + \listvar_\leng \frac{\listvar_\leng + 1}{2} + 1) + \avar_\cost + \listvar_\leng + 1,\listvar_\leng + 1,\max(\listvar_\mmax,\avar_\size)} \\
                                                        & =             & \pair{\interpret{\rev(\listvar)}_\cost + \interpret{\avar : \nil}_\cost + \interpret{\rev(\listvar)}_\leng + 1,                                                                   \\
                                                        &               & \phantom{AB} \interpret{\rev(\listvar)}_\leng+ \interpret{\avar : \nil}_\leng,\max(\interpret{\rev(\listvar)}_\mmax,\interpret{\avar : \nil}_\mmax)}                              \\
                                                        & =             & \interpret{\append(\rev(\listvar), \avar : \nil)}                                                                                                                                 \\
    \end{array}
\]

\subsection{Quot/minus}

The full TRS for division in Example \ref{ex:quot} is:
\[
  \begin{array}{rclcrcl}
  \minus(x, \zero) & \to & x & \quad & \minus(\suc(x), \suc(y)) & \to & \minus(x,y) \\
  \quot(\zero, \suc(y)) & \to & \zero & & \quot(\suc(x), \suc(y)) & \to & \suc(\quot(\minus(x,y), \suc(y))) \\
  \end{array}
\]

Recall the interpretations we used:
    \[
        \begin{array}{rclcrcl}
            \interpret{\zero}        & = & \pair{0,0}                                                                            & \quad &
            \interpret{\minus(x, y)} & = & \pair{x_\cost + y_\cost + y_\size + 1, x_\size}                                                 \\
            \interpret{\suc(x)}      & = & \pair{x_\cost,x_\size+1}                                                              & \quad &
            \interpret{\quot(x,y)}   & = & \pair{x_\cost + x_\size + y_\cost + x_\size * y_\cost + x_\size * y_\size + 1, x_\size}           \\
        \end{array}
    \]

Then:
\[
\begin{array}{rcccccl}
\interpret{\minus(x, \zero)} & = & \pair{x_\cost + 1,x_\size} & \sortgr{\nat} & \pair{x_\cost,x_\size} & = & \interpret{x} \\
\interpret{\minus(\suc(x), \suc(y))} & = & \pair{x_\cost + y_\cost + (y_\size + 1) + 1,x_\size} & \sortgr{\nat} & \pair{x_\cost + y_\cost + y_\size + 1, x_\size} \\
\end{array}
\]

And:
\[
\begin{array}{rll}
\interpret{\minus(\suc(x), \suc(y))}
  & = & \pair{x_\cost + y_\cost + (y_\size + 1) + 1,x_\size} \\
  & > & \pair{x_\cost + y_\cost + y_\size + 1, x_\size} \\
  & = & \interpret{\minus(x,y)} \\
\interpret{\quot(\suc(x), \suc(y))}
  & = & \pair{x_\cost + (x_\size+1) + y_\cost + (x_\size+1) * y_\cost + (x_\size+1) * (y_\size + 1) + 1, x_\size+1} \\
  & = & \pair{x_\cost + x_\size + 1 + y_\cost + x_\size * y_\cost + y_\cost + x_\size * (y_\size + 1) + y_\size + 1 + 1, x_\size+1} \\
  & > & \pair{ (x_\cost + y_\cost + y_\size + 1) + x_\size + y_\cost + x_\size * y_\cost + x_\size * (y_\size + 1) + 1,  x_\size + 1} \\
  & = & \pair{ \interpret{\quot(\minus(x,y), \suc(y))}_\cost, \interpret{\quot(\minus(x,y), \suc(y))}_\size + 1} \\
  & = & \interpret{\suc(\quot(\minus(x,y), \suc(y)))} \\
\end{array}
\]

\subsection{Extrec}

Recall the system in Example \ref{ex:extrec}:
\[
    \begin{array}{rclcrcl}
        \avar \add \zero                       & \arrz & \avar                                                        & \quad &
        \sAr{rec}(\zero,\bvar,\aFuncVar)       & \arrz & \bvar                                                                  \\
        \avar \add \suc(\bvar)                 & \arrz & \suc(\avar \add \bvar)                                       & \quad &
        \sAr{rec}(\suc(\avar),\bvar,\aFuncVar) & \arrz & \aFuncVar \cdot \avar \cdot \sAr{rec}(\avar,\bvar,\aFuncVar)           \\
                                               &       &                                                              &       &
        \avar \otimes \bvar                    & \arrz & \sAr{rec}(\bvar,\zero,\abs{n}{\abs{m}{\avar \oplus m}})                \\
    \end{array}
\]
With $\sAr{rec} :: [\nat \times \nat \times (\nat \arrtype \nat \arrtype \nat)] \arrtype \nat$.
We let $\typeinterpret{\nat} = \N^2$ as before, and let:
\[
    \begin{array}{rcl}
        \interpret{\zero}                            & = & \pair{0,0}                                                                                                            \\
        \interpret{\suc(\avar)}                      & = & \pair{\avar_\cost,\avar_\size+1}                                                                                      \\
        \interpret{\avar \add \bvar}                 & = & \pair{\avar_\cost + \bvar_\cost + \bvar_\size + 1,\avar_\size + \bvar_\size}                                          \\
        \interpret{\avar \otimes \bvar}              & = & \pair{1 + \bvar_\size * (\avar_\cost + \bvar_\cost + \avar_\size * (\bvar_\size+1)/2 + 3), \avar_\size * \bvar_\size} \\
        \interpret{\sAr{rec}(\avar,\bvar,\aFuncVar)} & = & \mathit{Helper}(\avar,\aFuncVar)^{\avar_\size}(\
        \pair{1 + \avar_\cost + \bvar_\cost + \avar_\size + \aFuncVar(\nul{\nat},\nul{\nat})_\cost, \bvar_\size}\ )                                                              \\
        \mathit{Helper}(\avar,\aFuncVar)             & = & \cvar \mapsto \pair{\aFuncVar(\avar,\cvar)_\cost, \max(\cvar_\size,\aFuncVar(\avar,\cvar)_\size)}                     \\
    \end{array}
\]
Then we always have (*A) $\mathit{Helper}(\avar,\aFuncVar)(\cvar) \typegeq{\nat} \cvar$ because
$\aFuncVar(\avar,\cvar)_\cost \geq \cvar_\cost$ which we will see in Lemma \ref{lem:argumentincrease},
and clearly $\max(\cvar_\size,\aFuncVar(\avar,\cvar)_\size) \geq \cvar_\size$.
Hence, the monotonicity requirements are satisfied.
We also clearly have (*B) $\mathit{Helper}(\avar,\aFuncVar)(\cvar) \typegeq{\nat} \aFuncVar(\avar,\cvar)$,
since clearly $\max(\cvar_\size,\aFuncVar(\avar,\cvar)_\size) \geq \aFuncVar(\avar,\cvar)_\size$.
We have:
\begin{itemize}
    \item $\interpret{\avar \add \zero} \typegr{\nat} \interpret{\avar}$: \\
          $\interpret{\avar \add \zero} = \pair{\avar_\cost + 0 + 0 + 1, \avar_\size + 1} = \pair{\avar_\cost + 1,\avar_\size} > \pair{\avar_\cost,\avar_\size} = \interpret{\avar}$.
    \item $\interpret{\avar \add \suc(\bvar)} \typegr{\nat} \interpret{\suc(\avar \add \bvar)}$: \\
          $\interpret{\avar \add \suc(\bvar)} = \pair{\avar_\cost + \bvar_\cost + (\bvar_\size + 1) + 1,\avar_\size + (\bvar_\size + 1)}
              > \pair{\avar_\cost + \bvar_\cost + \bvar_\size + 1,\avar_\size + \bvar_\size + 1} = \interpret{\suc(\avar \add \bvar)}$
    \item $\interpret{\sAr{rec}(\zero,\bvar,\aFuncVar)} \typegr{\nat} \interpret{\bvar}$: \\
          $\interpret{\sAr{rec}(\zero,\bvar,\aFuncVar)} =
              \mathit{Helper}(\pair{0,0},\aFuncVar)^0(\ \pair{1 + 0 + \bvar_\cost + 0 + \aFuncVar(\nul{\nat},\nul{\nat})_\cost, \bvar_\size}\ ) =
              \pair{1 + \bvar_\cost + \aFuncVar(\nul{\nat},\nul{\nat})_\cost, \bvar_\size} \typegr{\nat}
              \pair{\bvar_\cost,\bvar_\size} = \interpret{\bvar}$.
    \item $\interpret{\sAr{rec}(\suc(\avar),\bvar,\aFuncVar)} \typegr{\nat} \interpret{\aFuncVar \cdot \avar \cdot \sAr{rec}(\avar,\bvar,\aFuncVar)}$: \\
          $\interpret{\sAr{rec}(\suc(\avar),\bvar,\aFuncVar)} =
              \mathit{Helper}(\pair{\avar_\cost,\avar_\size+1},\aFuncVar)^{\avar_\size+1}(\
              \pair{1 + \avar_\cost + \bvar_\cost + (\avar_\size + 1) + \aFuncVar(\nul{\nat},\nul{\nat})_\cost, \bvar_\size}\ ) =
              \mathit{Helper}(\pair{\avar_\cost,\avar_\size+1},\aFuncVar)(\mathit{Helper}(\pair{\avar_\cost,\avar_\size+1},\aFuncVar)^{\avar_\size}(\
              \pair{2 + \avar_\cost + \bvar_\cost + \avar_\size + \aFuncVar(\nul{\nat},\nul{\nat})_\cost, \bvar_\size}\ )) \typegeq{\nat}
              \aFuncVar(\pair{\avar_\cost,\avar_\size+1},\mathit{Helper}(\pair{\avar_\cost,\avar_\size+1},\aFuncVar)^{\avar_\size}(\
              \pair{2 + \avar_\cost + \bvar_\cost + \avar_\size + \aFuncVar(\nul{\nat},\nul{\nat})_\cost, \bvar_\size}\ ))$ by (*B), $\typegr{\nat}
              \aFuncVar(\pair{\avar_\cost,\avar_\size},\mathit{Helper}(\pair{\avar_\cost,\avar_\size},\aFuncVar)^{\avar_\size}(\
              \pair{1 + \avar_\cost + \bvar_\cost + \avar_\size + \aFuncVar(\nul{\nat},\nul{\nat})_\cost, \bvar_\size}\ ))$ by monotonicity, $=
              \aFuncVar(\avar,\mathit{Helper}(\avar,\aFuncVar)^{\avar_\size}(\
              \pair{1 + \avar_\cost + \bvar_\cost + \avar_\size + \aFuncVar(\nul{\nat},\nul{\nat})_\cost, \bvar_\size}\ ) )
              = \interpret{\aFuncVar \cdot \avar \cdot \sAr{rec}(\avar,\bvar,\aFuncVar)}$.
    \item $\interpret{\avar \otimes \bvar} \typegr{\nat} \interpret{\sAr{rec}(\bvar,\zero,\abs{n}{\abs{m}{\avar \oplus m}})}$:
          \begin{itemize}
              \item $\interpret{\abs{n}{\abs{m}{\avar \oplus m}}} = n \mapsto m \mapsto \pair{\avar_\cost + n_\cost + m_\cost + m_\size + 3, \avar_\size + m_\size}$
              \item $\mathit{Helper}(\bvar,\interpret{\abs{n}{\abs{m}{\avar \oplus m}}}) = m \mapsto \pair{\avar_\cost + \bvar_\cost + m_\cost + m_\size + 3, \avar_\size + m_\size}$
              \item For given $i$,
                    $\mathit{Helper}(\bvar,\interpret{\abs{n}{\abs{m}{\avar \oplus m}}})^i(m)_\size = (\sum_{j=0}^i \avar_\size) + m_\size = \avar_\size * i + m_\size$
              \item $\mathit{Helper}(\bvar,\interpret{\abs{n}{\abs{m}{\avar \oplus m}}})^{\bvar_\size} = m \mapsto
                        \pair{\sum_{i=1}^{\bvar_\size}(\avar_\cost + \bvar_\cost + (\avar_\size * i + m_\size) + 3) + m_\cost, \bvar_\size * \avar_\size + m_\size} =
                        \pair{\bvar_\size * (\avar_\cost + \bvar_\cost + m_\size + 3) + \avar_\size * \sum_{i=1}^{\bvar_\size}(i) + m_\cost, \bvar_\size * \avar_\size + m_\size} =
                        \pair{\bvar_\size * (\avar_\cost + \bvar_\cost + m_\size + 3) + \avar_\size * (\bvar_\size * (\bvar_\size + 1) / 2) + m_\cost, \bvar_\size * \avar_\size + m_\size} =
                        \pair{\bvar_\size * (\avar_\cost + \bvar_\cost + m_\size + \avar_\size * (\bvar_\size+1)/2 + 3) + m_\cost, \bvar_\size * \avar_\size + m_\size}$
          \end{itemize}
          Hence, $\interpret{\avar \otimes \bvar} = \pair{1 + \bvar_\size * (\avar_\cost + \bvar_\cost + \avar_\size * (\bvar_\size+1)/2 + 3), \avar_\size * \bvar_\size} \typegr{\nat}
              \pair{\bvar_\size * (\avar_\cost + \bvar_\cost + \avar_\size * (\bvar_\size+1)/2 + 3) + 0, \avar_\size * \bvar_\size + 0} =
              \interpret{\sAr{rec}(\bvar,\zero,\abs{n}{\abs{m}{\avar \oplus m}})}$
\end{itemize}

\subsection{Filter}

The next example also comes from the Termination Problem Database, version 11.0 \cite{tpdb}.
This example was only briefly mentioned in the text, but included here to demonstrate that our
method can handle many typical examples of higher-order term rewriting systems.
\[
    \begin{array}{rclcrcl}
        \sAr{rand}(\avar)                         & \arrz & \avar                                                                       & \quad &
        \sAr{filter}(\aFuncVar,\nil)              & \arrz & \nil                                                                                  \\
        \sAr{rand}(\suc(\avar))                   & \arrz & \sAr{rand}(\avar)                                                           & \quad &
        \sAr{filter}(\aFuncVar,\avar : \aListVar) & \arrz & \sAr{consif}(\aFuncVar \cdot \avar,\avar,\sAr{filter}(\aFuncVar,\aListVar))           \\
        \sAr{bool}(\zero)                         & \arrz & \sAr{false}                                                                 &       &
        \sAr{consif}(\sAr{true},\avar,\aListVar)  & \arrz & \avar : \aListVar                                                                     \\
        \sAr{bool}(\suc(\zero))                   & \arrz & \sAr{true}                                                                  &       &
        \sAr{consif}(\sAr{false},\avar,\aListVar) & \arrz & \aListVar                                                                             \\
    \end{array}
\]
As we did in \seclong\ref{app:rev}, we will use the notation $\listvar$ instead of $\aListVar$ to
avoid clutter in the proof.
We let $\typeinterpret{\nat} = \N^2$ and $\typeinterpret{\lst} = \N^3$ as before, and additionally
let $\typeinterpret{\sTp{boolean}} = \N$ (so only a cost component and no size components).  We let:
\[
    \begin{array}{rclcrclcrclcrcl}
        \interpret{\sAr{true}}                                              & = & \pair{0}                                                                                                       &  &
        \interpret{\suc(\avar)}                                             & = & \pair{\avar_\cost,\avar_\size+1}                                                                               &  &
        \interpret{\sAr{bool}(\avar)}                                       & = & \pair{\avar_\cost + 1}                                                                                              \\
        \interpret{\sAr{false}}                                             & = & \pair{0}                                                                                                       &  &
        \interpret{\nil}                                                    & = & \pair{0,0,0}                                                                                                   &  &
        \interpret{\sAr{rand}(\avar)}                                       & = & \pair{1 + \avar_\cost + \avar_\size,\avar_\size}                                                                    \\
        \interpret{\zero}                                                   & = & \pair{0,0}                                                                                                     &  &
        \interpret{\avar : \listvar}                                       & = & \multicolumn{6}{l}{\pair{\avar_\cost + \listvar_\cost,\listvar_\leng + 1,\max(\avar_\size,\listvar_\mmax)}}      \\
        \multicolumn{3}{r}{\interpret{\sAr{consif}(\cvar,\avar,\listvar)}} & = &
        \multicolumn{8}{l}{\pair{\cvar_\cost + \avar_\cost + \listvar_\cost + 1,\listvar_\leng + 1,\max(\avar_\size,\listvar_\mmax)}}                                                              \\
        \multicolumn{3}{r}{\interpret{\sAr{filter}(\aFuncVar,\listvar)}}   & = &
        \multicolumn{8}{l}{\pair{1 + (\listvar_\leng + 1) * (2 + \listvar_\cost + \aFuncVar(\pair{\listvar_\cost,\listvar_\mmax})_\cost ), \listvar_\leng,\listvar_\mmax}}.                      \\
    \end{array}
\]

It is easy to see that monotonicity requirements are satisfied. We have:

\begin{itemize}
    \item $\interpret{\sAr{rand}(\avar)} \typegr{\nat} \interpret{\avar}$ \\
          $\interpret{\sAr{rand}(\avar)} = \pair{1 + \avar_\cost + \avar_\size,\avar_\size} \typegr{\nat} \pair{\avar_\cost,\avar_\size} = \interpret{\avar}$
    \item $\interpret{\sAr{rand}(\suc(\avar))} \typegr{\nat} \interpret{\sAr{rand}(\avar)}$ \\
          $\interpret{\sAr{rand}(\suc(\avar))} = \pair{1 + \avar_\cost + \avar_\size + 1,\avar_\size + 1} \typegr{\nat} \pair{1 + \avar_\cost + \avar_\size,\avar_\size} =
              \interpret{\sAr{rand}(\avar)}$
    \item $\interpret{\sAr{bool}(\zero)} \typegr{\sTp{boolean}} \interpret{\sAr{false}}$ \\
          $\interpret{\sAr{bool}(\zero)} = \pair{0 + 1} \typegr{\sTp{boolean}} \pair{0} = \interpret{\sAr{false}}$
    \item $\interpret{\sAr{bool}(\suc(\zero))} \typegr{\sTp{boolean}} \interpret{\sAr{true}}$ \\
          $\interpret{\sAr{bool}(\zero)} = \pair{0 + 1} \typegr{\sTp{boolean}} \pair{0} = \interpret{\sAr{true}}$
    \item $\interpret{\sAr{consif}(\sAr{true},\avar,\listvar)} \typegr{\lst} \interpret{\avar : \listvar}$ \\
          $\interpret{\sAr{consif}(\sAr{true},\avar,\listvar)} = \pair{0 + \avar_\cost + \listvar_\cost + 1,\listvar_\leng + 1,\max(\avar_\size,\listvar_\mmax)}
              \typegr{\lst} \pair{\avar_\cost + \listvar_\cost,\listvar_\leng + 1,\max(\avar_\size,\listvar_\mmax)} = \interpret{\avar : \listvar}$
    \item $\interpret{\sAr{consif}(\sAr{false},\avar,\listvar)} \typegr{\lst} \interpret{\listvar}$ \\
          $\interpret{\sAr{consif}(\sAr{false},\avar,\listvar)} = \pair{0 + \avar_\cost + \listvar_\cost + 1,\listvar_\leng + 1,\max(\avar_\size,\listvar_\mmax)}
              \typegr{\lst} \pair{\listvar_\cost,\listvar_\leng,\listvar_\mmax} = \interpret{\listvar}$
    \item $\interpret{\sAr{filter}(\aFuncVar,\nil)} \typegr{\lst} \interpret{\nil}$ \\
          $\interpret{\sAr{filter}(\aFuncVar,\nil)} = \pair{1 + \dots,0,0} \typegr{\lst} \pair{0,0,0} = \interpret{\nil}$
    \item $\interpret{\sAr{filter}(\aFuncVar,\avar : \listvar)} \typegr{\lst} \interpret{\sAr{consif}(\aFuncVar \cdot \avar,\avar,\sAr{filter}(\aFuncVar,\listvar))}$ \\
          $\interpret{\sAr{filter}(\aFuncVar,\avar : \listvar)} =
              \pair{1 + (\listvar_\leng + 2) * (2 + \avar_\cost + \listvar_\cost + \aFuncVar(\pair{\avar_\cost + \listvar_\cost,\max(\avar_\size,\listvar_\mmax)})_\cost),
                  \listvar_\leng+1,\max(\avar_\size,\listvar_\mmax)} =
              \pair{3 + \avar_\cost + \listvar_\cost + \aFuncVar(\pair{\avar_\cost + \listvar_\cost,\max(\avar_\size,\listvar_\mmax)})_\cost +
                  (\listvar_\leng + 1) * (2 + \avar_\cost + \listvar_\cost + \aFuncVar(\pair{\avar_\cost + \listvar_\cost,\max(\avar_\size,\listvar_\mmax)})_\cost),
                  \listvar_\leng+1,\max(\avar_\size,\listvar_\mmax)} \typegr{\lst} \\
              \pair{2 + \avar_\cost + \aFuncVar(\pair{\avar_\cost,\avar_\size})_\cost +
                  (\listvar_\leng + 1) * (2 + \listvar_\cost + \aFuncVar(\pair{\listvar_\cost,\listvar_\mmax})_\cost),
                  \listvar_\leng+1,\max(\avar_\size,\listvar_\mmax)} = \\
              \pair{\aFuncVar(\avar)_\cost + \avar_\cost + (1 + (\listvar_\leng + 1) * (2 + \listvar_\cost + \aFuncVar(\pair{\listvar_\cost,\listvar_\mmax})_\cost )) + 1,
                  \listvar_\leng + 1,\max(\avar_\size,\listvar_\mmax)} \\
              = \pair{\aFuncVar(\avar)_\cost + \avar_\cost + \interpret{\sAr{filter}(\aFuncVar,\listvar)}_\cost + 1,
                  \interpret{\sAr{filter}(\aFuncVar,\listvar)}_\leng + 1,\max(\avar_\size,\interpret{\sAr{filter}(\aFuncVar,\listvar)}_\mmax)} \\
              = \interpret{\sAr{consif}(\aFuncVar \cdot \avar,\avar,\sAr{filter}(\aFuncVar,\listvar))}$
\end{itemize}

\subsection{Deriv}

Our final example also comes from the termination problem database.
This example seems to be designed to calculate a function's derivative.
It is worth noting that all symbols other than $\sAr{der}$ are constructors.

\[
    \begin{array}{rclcrcl}
        \sAr{der}(\abs{\avar}{\bvar})                                                    & \arrz & \abs{\cvar}{\zero}                                                                                 & \quad\quad &
        \sAr{der}(\abs{\avar}{\sAr{sin}(\avar)})                                         & \arrz & \abs{\cvar}{\sAr{cos}(\cvar)}                                                                                     \\
        \sAr{der}(\abs{\avar}{\avar})                                                    & \arrz & \abs{\cvar}{\sAr{1}}                                                                               & \quad\quad &
        \sAr{der}(\abs{\avar}{\sAr{cos}(\avar)})                                         & \arrz & \abs{\cvar}{\sAr{min}(\sAr{cos}(\cvar))}                                                                          \\
        \sAr{der}(\abs{\avar}{\sAr{plus}(\aFuncVar \cdot \avar,\bFuncVar \cdot \avar)})  & \arrz &
        \multicolumn{5}{l}{\abs{\cvar}{\sAr{plus}(\sAr{der}(\aFuncVar) \cdot \cvar,\sAr{der}(\bFuncVar) \cdot \cvar)}}                                                                                               \\
        \sAr{der}(\abs{\avar}{\sAr{times}(\aFuncVar \cdot \avar,\bFuncVar \cdot \avar)}) & \arrz &
        \multicolumn{5}{l}{\abs{\cvar}{\sAr{plus}(\sAr{times}(\sAr{der}(\aFuncVar) \cdot \cvar,\bFuncVar \cdot \cvar),\sAr{times}(\aFuncVar \cdot \cvar,\sAr{der}(\bFuncVar) \cdot \cvar))}}                         \\
        \sAr{der}(\abs{\avar}{\sAr{ln}(\aFuncVar \cdot \avar)})                          & \arrz & \multicolumn{5}{l}{\abs{\cvar}{\sAr{div}(\sAr{der}(\aFuncVar) \cdot \cvar,\aFuncVar \cdot \cvar)}}                \\
    \end{array}
\]
With $\sAr{der} :: [\real \arrtype \real] \arrtype \real \arrtype \real$.
We let $\typeinterpret{\real} = \N^3$ where the first component indicates cost, and the
second and third component roughly indicate the number of plus/times/ln occurrences and
the number of times/ln occurrences respectively.
We will denote $\avar_\size$ for $\avar_2$, and $\avar_\star$ for $\avar_3$.
We use the following interpretation:
\[
    \begin{array}{rclcrcl}
        \interpret{\zero}                                  & = & \pair{0,0,0}                                                                                                          &  &
        \hspace{-20pt}\interpret{\sAr{plus}(\avar,\bvar)}  & = & \pair{\avar_\cost + \bvar_\cost,\avar_\size + \bvar_\size + 1,\avar_\star + \bvar_\star}                                                                                                                    \\
        \interpret{\sAr{1}}                                & = & \pair{0,0,0}                                                                                                          &  &
        \hspace{-20pt}\interpret{\sAr{times}(\avar,\bvar)} & = & \pair{\avar_\cost + \bvar_\cost,\avar_\size + \bvar_\size + 1,\avar_\star + \bvar_\star + 1}                                                                                                                \\
        \interpret{\sAr{cos}(\avar)}                       & = & \avar                                                                                                                 &  &
        \interpret{\sAr{ln}(\avar)}                        & = & \pair{\avar_\cost, \avar_\size + 1, \avar_\star + 1}                                                                                                                                                        \\
        \interpret{\sAr{sin}(\avar)}                       & = & \avar                                                                                                                 &  &
        \interpret{\sAr{der}(\aFuncVar)}                   & = & \cvar \mapsto \langle                                                                                                                                                                                       \\
        \interpret{\sAr{min}(\avar)}                       & = & \pair{\avar_\cost,0,0}                                                                                                &  &
                                                           &   & \phantom{A}1 + \aFuncVar(\cvar)_\cost + 2 * \aFuncVar(\cvar)_\size + \aFuncVar(\cvar)_\star * \aFuncVar(\cvar)_\cost,                                                                                       \\
        \interpret{\sAr{div}(\avar,\bvar)}                 & = & \pair{\avar_\cost+\bvar_\cost,0,0}                                                                                    &  &
                                                           &   & \phantom{A}\aFuncVar(\cvar)_\size * (\aFuncVar(\cvar)_\star + 1),                                                                                                                                           \\
                                                           &   &                                                                                                                       &  &   &  & \phantom{A}\aFuncVar(\cvar)_\star * (\aFuncVar(\cvar)_\star + 1)\ \rangle \\
    \end{array}
\]
It is easy to see that monotonicity requirements are satisfied.
In addition, all the rules are oriented by this interpretation:

\begin{itemize}
    \item $\interpret{\sAr{der}(\abs{\avar}{\bvar})} \typegr{\real} \interpret{\abs{\cvar}{\zero}}$
          \begin{itemize}
              \item Note that by choice of $\makesm{\real,\real}$, we have $\interpret{\abs{\avar}{\bvar}} = \avar \mapsto \pair{1 + \avar_\cost + \bvar_\cost,\bvar_\size,\bvar_\star}$.
          \end{itemize}
          $\interpret{\sAr{der}(\abs{\avar}{\bvar})} = \cvar \mapsto
              \pair{1 + (1 + \cvar_\cost + \bvar_\cost) + 2 * \bvar_\size + \bvar_\star * (1 + \cvar_\cost + \bvar_\cost), \bvar_\size * (\bvar_\star + 1), \bvar_\star * (\bvar_\star + 1)}
              \typegr{\real}
              \cvar \mapsto \pair{1 + \cvar_\cost,0,0} = \interpret{\abs{\cvar}{\zero}}$
    \item $\interpret{\sAr{der}(\abs{\avar}{\bvar})} \typegr{\real} \interpret{\abs{\cvar}{\zero}}$
          \begin{itemize}
              \item Note that by choice of $\makesm{\real,\real}$, we have $\interpret{\abs{\avar}{\avar}} = \avar \mapsto \pair{1 + \avar_\cost,\avar_\size,\avar_\star}$.
          \end{itemize}
          $\interpret{\sAr{der}(\abs{\avar}{\avar})} = \cvar \mapsto
              \pair{1 + (1 + \cvar_\cost) + \dots, \cvar_\size * (\cvar_\star + 1), \cvar_\star * (\cvar_\star + 1)}
              \typegr{\real}
              \cvar \mapsto \pair{1 + \cvar_\cost,0,0} = \interpret{\abs{\cvar}{\sAr{1}}}$
    \item $\interpret{\sAr{der}(\abs{\avar}{\sAr{sin}(\avar)})} \typegr{\real} \interpret{\abs{\cvar}{\sAr{cos}(\cvar)}}$
          \begin{itemize}
              \item Note that $\interpret{\abs{\avar}{\sAr{sin}(\avar)}} = \avar \mapsto \pair{1 + \avar_\cost,\avar_\size,\avar_\star}$
          \end{itemize}
          $\interpret{\sAr{der}(\abs{\avar}{\sAr{sin}(\avar)})} =
              \cvar \mapsto \pair{1 + (1 + \cvar_\cost) + 2 * \cvar_\size + \cvar_\star * (1 + \cvar_\cost), \cvar_\size * (\cvar_\star + 1), \cvar_\star * (\cvar_\star + 1)} \typegr{\real}
              \cvar \mapsto \pair{1 + \cvar_\cost,\cvar_\size,\cvar_\star} = \interpret{\abs{\cvar}{\sAr{cos}(\cvar)}}$
    \item $\interpret{\sAr{der}(\abs{\avar}{\sAr{cos}(\avar)})} \typegr{\real} \interpret{\abs{\cvar}{\sAr{min}(\sAr{cos}(\cvar))}}$ \\
          $\interpret{\sAr{der}(\abs{\avar}{\sAr{cos}(\avar)})} = \cvar \mapsto \pair{1 + (1 + \cvar_\cost) + \dots,\cvar_\size * (\cvar_\star + 1), \cvar_\star * (\cvar_\star + 1)}
              \typegr{\real} \cvar \mapsto \pair{1 + \cvar_\cost,0,0} = \interpret{\abs{\cvar}{\sAr{min}(\sAr{cos}(\cvar))}}$.
    \item $\interpret{\sAr{der}(\abs{\avar}{\sAr{plus}(\aFuncVar \cdot \avar,\bFuncVar \cdot \avar)})} \typegr{\real}
              \interpret{\abs{\cvar}{\sAr{plus}(\sAr{der}(\aFuncVar) \cdot \cvar,\sAr{der}(\bFuncVar) \cdot \cvar)}}$
          \begin{itemize}
              \item $\interpret{\abs{\avar}{\sAr{plus}(\aFuncVar \cdot \avar,\bFuncVar \cdot \avar)}} = \avar \mapsto \pair{1 +
                            \aFuncVar(\avar)_\cost + \bFuncVar(\avar)_\cost,\aFuncVar(\avar)_\size + \bFuncVar(\avar)_\size + 1,\aFuncVar(\avar)_\star + \bFuncVar(\avar)_\star}$
          \end{itemize}
          $\interpret{\sAr{der}(\abs{\avar}{\sAr{plus}(\aFuncVar \cdot \avar,\bFuncVar \cdot \avar)})} =
              \cvar \mapsto \pair{1 + (1 + \aFuncVar(\cvar)_\cost + \bFuncVar(\cvar)_\cost) + 2 * (\aFuncVar(\cvar)_\size + \bFuncVar(\cvar)_\size + 1) +
                  (\aFuncVar(\cvar)_\star + \bFuncVar(\cvar)_\star) * (1 + \aFuncVar(\cvar)_\cost + \bFuncVar(\cvar)_\cost),
                  (\aFuncVar(\cvar)_\size + \bFuncVar(\cvar)_\size + 1) * (\aFuncVar(\cvar)_\star + \bFuncVar(\cvar)_\star + 1),
                  (\aFuncVar(\cvar)_\star + \bFuncVar(\cvar)_\star) * (\aFuncVar(\cvar)_\star + \bFuncVar(\cvar)_\star + 1)} \typegr{\real} \\
              \cvar \mapsto \pair{1 + \aFuncVar(\cvar)_\cost + \bFuncVar(\cvar)_\cost + 2 * \aFuncVar(\cvar)_\size + 2 * \bFuncVar(\cvar)_\size + 2 +
                  \aFuncVar(\cvar)_\star * \aFuncVar(\cvar)_\cost + \bFuncVar(\cvar)_\star * \bFuncVar(\cvar)_\cost),
                  \aFuncVar(\cvar)_\size * (\aFuncVar(\cvar)_\star + 1) + \bFuncVar(\cvar)_\size * (\bFuncVar(\cvar)_\star + 1) + 1,
                  \aFuncVar(\cvar)_\star * (\aFuncVar(\cvar)_\star + 1) + \bFuncVar(\cvar)_\star * (\bFuncVar(\cvar)_\star + 1)} =\\
              \cvar \mapsto \pair{1 + (1 + \aFuncVar(\cvar)_\cost + 2 * \aFuncVar(\cvar)_\size + \aFuncVar(\cvar)_\star * \aFuncVar(\cvar)_\cost) +
                  (1 + \bFuncVar(\cvar)_\cost + 2 * \bFuncVar(\cvar)_\size + \bFuncVar(\cvar)_\star * \bFuncVar(\cvar)_\cost),
                  \interpret{\sAr{plus}(\sAr{der}(\aFuncVar) \cdot \cvar,\sAr{der}(\bFuncVar) \cdot \cvar)}_\size,
                  \interpret{\sAr{plus}(\sAr{der}(\aFuncVar) \cdot \cvar,\sAr{der}(\bFuncVar) \cdot \cvar)}_\star} =
              \interpret{\abs{\cvar}{\sAr{plus}(\sAr{der}(\aFuncVar) \cdot \cvar,\sAr{der}(\bFuncVar) \cdot \cvar)}}$
    \item $\interpret{\sAr{der}(\abs{\avar}{\sAr{times}(\aFuncVar \cdot \avar,\bFuncVar \cdot \avar)})} \typegr{\real}
              \interpret{\abs{\cvar}{\sAr{plus}(\sAr{times}(\sAr{der}(\aFuncVar) \cdot \cvar,\bFuncVar \cdot \cvar),\sAr{times}(\aFuncVar \cdot \cvar,\sAr{der}(\bFuncVar) \cdot \cvar))}}$
          \begin{itemize}
              \item $\interpret{\abs{\avar}{\sAr{times}(\aFuncVar \cdot \avar,\bFuncVar \cdot \avar)})} = \avar \mapsto \pair{
                            1 + \aFuncVar(\avar)_\cost + \bFuncVar(\avar)_\cost,\aFuncVar(\avar)_\size + \bFuncVar(\avar)_\size + 1,\aFuncVar(\avar)_\star + \bFuncVar(\avar)_\star + 1}$
              \item $\interpret{\sAr{times}(\sAr{der}(\aFuncVar) \cdot \cvar,\bFuncVar \cdot \cvar)} = \pair{
                            1 + \aFuncVar(\cvar)_\cost + 2 * \aFuncVar(\cvar)_\size + \aFuncVar(\cvar)_\star * \aFuncVar(\cvar)_\cost + \bFuncVar(\cvar)_\cost,
                            \aFuncVar(\cvar)_\size * (\aFuncVar(\cvar)_\star + 1) + \bFuncVar(\cvar)_\size + 1,
                            \aFuncVar(\cvar)_\star * (\aFuncVar(\cvar)_\star + 1) + \bFuncVar(\cvar)_\star + 1}$
              \item $\interpret{\sAr{times}(\aFuncVar \cdot \cvar,\sAr{der}(\bFuncVar) \cdot \cvar))} = \pair{
                            1 + \bFuncVar(\cvar)_\cost + 2 * \bFuncVar(\cvar)_\size + \bFuncVar(\cvar)_\star * \bFuncVar(\cvar)_\cost + \aFuncVar(\cvar)_\cost,
                            \bFuncVar(\cvar)_\size * (\bFuncVar(\cvar)_\star + 1) + \aFuncVar(\cvar)_\size + 1,
                            \bFuncVar(\cvar)_\star * (\bFuncVar(\cvar)_\star + 1) + \aFuncVar(\cvar)_\star + 1}$
          \end{itemize}
          $\interpret{\sAr{der}(\abs{\avar}{\sAr{times}(\aFuncVar \cdot \avar,\bFuncVar \cdot \avar)})} = \cvar \mapsto \pair{1 + \text{cost}, \text{size}, \text{star}}$, where:
          \begin{itemize}
              \item $\text{cost} = (1 + \aFuncVar(\cvar)_\cost + \bFuncVar(\cvar)_\cost) + 2 * (\aFuncVar(\cvar)_\size + \bFuncVar(\cvar)_\size + 1) +
                        (\aFuncVar(\cvar)_\star + \bFuncVar(\cvar)_\star + 1) * (1 + \aFuncVar(\cvar)_\cost + \bFuncVar(\cvar)_\cost)$;
              \item $\text{size} = (\aFuncVar(\cvar)_\size + \bFuncVar(\cvar)_\size + 1) * (\aFuncVar(\cvar)_\star + \bFuncVar(\cvar)_\star + 2)$;
              \item $\text{star} = (\aFuncVar(\cvar)_\star + \bFuncVar(\cvar)_\star + 1) * (\aFuncVar(\cvar)_\star + \bFuncVar(\cvar)_\star + 2)$.
          \end{itemize}
          We have $\text{size} =
              \aFuncVar(\cvar)_\size + \bFuncVar(\cvar)_\size + 1 + (\aFuncVar(\cvar)_\size + \bFuncVar(\cvar)_\size + 1) * (\aFuncVar(\cvar)_\star + \bFuncVar(\cvar)_\star + 1) \geq
              \aFuncVar(\cvar)_\size + \bFuncVar(\cvar)_\size + 1 + \aFuncVar(\cvar)_\size * (\aFuncVar(\cvar)_\star + 1) + \bFuncVar(\cvar)_\size * (\bFuncVar(\cvar)_\star + 1) + 1 * 1 =
              (\aFuncVar(\cvar)_\size * (\aFuncVar(\cvar)_\star + 1) + \bFuncVar(\cvar)_\size + 1) + (\bFuncVar(\cvar)_\size * (\bFuncVar(\cvar)_\star + 1) + \aFuncVar(\cvar)_\size + 1) =
              \interpret{\sAr{plus}(\sAr{times}(\sAr{der}(\aFuncVar) \cdot \cvar,\bFuncVar \cdot \cvar),\sAr{times}(\aFuncVar \cdot \cvar,\sAr{der}(\bFuncVar) \cdot \cvar))}_\size$.

          The proof that $\text{star} \geq \interpret{\sAr{plus}(\sAr{times}(\sAr{der}(\aFuncVar) \cdot \cvar,\bFuncVar \cdot \cvar),\sAr{times}(\aFuncVar \cdot \cvar,
                  \sAr{der}(\bFuncVar) \cdot \cvar))}_\star$ is the same, just with $\cdot_\size$ replaced by $\cdot_\star$.

          Finally, we have $\text{cost} >
              \aFuncVar(\cvar)_\cost + \bFuncVar(\cvar)_\cost + 2 * \aFuncVar(\cvar)_\size + 2 * \bFuncVar(\cvar)_\size + 2 +
              1 + \aFuncVar(\cvar)_\cost + \bFuncVar(\cvar)_\cost + (\aFuncVar(\cvar)_\star + \bFuncVar(\cvar)_\star) * (\aFuncVar(\cvar)_\cost + \bFuncVar(\cvar)_\cost) = \\
              1 + 1 + \aFuncVar(\cvar)_\cost + 2 * \aFuncVar(\cvar)_\size + \aFuncVar(\cvar)_\star * \aFuncVar(\cvar)_\cost + \bFuncVar(\cvar)_\cost +
              1 + \bFuncVar(\cvar)_\cost + 2 * \bFuncVar(\cvar)_\size + \bFuncVar(\cvar)_\star * \bFuncVar(\cvar)_\cost + \aFuncVar(\cvar)_\cost =
              1 + \interpret{\sAr{plus}(\sAr{times}(\sAr{der}(\aFuncVar) \cdot \cvar,\bFuncVar \cdot \cvar),\sAr{times}(\aFuncVar \cdot \cvar,\sAr{der}(\bFuncVar) \cdot \cvar))}_\cost$
    \item $\interpret{\sAr{der}(\abs{\avar}{\sAr{ln}(\aFuncVar \cdot \avar)})} \typegr{\real} \interpret{\abs{\cvar}{\sAr{div}(\sAr{der}(\aFuncVar) \cdot \cvar,\aFuncVar \cdot \cvar)}}$
          \begin{itemize}
              \item $\interpret{\abs{\avar}{\sAr{ln}(\aFuncVar \cdot \avar)}} = \avar \mapsto \pair{\aFuncVar(\avar)_\cost,\aFuncVar(\avar)_\size + 1,\aFuncVar(\avar)_\star + 1}$
          \end{itemize}
          $\interpret{\sAr{der}(\abs{\avar}{\sAr{ln}(\aFuncVar \cdot \avar)})} = \cvar \mapsto \pair{1 + \aFuncVar(\cvar)_\cost + 2 * (\aFuncVar(\cvar)_\size + 1) +
                  (\aFuncVar(\cvar)_\star + 1) *  \aFuncVar(\cvar)_\cost,(\aFuncVar(\cvar)_\size + 1) * (\aFuncVar(\cvar)_\star + 2),(\aFuncVar(\cvar)_\star + 1) * (\aFuncVar(\cvar)_\star + 2)}
              \typegr{\real}
              \cvar \mapsto \pair{\aFuncVar(\cvar)_\cost + 2 * (\aFuncVar(\cvar)_\size + 1) + (\aFuncVar(\cvar)_\star + 1) *  \aFuncVar(\cvar)_\cost,0,0} =
              \cvar \mapsto \pair{\aFuncVar(\cvar)_\cost + 2 * \aFuncVar(\cvar)_\size + 2 + \aFuncVar(\cvar)_\star * \aFuncVar(\cvar)_\cost + \aFuncVar(\cvar)_\cost,0,0} =
              \cvar \mapsto \pair{1 + (1 + \aFuncVar(\cvar)_\cost + 2 * \aFuncVar(\cvar)_\size + \aFuncVar(\cvar)_\star * \aFuncVar(\cvar)_\cost) + \aFuncVar(\cvar)_\cost,0,0} =
              \interpret{\abs{\cvar}{\sAr{div}(\sAr{der}(\aFuncVar) \cdot \cvar,\aFuncVar \cdot \cvar)}}$
\end{itemize}

\newpage

\newpage
\section{Extended Proofs}

In this section, we give extended proofs for some results stated in the paper.
It is worth noting that we make heavy use of \textit{function extensionality} for functions in $\SM$; that is, if $F$ and $G$ are both functions in some $\typeinterpret{\atype \arrtype \btype}$, and $F(x) = G(x)$ for all $x \in \typeinterpret{\atype}$, then $F = G$.

We do not prove Theorem \ref{thm:sortinterpretworks} or Lemma \ref{lem:maxpol} here since they are essentially
simpler cases of Theorem \ref{thm:typeinterpretworks} and Lemma \ref{lem:wm} respectively.  Hence, we start
with Section \ref{sec:ho-tp-int}.

\subsection{Proofs for Section \ref{subsec:SM}}

We start by proving the claim in the text that $(\typeinterpret{\atype},\typegr{\atype},\typegeq{\atype})$ is an extended well-founded set for all types $\atype$.

\typeOrd*
\begin{proof}
    We prove the result by induction on all types $\atype$.
    For a base type $\asort$, all items are satisfied by the conditions we impose on extended well-founded sets $(A_\asort, >_\asort, \geq_\asort)$.
    For $\atype = \btype \arrtype \ctype$ we reason as follows.
    \begin{itemize}
        \item $\typegr{\btype \arrtype \ctype}$ is well-founded and $\typegeq{\btype \arrtype \ctype}$ is reflexive.

              Suppose, by contradiction, that there is an infinite chain $F_1 \typegr{\btype \arrtype \ctype} F_2 \typegr{\btype \arrtype \ctype} \dots$ in $\typeinterpret{\btype \arrtype \ctype}$. Then by definition of $\typegr{\btype \arrtype \ctype}$: $\typeinterpret{\btype}$ is non-empty, and for all $\avar \in \typeinterpret{\btype}$, $F_1(\avar) \typegr{\ctype} F_2(\avar) \typegr{\ctype} \ldots$. This induces an infinite $\typegr{\ctype}$-chain in $\typeinterpret{\ctype}$, contradicting the IH.

              For reflexivity, notice that $F \typegeq{\btype \arrtype \ctype} F$ iff for all $x \in \typeinterpret{\btype}$, $F(x) \typegeq{\ctype} F(x)$, which follows directly by reflexivity of $\typegeq{\ctype}$ (IH).

        \item Both relations are transitive.

              For $\typegr{\btype \arrtype \ctype}$. Suppose $F \typegr{\btype \arrtype \ctype} G \typegr{\btype \arrtype \ctype} H$,
              then for all $\avar \in \typeinterpret{\btype}$,
              $F(\avar) \typegr{\ctype} G(\avar) \typegr{\ctype} H(\avar)$ holds by definition of $\typegr{\btype \arrtype \ctype}$.
              The IH give us $F(\avar) \typegr{\ctype} H(\avar)$, for all $x \in \typeinterpret{\btype}$,
              which is exactly $F \typegr{\btype \arrtype \ctype} H$. Non-emptiness of $\typeinterpret{\btype}$ holds by assumption.
              The case for $\typegeq{\btype \arrtype \ctype}$ is analogous.

        \item For all $F,G,H$ in $\typeinterpret{\btype \arrtype \ctype}$, $F \typegr{\btype \arrtype \ctype} G$ implies $F \typegeq{\btype \arrtype \ctype} G$,
              and $F \typegr{\btype \arrtype \ctype} G \typegeq{\btype \arrtype \ctype} H$ implies $F \typegr{\btype \arrtype \ctype} H$.

              Suppose $F \typegr{\btype \arrtype \ctype} G$.
              By definition, $F(\avar) \typegr{\ctype} G(\avar)$ for all $\avar \in \typeinterpret{\btype}$. By IH, $F(x) \typegeq{\ctype} G(x)$, for all $x \in \typeinterpret{\ctype}$, which means $F \typegeq{\ctype} G$.

              If, moreover, $G \typegeq{\btype \arrtype \ctype} H$, the reasoning is similar: expand the definitions and apply the induction hypothesis.
              \qedhere
    \end{itemize}
\end{proof}

In the text, we quietly asserted that Definition \ref{def:sminterpret} is well-defined.  Let us now prove this.

\begin{lemma}
For all terms $\aterm :: \atype$ and suitable $\alpha$ as described in Definition \ref{def:sminterpret} we have: $\ainterpret{\aterm} \in \typeinterpret{\aterm}$.
\end{lemma}

\begin{proof}
We will prove by induction on the form of $\aterm$: $\ainterpret{\aterm} \in \typeinterpret{\aterm}$,
and for all variables $\avar$ occurring in the domain of $\alpha$:
either $d \mapsto \interpret{\aterm}_{[\avar:=d]}$ is a strongly monotonic function, or it is a constant function.
Recall the definition of $\ainterpret{\aterm}$.
    \[
        \begin{array}{lcl}
            \ainterpret{\avar} = \alpha(\avar)\ \text{for variables}\ \avar              & \quad &
            \ainterpret{\afun(\aterm_1,\dots,\aterm_k)} =\funcinterpret{\afun}(\ainterpret{\aterm_1},\dots,\ainterpret{\aterm_k})                                             \\
            \ainterpret{\app{\aterm}{\bterm}} = \ainterpret{\aterm}(\ainterpret{\bterm}) &       &
            \ainterpret{\abs{x}{\aterm}} = \makesm{\atype,\btype}(d \mapsto \interpret{\aterm}_{\alpha[x:=d]})\ \text{if}\ \avar :: \atype\ \text{and}\ \aterm :: \btype \\
        \end{array}
    \]
Consider the form of $\aterm$.
\begin{itemize}
\item If $\aterm = \avar$ then $\ainterpret{\avar} = \alpha(\avar) \in \typeinterpret{\atype}$ by assumption.
  Moreover, $d \mapsto \interpret{\aterm}_{\alpha[\avar:=d]}$ is the function $d \mapsto d$, which is strongly monotonic:
  if $a \typegr{\atype} b$ then $(d \mapsto d)(a) = a \typegr{\atype} b = (d \mapsto d)(b)$.
  For all other variables $y$, the function $d \mapsto \interpret{\aterm}_{\alpha[\bvar:=d]}$ is the constant
  function $d \mapsto \alpha(\avar)$.
\item If $\aterm = \app{\bterm}{\cterm}$ then $\bterm :: \btype \arrtype \atype$ and $\cterm :: \btype$.
  By the induction hypothesis, $\ainterpret{\bterm} \in \typeinterpret{\btype \arrtype \atype} \subseteq
  \typeinterpret{\btype} \arrfunc \typeinterpret{\atype}$, and $\ainterpret{\cterm} \in \typeinterpret{\btype}$.
  Hence, $\ainterpret{\bterm}(\ainterpret{\cterm}) \in \typeinterpret{\atype}$.
  Also by the induction hypothesis, $d \mapsto \interpret{\bterm}_{\alpha[\avar:=d]}$ is either strongly monotonic
  or constant, and the same holds for $d \mapsto \interpret{\cterm}_{\alpha[\avar:=d]}$.
  We have four cases:
  \begin{itemize}
  \item Both are constant: then $d \mapsto \interpret{\bterm}_{\alpha[\avar:=d]}(\interpret{\cterm}_{\alpha[\avar:=d]})$
    is constant as well.
  \item $d \mapsto \interpret{\bterm}_{\alpha[\avar:=d]}$ is constant and $d \mapsto \interpret{\cterm}_{\alpha[\avar:=d]}$
    is strongly monotonic: then for $a \typegr{} b$ we have: $\interpret{\bterm}_{\alpha[\avar:=a]} = \interpret{\bterm}_{
    \alpha[\avar:=b]} = \ainterpret{\bterm}$, and $\interpret{\cterm}_{\alpha[\avar:=a]} \typegr{\btype}
    \interpret{\cterm}_{\alpha[\avar:=b]}$.  Hence, by monotonicity of $\ainterpret{\bterm}$ we have:
    $\interpret{\aterm}_{[\avar:=a]} = \ainterpret{\bterm}(\interpret{\cterm}_{\alpha[\avar:=a]}) \typegr{\atype}
    \ainterpret{\bterm}(\interpret{\cterm}_{\alpha[\avar:=b]}) = \interpret{\aterm}_{[\avar:=b]}$.
  \item $d \mapsto \interpret{\bterm}_{\alpha[\avar:=d]}$ is strongly monotonic and $d \mapsto
    \interpret{\cterm}_{\alpha[\avar:=d]}$ is constant: then for $a \typegr{} b$ we have: $\interpret{\bterm}_{
    \alpha[\avar:=a]} \typegr{\btype \arrtype \atype} \interpret{\bterm}_{\alpha[\avar:=b]}$ and
    $\interpret{\cterm}_{\alpha[\avar:=a]} = \interpret{\cterm}_{\alpha[\avar:=b]} = \ainterpret{\cterm}$.
    By definition of $\typegr{\btype \arrtype \atype}$, we thus have
    $\interpret{\aterm}_{[\avar:=a]} = \interpret{\bterm}_{\alpha[\avar:=a]}(\ainterpret{\cterm}) \typegr{\atype}
    \interpret{\bterm}_{\alpha[\avar:=b]}(\ainterpret{\cterm}) = \interpret{\aterm}_{[\avar:=b]}$.
  \item Both are strongly monotonic: then by monotonicity of $\interpret{\bterm}_{\alpha[\avar:=a]}$ we have
    that $\interpret{\aterm}_{\alpha[\avar:=a]} =
    \interpret{\bterm}_{\alpha[\avar:=a]}(\interpret{\cterm}_{\alpha[\avar:=a]}) \typegr{\atype}
    \interpret{\bterm}_{\alpha[\avar:=a]}(\interpret{\cterm}_{\alpha[\avar:=b]})$, and this $\typegr{\atype}
    \interpret{\bterm}_{\alpha[\avar:=b]}(\interpret{\cterm}_{\alpha[\avar:=b]}) = \interpret{\aterm}_{\alpha
    [\avar:=b]}$ because $\interpret{\bterm}_{\alpha[\avar:=a]} \typegr{\btype \arrtype \atype}
    \interpret{\bterm}_{\alpha[\avar:=b]}$.
  \end{itemize}
\item If $\aterm = \afun(\aterm_1,\dots,\aterm_k)$ with $\afun :: [\btype_1 \times \dots \times \btype_k] \arrtype \atype$
  then note that $\ainterpret{\aterm}$ is exactly $\interpret{\cvar \cdot \aterm_1 \cdots \aterm_k}_{\alpha[\cvar:=
  \funcinterpret{\afun}]}$ for a fresh variable $\cvar$.  Hence, the two statements we need to prove follow by
  using first the case for a variable, and then $k$ times the case for an application.
\item Finally, if $\aterm = \abs{\avar}{\bterm}$ with $\atype = \btype \arrtype \ctype$, then
  $\ainterpret{\aterm} = \makesm{\btype,\ctype}(d \mapsto \interpret{\bterm}_{\alpha[\avar:=d]})$.  Since, by the
  induction hypothesis, $d \mapsto \interpret{\bterm}_{\alpha[\avar:=d]}$ is either a constant or a strongly monotonic
  function from $\typeinterpret{\btype}$ to $\typeinterpret{\ctype}$, this is well-defined, and $\makesm{\btype,
  \ctype}(d \mapsto \interpret{\bterm}_{\alpha[\avar:=d]})$ yields an element of $\typeinterpret{\btype \arrtype
  \ctype}$.

  Now, let $\bvar$ be a variable.  If $\bvar = \avar$, then $e \mapsto \interpret{\aterm}_{\alpha[\bvar:=e]}$ is
  clearly a constant function: $\interpret{\aterm}_{\alpha[\bvar:=e} = \makesm{\btype,\ctype}(d \mapsto \interpret{
  \bterm}_{\alpha[\avar:=e][\avar:=d]}) = \makesm{\btype,\ctype}(d \mapsto \interpret{\bterm}_{\alpha[\avar:=d]}$.
  Otherwise, note that by the induction hypothesis either $e \mapsto \interpret{\bterm}_{\alpha[\bvar:=e][\avar:=d]}$
  is constant, or it is strongly monotonic.
  If it is constant, then for all $a,b$:
  $d \mapsto \interpret{\bterm}_{\alpha[\bvar:=a][\avar:=d]} = d \mapsto \interpret{\bterm}_{\alpha[\bvar:=b][\avar:=
  d]}$, and hence $\makesm{\btype,\ctype}(d \mapsto \interpret{\bterm}_{\alpha[\bvar:=a][\avar:=d]})$ is the same as
  $\makesm{\btype,\ctype}(d \mapsto \interpret{\bterm}_{\alpha[\bvar:=b][\avar:=d]})$; hence, $e \mapsto
  \interpret{\aterm}_{\alpha[\bvar:=e]}$ is constant too.
  Otherwise, if this function is strongly monotonic, then the function
  $d \mapsto \interpret{\bterm}_{\alpha[\bvar:=a][\avar:=d]}$ is pointwise greater than
  $d \mapsto \interpret{\bterm}_{\alpha[\bvar:=b][\avar:=d]}$.  Hence,
  $\interpret{\aterm}_{\alpha[\bvar:=a]} \typegr{\atype} \interpret{\aterm}_{\alpha[\bvar:=b]}$ as well.
  \qedhere
\end{itemize}
\end{proof}

To prove Theorem \ref{thm:typeinterpretworks} we need an AFS version of the so-called \textit{Substitution Lemma}.
We begin by giving a systematic way of extending a substitution (seen as a morphism between terms) to a valuation, seen as morphism from terms to elements of $\SMA$.
\begin{definition}\label{def:subst-lift}
    Given a substitution $\gamma = [\avar_1 := \aterm_1, \dots, \avar_n := \aterm_n]$
    and a valuation $\alpha$, we define $\alpha^\gamma$ as the valuation such that $\alpha^\gamma(\avar) = \alpha(\avar)$, if $\avar \notin \dom{\gamma}$; and $\alpha^\gamma(\avar) = \interpret{\avar\gamma}_\alpha$, otherwise.
\end{definition}
\begin{lemma}[Substitution Lemma]\label{lemma:sbst:lemma}
    For any substitution $\gamma$ and valuation $\alpha$,
    $\interpret{\aterm\gamma}_\alpha = \interpret{\aterm}_{\alpha^\gamma}$.
    Additionally, if $\interpret{\aterm} \typegr{\atype} \interpret{\bterm}$ $(\interpret{\aterm} \typegeq{\sigma} \interpret{\bterm})$, then $\interpret{\aterm\gamma} \typegr{\atype} \interpret{\bterm\gamma}$ $(\interpret{\aterm\gamma} \typegeq{\atype} \interpret{\bterm\gamma})$.
\end{lemma}
\begin{proof} By inspection of Definition \ref{def:subst-lift} it can be easily shown by induction on $\aterm$ that the following diagram commutes:
    \[
        \begin{tikzcd}
            && \terms \\
            \\
            \terms && \SMA
            \arrow["{\interpret{\cdot}_{\alpha^\gamma}}"', from=3-1, to=3-3]
            \arrow["\gamma", from=3-1, to=1-3]
            \arrow["{\interpret{\cdot}_\alpha}", from=1-3, to=3-3]
        \end{tikzcd}
    \]
    As a consequence, if $\interpret{\aterm}_\alpha \typegr{\atype} \interpret{\bterm}$ for any valuation $\alpha$,
    then $\interpret{\aterm}_{\alpha^\gamma} \typegr{\atype} \interpret{\bterm}_{\alpha^\gamma}$ in particular.
    So $\interpret{\aterm\gamma}_\alpha \typegr{\atype} \interpret{\bterm}_\alpha$.
    The case for $\typegeq{\atype}$ is analogous.
\end{proof}

Theorem \ref{thm:typeinterpretworks} is proved by induction on the rewrite relation.
\typeIntWorks*
\begin{proof}
    We reason by induction on $\aterm \arrz \bterm$. We have six cases to consider.
    \begin{itemize}
        \item Suppose $\aterm \arrz \bterm$ by $\ell \gamma \arrz r\gamma$.
              Compatibility gives $\interpret{\ell} \typegr{} \interpret{r}$, and by Lemma \ref{lemma:sbst:lemma} we have $\interpret{\ell\gamma} \typegr{} \interpret{r\gamma}$.

        \item The case $(\abs{\avar}{\aterm}) t \arrz \aterm[x := \bterm]$ follows by Compatibility.

        \item Suppose $\aterm \arrz \bterm$ by $\afun(\dots, \aterm, \dots) \arrz \afun(\dots, \bterm, \dots)$.
              By induction hypothesis, $\interpret{\aterm} \typegr{} \interpret{\bterm}$.
              If the reduction occurs in the first argument of $\afun$,
              then $\interpret{\afun(\aterm, \dots)} \typegr{} \interpret{\afun(\bterm, \dots)}$ by the fact that $\funcinterpret{\afun}$ is in $\SM$.
              For the other cases, observe that $\funcinterpret{\afun}(\interpret{\aterm_1},\dots,\interpret{\aterm_{i}}) \in \SM$ for all $i$,
              so $\funcinterpret{\afun}(\interpret{\aterm_1},\dots,\interpret{\aterm_{i}},\aterm) \typegr{}
              \funcinterpret{\afun}(\interpret{\aterm_1},\dots,\interpret{\aterm_{i}},\bterm)$.

        \item Suppose $\abs{\avar}{\aterm} \arrz \abs{\avar}{\bterm}$, with $\aterm \arrz \bterm$.
              If $x \notin \fvars{\aterm}$ then $d \mapsto \interpret{\aterm}_{\alpha[x := d]} \typegr{} d \mapsto \interpret{\bterm}_{\alpha[x := d]}$ are constant functions not in $\SM$.
              By Definition \ref{def:SM}, $\makesm{\atype, \btype}(d \mapsto \interpret{\aterm}_{\alpha[x := d]}) \typegr{\atype \arrtype \btype} \makesm{\atype, \btype}(d \mapsto \interpret{\bterm}_{\alpha[x := d]})$. On the other hand, if $x \in \fvars{\aterm}$, then $d \mapsto \interpret{\aterm}_{\alpha[x := d]} \typegr{\atype \arrtype \btype} d \mapsto \interpret{\bterm}_{\alpha[x := d]}$ are strongly monotonic functions, and the result follows by Definition \ref{def:monomaker}.

        \item The cases for application follow directly from Definition \ref{def:SM} and the induction hypothesis.
    \end{itemize}
\end{proof}

\subsection{Proofs for Section \ref{subsec:makesm}}

We prove some results regarding the functions $\nul{\atype}$, $\addcost{\atype}$ and $\costof{\atype}$.

First, as stated in the text:

\begin{lemma}\label{lem:builders} For all types $\atype$:
    \begin{enumerate}
        \item\label{lem:builders:nul}
              $\nul{\atype} \in \typeinterpret{\atype}$;
        \item\label{lem:builders:addcost}
              for all $n \in \N$ and $x \in \typeinterpret{\atype}$: $\addcost{\atype}(n,x) \in \typeinterpret{\atype}$;
        \item\label{lem:builders:costof}
              $\costof{\atype}$ is weakly monotonic and strict in its first argument;
        \item\label{lem:builders:addcost2}
              $\addcost{\atype}$ is weakly monotonic and strict in both its arguments.
    \end{enumerate}
\end{lemma}

\begin{proof}
    By a mutual induction on $\atype$.

    (\ref{lem:builders:nul})
    If $\atype = \asort \in \sortset$, then $\nul{\atype} = \pair{0,\dots,0}$ is clearly in $\atype{\asort}$.
    If $\atype = \btype \arrtype \ctype$ then $\nul{\btype \arrtype \ctype} = d \mapsto \addcost{\ctype}(
        \costof{\btype}(d),\nul{\ctype})$.  Clearly $\costof{\btype}(d) \in \N$ and by induction hypothesis
    (\ref{lem:builders:nul}) $\nul{\ctype} \in \typeinterpret{\ctype}$, so by induction hypothesis
    (\ref{lem:builders:addcost}) $\addcost{\ctype}(\costof{\btype}(d),\nul{\ctype}) \in \typeinterpret{
            \ctype}$.  We still need to see that this function is weakly monotonic and strict in its argument.
    So suppose $x \typegr{\asort} y$; the case for $x \typegeq{\asort} y$ is similar.
    Then $\costof{\btype}(x) > \costof{\btype}(y)$ by induction hypothesis (\ref{lem:builders:costof}).
    Hence, $\addcost{\ctype}(\costof{\btype}(x),\nul{\ctype}) \typegr{\ctype}
        \addcost{\ctype}(\costof{\btype}(y,\nul{\ctype})$ by induction hypothesis
    (\ref{lem:builders:addcost2}); that is $\nul{\atype}(x) \typegr{\ctype} \nul{\atype}(y)$.

    (\ref{lem:builders:addcost})
    If $\atype = \asort \in \sortset$, then $\addcost{\atype}(n,x) = \pair{n + x_1,x_2,\dots,x_{\typecount{\asort}}}
        \in \typeinterpret{\asort}$.  Otherwise, let $\atype = \btype \arrtype \ctype$, and let $n \in \N$ and
    $F \in \typeinterpret{\btype \arrtype \ctype}$.  Then $\addcost{\btype \arrtype \ctype}(n,F) = d \mapsto
        \addcost{\ctype}(n,F(d))$.  By induction hypothesis (\ref{lem:builders:addcost}),
    $\addcost{\ctype}(n,F(d)) \in \typeinterpret{\ctype}$, so $\addcost{\btype \arrtype \ctype}(n,F) \in
        \typeinterpret{\btype} \arrfunc \typeinterpret{\ctype}$; we only need to see that it is strongly monotonic.
    So let $u \typegr{\btype} w$; we will see that $\addcost{\btype \arrtype \ctype}(n,F,u) \typegr{\ctype}
        \addcost{\btype \arrtype \ctype}(n,F,w)$; the case for $u \typegeq{\btype} w$ is similar.
    We have $\addcost{\btype \arrtype \ctype}(n,F,u) = \addcost{\ctype}(n,F(u))$.  Since $F$ is strongly
    monotonic, $F(u) \typegr{\ctype} F(w)$.  By induction hypothesis (\ref{lem:builders:addcost2}),
    $\addcost{\ctype}(n,F(u)) \typegr{\ctype} \addcost{\ctype}(n,F(w)) = \addcost{\btype \arrtype \ctype}(n,F,w)$.

    (\ref{lem:builders:costof})
    Suppose $x \typegr{\atype} y$; the case for $x \typegeq{\atype} y$ is similar.
    If $\atype = \asort \in \sortset$, then $\costof{\atype}(x) = x_1 > y_1 = \costof{\atype}(y)$ y
    definition of $x \typegr{\asort} y$.
    If $\atype = \btype \arrtype \ctype$ then $\costof{\atype}(x) = \costof{\ctype}(x(\nul{\btype}))$
    Since $x \typegr{\btype \arrtype \ctype} y$ we have $x(\nul{\btype}) \typegr{\ctype} y(\nul{\btype})$.
    By induction hypothesis (\ref{lem:builders:costof}), $\costof{\ctype}(x(\nul{\btype})) >
        \costof{\ctype}(y(\nul{\btype}))$ follows as required.

    (\ref{lem:builders:addcost2})
    Suppose $n \geq m$ and $x \typegeq{\atype} y$.
    We will see that (a) $\addcost{\atype}(n,x) \typegeq{\atype} \addcost{\atype}(m,y)$.and (b) if $n > m$
    or $x \typegr{\atype} y$ then $\addcost{\atype}(n,x) \typegr{\atype} \addcost{\atype}(m,y)$.
    If $\atype = \asort \in \sortset$, then $\addcost{\atype}(n,x) = \pair{n + x_1,x_2,\dots,x_{\typecount{\asort}}}
        \typegeq{\asort} \pair{m + y_1,y_2,\dots,y_{\typecount{\asort}}}$ because each $x_i \geq y_i$ and $n \geq m$;
    in case (b) we have $n > m$ or $x_1 > y_1$ so certainly $n + x_1 > m + y_1$.
    If $\atype = \btype \arrtype \ctype$ then $\addcost{\atype}(n,x) = d \mapsto \addcost{\ctype}(n,x(d))$.
    By definition, $x(d) \typegeq{\ctype} y(d)$ and if $x \typegr{\atype} y$ even $x(d) \typegr{\ctype} y(d)$.
    Hence, by induction hypothesis (\ref{lem:builders:addcost2}), $\addcost{\ctype}(n,x(d)) \typegeq{\ctype}
        \addcost{\ctype}(m,y(d))$ and if $n > m$ or $x \typegr{\atype} y$ even $\addcost{\ctype}(n,x(d))
        \typegr{\ctype} \addcost{\ctype}(m,y(d))$.  This suffices, since $\typegeq{\btype \arrtype \ctype}$ and
    $\typegr{\btype \arrtype \ctype}$ just do a place-wise comparison.
\end{proof}

Next, the following lemmas provide basic properties of these functions (and how they interact with each other).

\begin{lemma}\label{lem:addcostprops}
    For all types $\atype$, for all $x \in \typeinterpret{\atype}$:
    \begin{enumerate}
        \item\label{lem:addcostprops:addcost0}
              $\addcost{\atype}(0,x) = x$;
        \item\label{lem:addcostprops:addcosttwice}
              for all $n,m \in \N$: $\addcost{\atype}(n,\addcost{\atype}(m,x)) = \addcost{\atype}(n+m,x)$;
        \item\label{lem:addcostprops:addcostpositive}
              if $n > 0$ then $\addcost{\atype}(n,x) \typegr{\atype} x$;
        \item\label{lem:addcostprops:addone}
              if $y \in \typeinterpret{\atype}$ is such that $x \typegr{\atype} y$ then
              $x \typegeq{\atype} \addcost(1,y)$;
        \item\label{lem:addcostprops:costof}
              for all $n \in \N$: $\costof{\atype}(\addcost{\atype}(n,x)) = n + \costof{\atype}(x)$.
    \end{enumerate}
\end{lemma}

\begin{proof}
    All items hold by induction on $\atype$.

    (\ref{lem:addcostprops:addcost0})
    If $\atype = \asort \in \sortset$ then $\addcost{\atype}(0,x) = \pair{0 + x_1,x_2,
        \dots,x_{\typecount{\asort}}} = \pair{x_1,\dots,x_{\typecount{\asort}}} = x$.
    If $\atype = \btype \arrtype \ctype$ then $\addcost{\atype}(0,x) = d \mapsto
        \addcost{\ctype}(0,x(d)) =$ (IH) $d \mapsto x(d) = d$ (extensionally).

    (\ref{lem:addcostprops:addcosttwice})
    If $\atype = \asort \in \sortset$, then $\addcost{\atype}(n,\addcost{\atype}(m,x)) =
        \pair{n + m + x_1,x_2,\dots,x_{\typecount{\asort}}} = \addcost{\atype}(n+m,x)$.
    If $\atype = \btype \arrtype \ctype$, then $\addcost{\atype}(n,\addcost{\atype}(m,x))
        = d \mapsto \addcost{\atype}(n,\addcost{\atype}(m,x(d))) =$ (IH)
    $d \mapsto \addcost{\atype}(n + m,x(d)) = \addcost{\atype}(n + m,x)$.

    (\ref{lem:addcostprops:addcostpositive})
    Let $n > 0$.
    If $\atype = \asort \in \sortset$ then $\addcost{\atype}(n,x) =
        \pair{n + x_1,x_2,\dots,x_{\typecount{\asort}}} \typegr{\asort}
        \pair{x_1,\dots,x_{\typecount{\asort}}} = x$.
    If $\atype = \btype \arrtype \ctype$ then $\addcost{\atype}(n,x) = d \mapsto
        \addcost{\ctype}(n,x(d))$ and by the induction hypothesis,
    $\addcost{\ctype}(n,x(d)) \typegr{\ctype} x(d)$.  Hence, since $\typegr{\atype}$
    does a pointwise comparison, $d \mapsto \addcost{\ctype}(n,x(d)) \typegr{\atype}
        d \mapsto x(d) = x$ (extensionally).

    (\ref{lem:addcostprops:addone})
    Let $x \typegr{\atype} y$.
    If $\atype = \asort \in \sortset$, then $x = \pair{x_1,\dots,x_{\typecount{\asort}}}$
    and $y = \pair{y_1,\dots,y_{\typecount{\asort]}}}$, and $x \typegr{\asort} y$ implies
    that $x_1 > y_1$ and each $x_i \geq y_i$.  But then also $x_1 \geq 1 + y_1$, so
    $x = \pair{x_1,\dots,x_{\typecount{\asort}}} \typegeq{\asort} \pair{1 + y_1,y_2,\dots,
            y_{\typecount{\asort}}} = \addcost{\asort}(1,y)$.
    If $\atype = \btype \arrtype \ctype$, then $x = d \mapsto x(d)$ and $y = d \mapsto
        y(d)$ and $x \typegr{\atype} y$ implies that $x(d) \typegr{\ctype} y(d)$ for all
    $d \in \typeinterpret{\btype}$.  By the induction hypothesis, $x(d) \typegeq{\ctype}
        \addcost{\ctype}(1,y(d))$ for all $d$, and therefore $x = d \mapsto x(d)
        \typegeq{\btype \arrtype \ctype} d \mapsto \addcost{\ctype}(1,y(d)) = \addcost{
            \btype \arrtype \ctype}(y)$.

    (\ref{lem:addcostprops:costof})
    If $\atype = \asort \in \sortset$, then $\costof{\atype}(\addcost{\atype}(n,x)) =
        \costof{\atype}(\pair{n + x_1,x_2,\dots,x_{\typecount{\asort}}}) = n + x_1 = n +
        \costof{\atype}(x)$.
    If $\atype = \btype \arrtype \ctype$, then $\costof{\atype}(\addcost{\atype}(n,x))
        = \costof{\atype}(d \mapsto \addcost{\ctype}(n, x(d)))
        = \costof{\ctype}(\addcost{\ctype}(n,x(\nul{\btype})))$, which by the induction hypothesis equals
    $n + \costof{\ctype}(x(\nul{\btype})) = n + \costof{\btype \arrtype \ctype}(x)$.
\end{proof}

\begin{lemma}\label{lem:addcostincrease}
    For all types $\atype,\btype$, $F \in \typeinterpret{\atype \arrtype \btype}$,
    $x \in \typeinterpret{\atype}$ and $n \in \N$: \\
    $F(\addcost{\atype}(n,x)) \typegeq{\atype} \addcost{\btype}(n,F(x))$.
\end{lemma}

\begin{proof}
    By induction on $n$.
    If $n = 0$, then $F(\addcost{\atype}(n,x)) = F(x) = \addcost{\btype}(n,F(x))$ by
    Lemma \ref{lem:addcostprops}(\ref{lem:addcostprops:addcost0}).
    If $n = i + 1$, then $\addcost{\atype}(n,x) = \addcost{\atype}(1,\addcost{\atype}(i,x))$
    by Lemma \ref{lem:addcostprops}(\ref{lem:addcostprops:addcosttwice}), which
    $\typegr{\atype} \addcost{\atype}(i,x)$
    by Lemma \ref{lem:addcostprops}(\ref{lem:addcostprops:addcostpositive}).
    Hence, by monotonicity, $F(\addcost{\atype}(n,x)) \typegr{\btype} F(\addcost{\atype}(i,x))$.
    By the induction hypothesis, $F(\addcost{\atype}(i,x)) \typegeq{\btype} \addcost{\btype}(i,
        F(x))$, so $F(\addcost{\atype}(n,x)) \typegr{\btype} \addcost{\btype}(i,F(x))$.  By
    Lemma \ref{lem:addcostprops}(\ref{lem:addcostprops:addone}) therefore
    $F(\addcost{\atype}(n,x) \typegeq{\btype} \addcost{\btype}(1,\addcost{\btype}(i,F(x))$.
    By Lemma \ref{lem:addcostprops}(\ref{lem:addcostprops:addcosttwice}) we
    thus have $F(\addcost{\atype}(n,x)) \typegeq{\btype}(i+1,F(x))$.
\end{proof}

\begin{lemma}\label{lem:estimatebycost}
    For all types $\atype$ and all $x \in \typeinterpret{\atype}$:
    $x \typegeq{\atype} \addcost{\atype}(\costof{\atype}(x),\nul{\atype})$.
\end{lemma}

\begin{proof}
    By induction on $\atype$.

    If $\atype = \asort \in \sortset$ then $x = \pair{x_1,x_2,\dots,
            x_{\typecount{\asort}}} \typegeq{\asort} \pair{x_1,0,\dots,0} =
        \addcost{\asort}(x_1,\pair{0,\dots,0}) = \addcost{\asort}(\costof{\asort}(
        x),\nul{\asort})$.
    In the remainder, we consider the case $\atype = \btype \arrtype \ctype$.

    In this case, $x = d \mapsto x(d)$ (extensionally), which by the induction
    hypothesis $\typegeq{\btype \arrtype \ctype} d \mapsto \addcost{\ctype}(\costof{
        \ctype}(x(d)),\nul{\ctype})$.
    On the other hand, $\addcost{\atype}(\costof{\atype}(x),\nul{\atype}) = d \mapsto\linebreak
        \addcost{\ctype}(\costof{\atype}(x),\nul{\atype}(d)) = d \mapsto \addcost{\ctype}(
        \costof{\atype}(x),\addcost{\ctype}(\costof{\btype}(d),\nul{\ctype}))$.  By
    Lemma \ref{lem:addcostprops}(\ref{lem:addcostprops:addcosttwice}) this is
    exactly $d \mapsto \addcost{\ctype}(\costof{\atype}(x) + \costof{\btype}(d),
        \nul{\ctype})$.

    Hence, by monotonicity of $\addcost{\ctype}$ (Lemma
    \ref{lem:builders}(\ref{lem:builders:addcost2})), it suffices if we can see that, for all $d$,
    $\costof{\ctype}(x(d)) \geq \costof{\atype}(x) + \costof{\btype}(d)$.
    To see this, note that by the induction hypothesis,
    $d \typegeq{\btype} \addcost{\btype}(\costof{\btype}(d),\nul{\btype})$.
    Hence, $x(d) \typegeq{\ctype} x(\addcost{\btype}(\costof{\btype}(d),\nul{\btype})$ by monotonicity
    of $x$.
    By Lemma \ref{lem:addcostincrease} we have
    $x(d) \typegeq{\ctype} \addcost{\ctype}(\costof{\btype}(d),x(\nul{\btype}))$.
    Hence, by monotonicity of $\costof{\ctype}$ (Lemma
    \ref{lem:builders}(\ref{lem:builders:costof})),
    $\costof{\ctype}(x(d)) \geq \costof{\ctype}(\addcost{\ctype}(\costof{\btype}(d),x(\nul{\btype})))$.
    By Lemma \ref{lem:addcostprops}(\ref{lem:addcostprops:costof}), this
    $= \costof{\btype}(d) + \costof{\ctype}(x(\nul{\btype}))
        = \costof{\btype}(d) + \costof{\atype}(x)$.
    Hence, we have obtained the required inequality
    $\costof{\ctype}(x(d)) \geq \costof{\atype}(x) + \costof{\btype}(d)$.
\end{proof}

\begin{lemma}\label{lem:argumentincrease}
For $F \in \typeinterpret{\atype \arrtype \btype}$ and $\avar \in \typeinterpret{\atype}$ we have:
$\costof{\btype}(F(\avar)) \geq \costof{\atype}(\avar)$.
\end{lemma}

\begin{proof}
Let $n := \costof{\atype}(\avar)$.
    By Lemma \ref{lem:estimatebycost}, $\avar \typegeq{\atype} \addcost{\atype}(\costof{\atype}(\avar),\nul{\atype}) =
    \addcost{\atype}(n,\nul{\atype})$.
    Hence, by monotonicity of $F$, $F(\avar) \typegeq{\btype} F(\addcost{\atype}(n,\nul{\atype}))$.
    By Lemma \ref{lem:addcostincrease}, this implies that $F(\avar) \typegeq{\btype} \addcost{\btype}(n,F(\nul{\atype}))$.
    Since $\costof{\btype}$ is strict in its first argument by Lemma \ref{lem:builders}(\ref{lem:builders:costof}),
    we thus have $\costof{\btype}(F(\avar)) \geq \costof{\atype}(\addcost{\btype}(n,F(\nul{\atype})))$, which
    $\geq n$ by Lemma \ref{lem:addcostprops}(\ref{lem:addcostprops:costof}).
\end{proof}

With these lemmas, we can prove the lemma stated in the text: that the function in Definition
\ref{def:makesmchoices} is indeed a $(\atype,\btype)$-monotonicity function.

\begin{lemma}\label{lem:makesmchoices}
    Let $\atype,\btype$ be simple types.
    Then $\addarg_{\atype,\btype}$ is a $(\atype,\btype)$-monotonicity function.
\end{lemma}

\begin{proof}
First, we must see that $\addarg_{\atype,\btype}$ maps each element of $C_{\atype,\btype}$ to an element of
$\typeinterpret{\atype,\btype}$.  Thus, let $F \in C_{\atype,\btype}$.  There are two cases:
\begin{itemize}
\item $F$ is a constant function in $\typeinterpret{\atype} \arrfunc \typeinterpret{\btype}$.

  Then $\addarg_{\atype,\btype}(F)$ is the function $d \mapsto \addcost{\btype}(\costof{\atype}(d) + 1,F(d))$.
  Since $F \in \typeinterpret{\atype} \arrfunc \typeinterpret{\btype}$ we have $F(d) \in
  \typeinterpret{\btype}$ so $\addcost{\btype}(\costof{\atype}(d) +1,F(d)) \in \typeinterpret{\btype}$ by
  Lemma \ref{lem:builders}(\ref{lem:builders:addcost}); hence, $d \mapsto \addcost{\btype}(\costof{\atype}(d)
  + 1,F(d)) \in \typeinterpret{\atype} \arrfunc \typeinterpret{\btype}$.

  It remains to be seen that this function is (a) weakly monotonic, and (b) strict in its only
  argument.  We show only the latter; the former is very similar.

  Let $x,y \in \typeinterpret{\atype}$ with $x \typegr{\atype} y$.
  Then by Lemma \ref{lem:builders}(\ref{lem:builders:costof}),
  $\costof{\atype}(x) > \costof{\atype}(y)$, which implies $\costof{\atype}(x) + 1 >
      \costof{\atype}(y) + 1$ as well.  Moreover, since $F$ is constant, we have
  $F(x) = F(y)$, so certainly $F(x) \typegeq{\btype} F(y)$.  Thus, by Lemma
  \ref{lem:builders}(\ref{lem:builders:addcost2}), we have $\addcost{\btype}(\costof{\atype}(x) + 1,
  F(x)) \typegr{\btype} \addcost{\btype}(\costof{\atype}(y) + 1,F(y))$.

\item $F$ is a function in $\typeinterpret{\atype \arrtype \btype}$; that is, a strongly monotonic
  function in $\typeinterpret{\atype} \arrfunc \typeinterpret{\btype}$.
  Then $\addarg_{\atype,\btype}(F)$ is the function $d \mapsto \addcost{\btype}(1,F(d))$.  By Lemma
  \ref{lem:builders}(\ref{lem:builders:addcost}) this function is indeed in $\typeinterpret{\atype}
  \arrfunc \typeinterpret{\btype}$.
  To see that it is monotonic, suppose that $x \typegr{\atype} y$; the case for $x \typegeq{\atype}
  y$ is similar.
  Then $F(x) \typegr{\btype} F(y)$ by strong monotonicity of $F$.
  By Lemma \ref{lem:builders}(\ref{lem:builders:addcost2}), $\addcost{\btype}(1,F(x))
  \typegr{\btype} \addcost{\btype}(1,F(y))$ as required.
\end{itemize}

Second, we will see that $\addarg_{\atype,\btype}$ is strongly monotonic.  That is, for $F, G \in
C_{\atype,\btype}$: (a) if $F(x) \typegeq{\btype} G(x)$ for all $x \in \typeinterpret{\atype}$ then
$\addarg_{\atype,\btype}(F) \typegeq{\atype \arrtype \btype} \addarg_{\atype,\btype}(G)$; (b) if
$F(x) \typegr{\btype} G(x)$ for all $x \in \typeinterpret{\atype}$ then $\addarg_{\atype,\btype}(F)
\typegr{\atype \arrtype \btype} \addarg_{\atype,\btype}(G)$.
We will only show (b); the proof of (a) is parallel.
There are four cases to consider:
\begin{itemize}
\item $F,G$ are both constant functions. Then $\addarg_{\atype,\btype}(F) = d \mapsto
  \addcost{\btype}(\costof{\atype}(d)+1,F(d)) \typegr{\atype \arrtype \btype} d \mapsto
  \addcost{\btype}(\costof{\atype}(d)+1,G(d)) = \addarg_{\atype,\btype}(G)$ by
  Lemma \ref{lem:builders}(\ref{lem:builders:addcost2}) and because $F(d) \typegr{\btype} G(d)$.
\item $F,G$ are both in $\typeinterpret{\atype \arrtype \btype}$.
  Then we must see that $d \mapsto \addcost{\btype}(1,F(d)) \typegr{\btype} d
  \mapsto \addcost{\btype}(1,G(d))$, so that $\addcost{\btype}(1,F(d)) \typegr{\btype}
  \addcost{\btype}(1,G(d))$ for all $d$.
  This holds by Lemma \ref{lem:builders}(\ref{lem:builders:addcost2}) because $F(d) \typegr{\btype}
  G(d)$ (by definition of $F \typegr{} G$).
\item $F$ is in $\typeinterpret{\atype \arrtype \btype}$ and $G$ is constant.  Then we must see that
  for all $d \in \typeinterpret{\atype}$ we have:
  $\addcost{\btype}(1,F(d)) \typegr{\btype} \addcost{\btype}(\costof{\atype}(d) + 1,G(d))$.
  By monotonicity of $\addcost{\btype}$ (Lemma \ref{lem:builders}(\ref{lem:builders:addcost2}))
  and by Lemma \ref{lem:addcostprops}(\ref{lem:addcostprops:addcosttwice}) it suffices if
  $F(d) \typegr{\btype} \addcost{\btype}(\costof{\atype}(d),G(d))$.

  So consider a fixed $d$.
  By Lemma \ref{lem:estimatebycost}, $d \typegeq{\atype} \addcost{\atype}(\costof{\atype}(d),
      \nul{\atype})$.
  Hence, by monotonicity of $F$ we have
  $F(d)) \typegeq{\btype} F(\addcost{\atype}(\costof{\atype}(d),\nul{\atype})))$.
  By Lemma \ref{lem:addcostincrease}, then
  $F(\addcost{\atype}(\costof{\atype}(d),\nul{\atype}))
      \typegeq{\btype} \addcost{\btype}(\costof{\atype}(d),F(\nul{\atype}))$.
  by Lemma \ref{lem:builders}(\ref{lem:builders:addcost2}).
  By assumption, $F(\nul{\atype}) \typegr{\btype} G(\nul{\atype})$, and since
  $G$ is a constant function, $G(\nul{\atype}) = G(d)$.
  Hence, $F(d) \typegr{\btype} \addcost{\btype}(\costof{\atype}(d),G(d))$.
\item $F$ is a constant function and $G$ is strongly monotonic.  This actually cannot happen!
              To see this, let $m := \costof{\btype}(F(\nul{\atype}))$.
              Note that $F(\nul{\atype}) = F(\addcost{\atype}(m,\nul{\atype}))$ since $F$ is constant,
              $\typegr{\btype} G(\addcost{\atype}(m,\nul{\atype}))$ since $F \typegr{} G$, which
              $\typegeq{\btype} \addcost{\btype}(m,G(\nul{\atype}))$ by Lemma
              \ref{lem:addcostincrease}.  Hence,
              $F(\nul{\atype}) \typegr{\btype} \addcost{\btype}(m,G(\nul{\atype}))$, so by
              Lemma \ref{lem:builders}(\ref{lem:builders:costof}) we have
              $m = \costof{\btype}(F(\nul{\atype})) > \costof{\btype}(\addcost{\btype}(m,
                  G(\nul{\atype}))) = m + \costof{\btype}(G(\nul{\atype})) \geq m$ by
              Lemma \ref{lem:addcostprops}(\ref{lem:addcostprops:costof}).
              This gives the required contradiction.
    \qedhere
\end{itemize}
\end{proof}

In addition, we can formally prove that both $\beta$- and $\eta$-reduction are oriented.

\makeSMbeta*

\begin{proof}
We have either $\ainterpret{\app{(\abs{\avar}{\aterm})}{\bterm}} =
  \addcost{\btype}(\costof{\atype}(\interpret{\bterm}_\alpha) + 1, \interpret{\aterm}_{\alpha[x := \interpret{\bterm}]})$
or $\ainterpret{\app{(\abs{\avar}{\aterm})}{\bterm}} =
  \addcost{\btype}(1, \interpret{\aterm}_{\alpha[x := \interpret{\bterm}]})$.
By Lemma \ref{lem:addcostprops}(\ref{lem:addcostprops:addcostpositive}) we have
  $\ainterpret{\app{(\abs{\avar}{\aterm})}{\bterm}} \typegr{\btype}
  \interpret{\aterm}_{\alpha[x := \interpret{\bterm}]}$ in both cases.
By Lemma \ref{lemma:sbst:lemma}, $\interpret{\aterm}_{\alpha[x := \interpret{\bterm}]} =
  \ainterpret{\aterm[x:=b]}$.
This completes the proof.
\end{proof}

\makeSMeta*

\begin{proof}[Proof]
Since $F \neq \avar$, we have that $d \mapsto \interpret{\app{F}{\avar}}_{\varinterpret[\avar:=d]} =
d \mapsto \varinterpret(F)(d)$, which by extensionality is exactly $\varinterpret(F)$.  Since
$\varinterpret(F)$ is monotonic by assumption on $\varinterpret$, we have
$\ainterpret{\abs{\avar}{F \, \avar}} = \addarg_{\atype,\btype}(d \mapsto
\interpret{\app{F}{\avar}}_{\varinterpret{\avar:=d}}) = \addarg_{\atype,\btype}(\varinterpret(F)) =
\addcost{\atype,\btype}(1,\varinterpret(F))$.
By Lemma \ref{lem:addcostprops}(\ref{lem:addcostprops:addcostpositive}) this
$\typegr{\atype \arrtype \btype} \varinterpret(F) = \interpret{F}$.
\end{proof}

\subsection{Proofs for Section \ref{sec:createmono}}

\iteration*

\begin{proof}
Let $Q$ indicate the function $(x_1,\dots,x_k,y) \mapsto F(x_1,\dots,x_k)^{G(x_1,\dots,x_k)}(y)$.

First, we note that $Q$ indeed maps to $\typeinterpret{\btype \arrtype \btype}$.  So let
$u_1 \in \typeinterpret{\atype_1},\dots,u_k \in \typeinterpret{\atype_k}$.
Since $F(u_1,\dots,u_k) \in \typeinterpret{\btype \arrtype \btype} \subseteq \typeinterpret{\btype}
\arrfunc \typeinterpret{\btype}$, by definition of repeated function application
$F(u_1,\dots,u_k)^{G(u_1,\dots,u_k)} \in \typeinterpret{\btype} \arrfunc \typeinterpret{\btype}$ as
well.  We must show that for all $v,v' \in \typeinterpret{\btype}$, if $v \typegr{\btype} v'$ then
$Q(u_1,\dots,u_k,v_1) = F(u_1,\dots,u_k)^{G(u_1,\dots,u_k)}(v) \typegr{\btype}
F(u_1,\dots,u_k)^{G(u_1,\dots,u_k)}(v') = Q(u_1,\dots,u_k,v')$.
We will show this by induction on the natural number $G(u_1,\dots,u_k)$.
\begin{itemize}
\item If $G(u_1,\dots,u_k) = 0$ then $F(u_1,\dots,u_k)^{G(u_1,\dots,u_k)}(v) = v \typegr{\btype}
  v'$ by assumption, which $= F(u_1,\dots,u_k)^{G(u_1,\dots,u_k)}(v')$.
\item If $G(u_1,\dots,u_k) = n + 1$ then note that, because $F(u_1,\dots,u_k) \in \typeinterpret{
  \btype \arrtype \btype}$ (so this defines a strongly monotonic function), we have
  $F(u_1,\dots,u_k,v) \typegr{\btype} F(u_1,\dots,u_k,v')$.
  Hence,$F(u_1,\dots,u_k)^{G(u_1,\dots,u_k)}(v) = F(u_1,\dots,u_k)^n(F(u_1,\dots,u_k,v))$
  (by definition), $\typegr{\btype} F(u_1,\dots,u_k)^n(F(u_1,\dots,u_k,v'))$ by the induction
  hypothesis.
  This suffices, as this equals $F(u_1,\dots,u_k)^{G(u_1,\dots,u_k)}(v')$.
\end{itemize}

It remains to be shown that $Q$ is weakly monotonic in its first $k$ arguments.  So suppose $u_1'
\in \typeinterpret{\atype_1},\dots,u_k' \in \typeinterpret{\atype_k}$.  We must show that
$Q(u_1,\dots,u_k) \typegr{\btype \arrtype \btype} Q(u_1',\dots,u_k')$.  We will do this by showing
that (**), for all $n,m$ with $n \geq m$ we have $F(u_1,\dots,u_k)^n \typegeq{\btype \arrtype \btype}
F(u_1',\dots,u_k')^m$.
Then $Q(u_1,\dots,u_k) \typegeq{\btype \arrtype \btype} Q(u_1',\dots,u_k')$ follows because
$G(u_1,\dots,u_k) \geq G(u_1',\dots,u_k')$ (by weak monotonicity of $G$).

To prove (**), we use induction on $n$.
\begin{itemize}
\item If $n = 0$, then also $m = 0$.  For all $v \in \typeinterpret{\btype}$ we have
  $F(u_1,\dots,u_k)^n(v) = v = F(u_1',\dots,u_k')^m$.
\item If $n = i + 1 = m$, then let $v \in \typeinterpret{\btype}$; we must show that
  $F(u_1,\dots,u_k)^i(F(u_1,\dots,u_k,v)) \typegeq{\btype} F(u_1',\dots,u_k')^i(
  F(u_1',\dots,u_k',v)$.
  But we know that $F(u_1,\dots,u_k,v) \typegeq{\btype} F(u_1',\dots,u_k',v)$: this holds because
  $F$ is $\weakset(\wmatypes;\typeinterpret{\btype \arrtype \btype})$.  Since we have already seen
  that, for all $i$, $F(u_1,\dots,u_k)^i \in \typeinterpret{\btype \arrtype \btype}$ and is
  therefore also a weakly monotonic function, \\
  $F(u_1,\dots,u_k)^i(F(u_1,\dots,u_k,v))
  \typegeq{\btype} F(u_1,\dots,u_k)^i(F(u_1',\dots,u_k',v))$.  By the induction hypothesis,
  $F(u_1,\dots,u_k)^i \typegeq{\btype \arrtype \btype} F(u_1',\dots,u_k')^i$.  By definition, this
  means that we have
  $F(u_1,\dots,u_k)^i(F(u_1',\dots,u_k',v)) \typegeq{\btype} F(u_1',\dots,u_k')^i(F(u_1',\dots,u_k',
  v)) = F(u_1',\dots,u_k')^m(v)$.
  We complete by transitivity of $\typegeq{\btype}$.
\item If $n = i + 1$ and $i \geq m$, then let $v \in \typeinterpret{\btype}$; we must show
  $F(u_1,\dots,u_k)^i(F(u_1,\dots,u_k,v)) \typegeq{\btype} F(u_1',\dots,u_k')^m(v)$.
  By assumption on $F$ we have $F(u_1,\dots,u_k,v) \typegeq{\btype} v$.
  As we saw before, $F(u_1,\dots,u_k)^i$ is monotonic, so also
  $F(u_1,\dots,u_k)^i(F(u_1,\dots,u_k,v)) \typegeq{\btype} F(u_1,\dots,u_k)^i(v)$.
  By the induction hypothesis, $F(u_1,\dots,u_k)^i(v) \typegeq{\btype} F(u_1',\dots,u_k')^m(v)$.
  \qedhere
\end{itemize}
\end{proof}

\subsection{Proofs for section \ref{sec:bounds}}
We use the following observation to prove the results from this section.
\begin{claim}\label{claim:bounding-finite-pols}
    If $2 \leq x^1,\dots,x^m$, then $\sum\limits_{i=1}^m x^i \leq \prod\limits_{i=1}^m x_i$.
\end{claim}
\begin{claimproof}
    This holds because for $x,y \geq 2$ we have $x + y \leq x * y$ (since $(2+a) + (2+b) = 4 + a + b \leq 4 + 2a + 2b + ab = (2 + a) * (2 + b)$), and by induction on $m$.
\end{claimproof}

\boundingIntData*

\begin{proof}
    \begin{enumerate}
        \item Since the interpretation $\funcinterpret{\sAr{c}} = \pair{P_1, \dots, P_{\typecount{\bsort}}}$ for each constructor $\sAr{c}$ is additive:
        by Definition \ref{def:linear-bol-int},
        for each $\sAr{c} \in \signature$, there exists a constant $a_\sAr{c}$ such that for all $(x^1, \dots, x^m)$,
        $\sum_{l=1}^{\typecount{\bsort}} P_l(x^1, \dots, x^m) \leq a_{\sAr{c}} + \sum_{i=1}^m \sum_{j=1}^{\typecount{\bsort}} x_j^i$.
        Let us set $a$ to be the maximum of such $a_{\sAr{c}}$, so for the sum of components $P_l$ of $\funcinterpret{\sAr{c}}$ we have:
        \begin{equation}\label{eq:bound-components-constructor-int}
            \sum_{l = 1}^{\typecount{\bsort}} P_l(x^1, \dots, x^m) \leq a + \sum_{i=1}^m \sum_{j=1}^{\typecount{\asort_i}} x_j^i.
        \end{equation}
        We prove by induction on the size of $s :: \bsort$ that $\sum_{l=1}^{\typecount{\bsort}} \interpret{s}_l \leq a * |s|$.
        Then certainly $\interpret{s}_l \leq a * |s|$ holds for any component $\interpret{s}_l$, and the first part of the lemma holds.

        For the base case, $|s| = 1$, $s$ is a constant $\sAr{c}$ and $\sum_{l=1}^{\typecount{\asort}} \interpret{\sAr{c}}_l \leq a_\sAr{c} \leq a$,
        by assumption (\ref{eq:bound-components-constructor-int})

        Let $|s| > 1$; then $s = \sAr{c}(d_1, \dots, d_m)$ and using (\ref{eq:bound-components-constructor-int}) above,
        we can expand the sum, as follows:
        \begin{align*}
           \sum_{l=1}^{\typecount{\bsort}} \interpret{\sAr{c}(d_1, \dots, d_m)}_l &= \sum_{l=1}^{\typecount{\bsort}} P_l(\interpret{d_1}, \dots, \interpret{d_m})\\
           &\stackrel{(\ref{eq:bound-components-constructor-int})}{\leq}
              a + \sum_{i=1}^m \sum_{j=1}^{\typecount{\asort_i}} \interpret{d_i}_j \\
           &\stackrel{(IH)}{\leq} a + \sum_{i=1}^m a * |d_i|\\
           &= a * \left(1 + \sum_{i=1}^m |d_i| \right)\\
           &= a * |s|.
        \end{align*}

        Hence, we are done choosing $b := a$.

        \item The proof follows the same structure as before:
        by Definition \ref{def:linear-bol-int},
        each $\funcinterpret{\sAr{c}} = \pair{P_1, \dots, P_{\typecount{\bsort}}}$ is now linearly bounded;
        that is, for each $\sAr{c} \in \signature$, there exists a constant $a_{\sAr{c}}$ such that for all $(x^1, \dots, x^m)$
        we have
        $P_l(x^1,\dots,x^m) \leq a_{\sAr{c}} * (1 + \sum_{i=1}^m\sum_{j=1}^{\typecount{\asort_i}} x^i_j)$.
        Let us set, as before, $a$ to be the maximum of such $a_{\sAr{c}}$ and define $k = \max(2,\max_i \typecount{\asort_i})$.
        Notice that $k$ is determined when we define the interpretation's domain, so it does not depend on the size of $s$.

        We prove by induction on the size of $s$ that $\interpret{s}_l \leq 2^{(a * k) * |s| }$,
        for each component $P_l$ of $\interpret{s}$.
        In the base case, where $s$ is a constant constructor,
        we have that
        $\interpret{\sAr{c}}_l \leq a_{\sAr{c}} \leq a * k < 2^{a * k}$ follows trivially.
        For the inductive step, we have $s = \sAr{c}(d_1, \dots, d_m)$.
        Then:
        \begin{align*}
            P_l(\interpret{d_1}, \dots, \interpret{d_m}) &
            \leq a_{\sAr{c}} * (1 + \sum\limits_{i = 1}^m\sum\limits_{j = 1}^{\typecount{\asort_i}} \interpret{d_i}_j) \\
            & \leq a * (2 * \sum\limits_{i = 1}^m\sum\limits_{j = 1}^{\typecount{\asort_i}} \interpret{d_i}_j)\ \text{because}\ 1 + z \leq 2z\ \text{for}\ z \geq 1 \\
            & \stackrel{(IH)}{\leq} 2 * a * \sum\limits_{i = 1}^m \left( \sum\limits_{j = 1}^{\typecount{\asort_i}} 2^{a*k*|d_i|} \right) \\
            & \leq 2 * a * k * \sum\limits_{i = 1}^m 2^{a*k*|d_i|} \\
            & \leq (2 * a * k) * \prod\limits_{i = 1}^m 2^{a*k*|d_i|}\ \text{by claim (\ref{claim:bounding-finite-pols})} \\
            & \leq 2^{a * k} * \prod\limits_{i = 1}^m \left(2^{a * k * |d_i| } \right)\ \text{because}\ 2z \leq 2^z\ \text{if}\ z \geq 2\\
            & = 2^{(a * k)} * 2^{(a * k) * \sum\limits_{i = 1}^m |d_i| }                                                              \\
            & = 2^{a * k \left( 1 + \sum\limits_{i = 1}^m |d_i|\right) }                                                      \\
            & = 2^{(a * k) * |s| }.
        \end{align*}

        Hence, we are done choosing $b := a * k$.
    \end{enumerate}
\end{proof}

The proofs for items in Corollary \ref{col:bound-basic-func-pol} follow from the same strategy used above and utilize the bounds established for data terms.

\end{document}